\documentclass[aps,prd,onecolumn,superscriptaddress,nofootinbib,showkeys,a4paper,11pt]{revtex4}

\usepackage{amsmath, amssymb, amsfonts, amsthm, bm, bbm}
\usepackage{enumitem}
\usepackage{multirow}
\usepackage{booktabs}
\usepackage{graphicx}
\usepackage{subfigure}  

\usepackage{tikz}
\usepackage{tikz-3dplot}

\newtheorem{thm}{Theorem}[section]
\newtheorem{cor}[thm]{Corollary}
\newtheorem{prop}[thm]{Proposition}
\newtheorem{rem}[thm]{Remark}

\newtheorem{lemma}[thm]{Lemma}
\newcommand{\be}{\begin{equation}} \newcommand{\ee}{\end{equation}}
\newcommand{\bea}{\begin{eqnarray}} \newcommand{\eea}{\end{eqnarray}}

\usepackage{pdflscape}
\usepackage{tabularx}

\usepackage{algorithm}
\usepackage{algpseudocode}

\usepackage[colorlinks=true,linkcolor=blue,citecolor=blue,urlcolor=blue]{hyperref}

\usepackage{comment}
\makeatletter
\renewcommand\subparagraph{\@startsection{subparagraph}{4}{\z@}%
  {-3.25ex \@plus -1ex \@minus -.2ex}%
  {1.5ex \@plus .2ex}%
  {\normalfont\normalsize\bfseries}}
\makeatother

\begin{document}

\title{Asymptotic Theorems and Averaging in Scalar Field Cosmology}

\author{Genly Leon}
\email{genly.leon@ucn.cl}
\affiliation{Departamento de Matemáticas, Universidad Católica del Norte, Avenida Angamos 0610, Casilla 1280, Antofagasta, Chile}
\affiliation{Institute of Systems Science, Durban University of Technology, Durban 4000,
South Africa}
\affiliation{Centre for Space Research, North-West University, Potchefstroom 2520, South Africa}

\author{Aleksander Kozak}
\email{aleksander.kozak@ucn.cl} 
\affiliation{Departamento de Matemáticas, Universidad Católica del Norte, Avenida Angamos 0610, Casilla 1280, Antofagasta, Chile}

\author{Claudio Michea}
\email{crm035@alumnos.ucn.cl}
\affiliation{Departamento de Física, Universidad Católica del Norte, Avda. Angamos 0610, Casilla 1280, Antofagasta, Chile}
\begin{abstract}
We present a hybrid study that combines a concise review of scalar‑field cosmology with new analytic developments that integrate averaging reductions for oscillatory regimes with dynamical‑systems techniques. For oscillatory fields, we derive an averaging reduction that yields an effective slow system whose time averages control dissipation; introducing uniform derivative bounds, Barbalat/LaSalle arguments, and a finite‑dimensional center/stable manifold reduction, we carry out late‑time analysis of the models. We prove persistence of equilibria, decay estimates, and local invariant manifolds under small $C^k$ perturbations of $\chi(\phi)$ and $G(a)$, quantify how averaged dissipation lifts to the full oscillatory dynamics with an $\mathcal{O}(H)$ error, and provide numerical examples. In addition to asymptotic reductions, we obtain exact quadrature solutions in general relativistic, anisotropic, and brane‑world settings, yielding closed‑form expressions for $t(a)$, $\phi(a)$, and $H(a)$ and enabling analytic computation of inflationary observables. 
\end{abstract}

\keywords{scalar field cosmology; anisotropy; analytic integration; averaging; persistence}

\maketitle

\section{Introduction}
\label{Sect:1}

Scalar fields play an important role in theoretical cosmology, as they feature in models of inflationary expansion, late-time acceleration, and departures from Einsteinian gravity. Minimally coupled models—where the scalar interacts with geometry only through the standard kinetic and potential terms—were present in early-universe scenarios, allowing slow-roll inflation, as discussed by Guth~\cite{Guth:1980zm}. In gravitational theory, the addition of scalar fields led to models such as Jordan's Kaluza–Klein extension~\cite{Jordan:1958zz}, Brans–Dicke theory~\cite{Brans:1961sx}, and Horndeski's tensor–scalar formulation~\cite{Horndeski:1974wa}, extending Einstein's general relativity.

Within this broader landscape, Ibáñez et al.~\cite{Ibanez:1995zs}, Coley~\cite{Coley:1997nk, Coley:1999mj}, and collaborators~\cite{Coley:2000zw, Coley:2000yc} investigated generalized scalar field models in connection with modified gravity, extended quintessence, and Galileon theories. Their work laid the foundation for anisotropic generalizations and the classification of dynamical attractors.  

Dynamical systems methods have proven effective for analyzing scalar cosmologies. Foster~\cite{Foster:1998sk} studied nonnegative potentials in flat FLRW backgrounds and identified regular asymptotic behavior for massless fields. Miritzis~\cite{Miritzis:2003ym} extended this to include perfect fluids and negative curvature, showing that potentials with vanishing minima lead to ever-expanding universes with decreasing energy densities.  

Nonminimal coupling models, in particular those introduced by Morales, Nápoles, and Leon~\cite{Dania&Yunelsy, Leon:2008de}, incorporated general potentials and coupling functions $V(\phi)$ and $\chi(\phi)$. These models showed that, despite divergent early-time behavior, scalar-driven de Sitter expansion can still emerge. Giambo~\cite{Giambo:2009byn} demonstrated that such models, including  $f(R)$ gravity, admit stable attractors and analytic decoupling near the singularity limit $\phi \rightarrow \infty$. Scalar–curvature couplings of the form $F(\phi)R$, explored by Shahalam~\cite{Shahalam:2019jgs}, led to viable models with functions like $F(\phi) = 1 - \xi \phi^2$ and potentials $V(\phi) = V_0(1 + \phi^p)^2$. These were further developed by Nojiri and Odintsov~\cite{Nojiri:2019riz}, Humieja et al.~\cite{Humieja:2019ywy}, Matsumoto~\cite{Matsumoto:2017gnx, Matsumoto:2015hua}, and Solomon~\cite{Solomon:2015hja}, each contributing mechanisms for late-time acceleration and scalar stability.  

In homogeneous FLRW spacetimes, Giambo~\cite{Giambo:2008ck} studied scalar cosmologies with divergent self-interactions, including exponential and polynomial growth profiles. He showed that under general conditions, nonnegative minima lead to stable equilibria and scalar dominance when $\gamma > 1$~\cite{Giambo:2009byn}. Leon and collaborators~\cite{LeonTorres:2010cdf, Leon:2012ccj, Fadragas:2014mra} performed detailed phase-space analyses of flat FLRW models for scalar–tensor theories. Their results showed that for smooth potentials with zero minima and controlled growth of the coupling, all energy components—matter, radiation, and the scalar kinetic term—decay as $t\rightarrow +\infty$, leading to de Sitter asymptotics. These findings generalize earlier results by Miritzis~\cite{Miritzis:2003ym} and Leon~\cite{Leon:2012ccj}.  

Recent analysis of the scalar-tensor theories by Leon and Silva~\cite{Leon:2020ovw, Leon:2020pfy, Leon:2019iwj} led to the development of a dynamical systems framework capable of handling non-canonical couplings, curved FLRW backgrounds, and mixed matter sectors. Their global phase-space construction~\cite{Leon:2020ovw} classifies invariant sets, equilibrium manifolds, and asymptotic regimes for broad families of potentials and couplings. This was extended in~\cite{Leon:2020pfy} with rigorous theorems on late-time attractors, showing scalar dominance under general conditions. Leon, González, Millano, and Silva~\cite{Leon:2020pvt} studied perturbative deformations of this framework, analyzing how scalar–matter interactions shift fixed-point structures and affect energy exchange. 

Another technique useful for analyzing late-time cosmic evolution in modified gravity scenarios is averaging over oscillatory scalar fields. These methods have recently been reviewed by Leon and Michea~\cite{Leon:2026vov}. Their study provides a qualitative framework for reducing fast oscillatory dynamics to effective slow systems. In particular, they establish conditions under which averaged trajectories approximate the full dynamics with controlled error, and they connect these reductions to persistence theory and invariant manifold analysis. This work provides a complementary foundation for the present study, in which averaged results are combined with dissipative estimates, exact quadrature solutions, and center‑manifold reductions.  

Building on prior results in the literature, we develop a unified analytic treatment of scalar cosmologies using quadrature techniques. By re-parametrizing the evolution in terms of the scale factor $a$, we can express the scalar field and cosmic time as functions of $a$, yielding exact solutions in some cases.  Historically, Chimento and Jakubi~\cite{Chimento:1995da} derived exact solutions for scalar fields coupled to fluids using scale-factor parametrization, while Coley~\cite{Coley:1999uh} classified attractor behavior across potentials and geometries. González and Quiros~\cite{Gonzalez:2007ht} developed dark-sector models incorporating phantom and quintessence regimes, and Copeland et al.~\cite{Copeland:1993jj} pioneered inflationary reconstruction from data, later refined by Lidsey et al.~\cite{Lidsey:1995np}.

The paper is organized as follows: in Section~\ref{sect:0}, we introduce the constrained dynamical system for a scalar field coupled to matter and geometry.  In the same part, we also present lemmas and theorems central to the further mathematical analysis of the cosmological scenarios.
We then discuss the technical aspects of the analysis in Section~\ref{sec:decay_center_perturb}. Section~\ref{sec:decay} derives refined decay estimates, including integrability of the dissipative quantity, uniform derivative bounds, pointwise $O(1/t)$ decay, and a bootstrap to exponential convergence near nondegenerate minima. In Section~\ref{sec:center} we construct the local invariant manifold $\mathcal{M}_{\mathrm{loc}}$, perform finite‑dimensional reduction on the Friedmann constraint, and record its regularity and attraction properties. In Section~\ref{sec:perturbations} we establish persistence of equilibria, decay estimates, and local manifolds under small $C^k$ perturbations of the coupling $\chi(\phi)$ and geometric term $G(a)$, with continuous dependence of constants and spectral gaps on perturbation size.  

Section~\ref{sec:discussion_averaging} summarizes the averaging framework, presents uniform derivative bounds, and explains how averaging errors enter the global analysis. It includes the Barbalat and LaSalle arguments, spectral gap considerations, and combined corollaries, together with a practical recipe for verifying hypotheses and a discussion of limitations. Section~\ref{sec:quad} introduces the quadratic potential.   

Section~\ref{Asect:2} develops the analytic framework for minimally coupled scalar fields in FLRW and Bianchi~I geometries. Reformulation in $\tau$‑time enables us to obtain quadrature expressions for $\phi(a)$ and $t(a)$ and closed‑form inflationary observables. Subsequent subsections classify FLRW solutions, analyze asymptotics, and extend the formalism to anisotropic Bianchi~I backgrounds, including tabulated attractor regimes and inflationary diagnostics.  

Section~\ref{Asect:3} examines non-minimal coupling of the scalar field to the matter sector, modifying energy exchange and inflationary dynamics in both FLRW and Bianchi~I geometries. Section~\ref{Asect:4} focuses on scalar field in brane‑world scenarios, incorporating geometric corrections and modified Friedmann equations. There, we derive the early- and late-time behaviors, introduce analytic source terms, and generalize to anisotropic brane cosmologies. The paper concludes with conclusions and possible directions for future research (Section~\ref{conclusions}).

\section{Scalar Field Dynamics with Matter and Geometry}
\label{sect:0}
In this section, we formulate the basic equations governing the evolution of the universe in the presence of a scalar field interacting with matter and geometry. Following previous studies~\cite{Fajman:2020yjb,Fajman:2021cli,Leon:2020pfy,Leon:2020ovw,Leon:2021lct,Leon:2021rcx,Leon:2021hxc,Leon:2020pvt}, we introduce a scalar field $\phi$ evolving under a potential $V(\phi)$ and exchanging energy with matter through a coupling function $\chi(\phi)$. Geometric effects are incorporated via a term $G(a)$, dependent on the scale factor $a(t)$, allowing for flat, curved, or anisotropic spatial configurations. The equations governing the cosmic evolution are the following:

\begin{subequations}\label{friedmannsystem}
\begin{align}
\dot{H} &= -\tfrac{1}{2}(\gamma \rho_m + \dot{\phi}^2) + \tfrac{1}{6} a G'(a),\\
\dot{a} &= a H, \\
\dot{\rho}_m &= -3\gamma H \rho_m - Q(\phi, \dot{\phi}, \rho_m), \\
\dot{\phi} \ddot{\phi} &= -3 H \dot{\phi}^2 - \dot{\phi} V'(\phi) + Q(\phi, \dot{\phi}, \rho_m),
\end{align}
\end{subequations}

where $ a(t) $ is the scale factor and $ H = \dot{a}/a $ the Hubble parameter. 

For simplicity, we consider an interaction term
\begin{equation}
Q(\phi, \dot{\phi}, \rho_m) := \tfrac{1}{2}(4 - 3\gamma) \rho_m \dot{\phi} \frac{d \ln \chi(\phi)}{d\phi}
\end{equation}
which governs energy exchange between matter and the scalar field, consistent with scalar-tensor theories~\cite{Kaloper:1997sh,Gonzalez:2007ht}.

The full dynamical system is expressed in terms of the Hubble rate, scale factor, matter density, and scalar field, defining a five-dimensional phase space. These equations, however, are subject to the Friedmann constraint. 

As will be shown in the following sections of the paper, under suitable conditions, the system approaches a stable regime in which the scalar field becomes stationary, matter vanishes, and the Hubble rate converges to a constant. These results characterize the long-term behavior of scalar field cosmologies and provide a foundation for the analytic techniques developed in subsequent sections.

\subsection{System, assumptions and main result}
We consider the first-order system obtained from the field equations after setting $y=\dot\phi$
\begin{equation}\label{system}
\begin{aligned}
\dot H &= -\tfrac12(\gamma\rho_m+y^2)+\tfrac16\,aG'(a),\\
\dot\rho_m &= -3\gamma H\rho_m-\tfrac12(4-3\gamma)\rho_m\,y\,\frac{d\ln\chi(\phi)}{d\phi},\\
\dot a &= aH,\\
\dot y &= -3Hy - V'(\phi) + \tfrac12(4-3\gamma)\rho_m\,\frac{d\ln\chi(\phi)}{d\phi},\\
\dot\phi &= y.
\end{aligned}
\end{equation}
evolving on the Friedmann constraint surface
\begin{equation}\label{constraint}
\Psi=\big\{(H,\rho_m,a,y,\phi)\in\mathbb{R}^5:\;3H^2=\rho_m+\tfrac12 y^2+V(\phi)+\Lambda+G(a)\big\}.
\end{equation}
Assuming the geometric ansatz
\begin{equation}\label{Gdef}
G(a)=\kappa^2 a^{-p},\quad p>0 \implies aG'(a)=-pG(a)=-p\kappa^2 a^{-p}.
\end{equation}
The geometric contribution $G(a)$ in the Friedmann constraint
\eqref{constraint} may be interpreted as an effective perfect fluid with
energy density $\rho_{G}(a) := G(a) = \kappa^{2} a^{-p}$.

For a barotropic fluid with constant equation-of-state parameter $w_{G}$, the energy density scales as
$
\rho_{G}(a) \propto a^{-3(1+w_{G})}.
$
Comparing this with the ansatz \eqref{Gdef}, we obtain the correspondence
\begin{equation}
p = 3(1+w_{G})
\implies
w_{G} = \frac{p}{3} - 1.
\end{equation}
Thus, $G(a)$ behaves dynamically as a perfect fluid whose equation of state
is determined entirely by the exponent $p$.  Several physically relevant
cases are
\begin{itemize}
    \item $p=1 \;\implies\; w_{G}=-\tfrac{2}{3}$: a domain-wall--like
    fluid;
    \item $p=2 \;\implies\; w_{G}=-\tfrac{1}{3}$: a cosmic-string--like
    fluid;
    \item $p=3 \;\implies\; w_{G}=0$: a dust-like (pressureless matter)
    component;
    \item $p=4 \;\implies\; w_{G}=\tfrac{1}{3}$: a radiation-like
    \item $p=6 \;\implies\; w_{G}=1$: a stiff-like fluid.
\end{itemize}

\begin{rem}\label{remII}
Assume that $V(\phi) \geq 0$ is a function of class $C^{2}(\mathbb{R})$ with a local minimum at 
$\phi = 0$ and $V(0) = 0$. 

If $\Lambda = 0$, then the configuration 
\begin{equation}
(0,0,a^{*},0,0), \qquad a^{*} \to +\infty,
\end{equation}
is an equilibrium point for the flow of \eqref{system}, since the Friedmann constraint
\begin{equation}\label{eq:constraint_equilibrium}
3H^{2} = \rho_{m} + \tfrac{1}{2}y^{2} + V(\phi) + \Lambda + G(a)
\end{equation}
reduces to
\begin{equation}\label{eq:limit_condition}
G(a) = 0,
\end{equation}
which holds in the limit $a^{*}\to +\infty$ because 
\begin{equation}\label{eq:Glimit}
G(a) = \kappa^{2}a^{-p} \;\rightarrow\; 0 \quad \text{for } p>0.
\end{equation}
The set
\begin{equation}\label{eq:invariant_set}
\{(H,\rho_{m},a,\gamma,\phi) \in \Psi : H = 0\}
\end{equation}
is invariant under the flow of \eqref{system}, and $H$ does not change sign. Conversely, if there exists an orbit 
with 
\begin{equation}\label{eq:Hcross}
H(0) > 0 \quad \text{and} \quad H(t_{1}) < 0 \quad \text{for some } t_{1} > 0,
\end{equation}
then the solution must pass through the origin, violating the existence and uniqueness of solutions for a $C^{1}$ flow. 

For $\Lambda > 0$, however, the constraint \eqref{eq:constraint_equilibrium} enforces
\begin{equation}\label{eq:deSitter}
H^{2} = \tfrac{\Lambda}{3} + \tfrac{1}{3}\big(\rho_{m} + \tfrac{1}{2}y^{2} + V(\phi) + G(a)\big) > 0,
\end{equation}
so the hypersurface $H=0$ is excluded from the admissible phase space. In this case, $H=0$ cannot be crossed, 
and the late-time attractor is de Sitter with
\begin{equation}\label{eq:HdeSitter}
H \rightarrow \sqrt{\tfrac{\Lambda}{3}}
\end{equation}
rather than the static equilibrium above.
\end{rem}

We now formulate the hypotheses that will be used by us to prove crucial theorems related to the description of late-time evolution of the system \eqref{system}:

\noindent\textbf{Hypotheses.}
\begin{enumerate}[label=(H\arabic*)]
\item \label{H1} $V\in C^2(\mathbb{R})$, $V(\phi)\ge0$ for all $\phi$, and $V(\phi)=0$ iff $\phi=0$.
\item \label{H2} $V'(\phi)$ is bounded on any set where $V(\phi)$ is bounded.
\item \label{H3} $\Lambda\ge0$ and $G(a)=\kappa^2 a^{-p}$ with $p>0$.
\item \label{H4} $\chi\in C^2(\mathbb{R})$ and there exists $K_1>0$ such that $\big|\chi'(\phi)/\chi(\phi)\big|\le K_1$ on bounded $\phi$-sets.
\item \label{H5} $\gamma\in[1,2]$.
\item \label{H6} Initial data satisfy $H(t_0)>0$.
\end{enumerate}
Additional alternatives used below:
\begin{itemize}
\item \emph{Case A (symmetric potential):} $V'(\phi)<0$ for $\phi<0$ and $V'(\phi)>0$ for $\phi>0$.
\item \emph{Case B (runaway potential):} $V(\phi)\ge0$, $V'(\phi)<0$ for all $\phi$, and $\lim_{\phi\to-\infty}V(\phi)=+\infty$.
\end{itemize}

We have the \textbf{auxiliary lemma}
\begin{lemma}[Barbalat]\label{barbalat}
If $f:[t_0,\infty)\to\mathbb{R}$ is uniformly continuous and $f\in L^1([t_0,\infty))$, then $\lim_{t\to\infty}f(t)=0$.
\end{lemma}

Lemma \ref{barbalat} is used in the proof of the \textbf{main theorem}:
\begin{thm}[Asymptotic Behavior {\cite{Leon:2019iwj}}]
\label{UnifiedTheorem}
Under hypotheses \ref{H1}--\ref{H6}, every forward solution of \eqref{system} with $H(t_0)>0$ satisfies
\begin{equation}
\label{first-claim}
\lim_{t\to\infty}\big(\rho_m(t),\,y(t),\,\kappa^2 a(t)^{-p}\big)=(0,0,0).
\end{equation}
Moreover, under Case A, the scalar field has the asymptotic alternatives
\begin{equation}
\label{second-claim}
\lim_{t\to\infty}\phi(t)\in\{-\infty,0,+\infty\},
\end{equation}
and if additionally
\begin{equation}
3H(t_0)^2<\min\{\lim_{\phi\to -\infty}V(\phi), \lim_{\phi\to +\infty}V(\phi)\}
\end{equation}
then
\begin{equation}
    \label{third-claim}
\phi(t)\to0, \quad H(t)\to\sqrt{\Lambda/3}.
\end{equation}
Under Case B, one has
\begin{equation}
    \label{fourth-claim}
    \phi(t)\to+\infty,
\end{equation}
and if $\lim_{\phi\to\infty}V(\phi)=0$ then
\begin{equation}
    \label{fifth-claim}
    H(t)\to\sqrt{\Lambda/3}.
\end{equation}
\end{thm}
\begin{proof}
Define
\begin{equation}\label{fdef}
f(t):=\tfrac{p}{6}\kappa^2 a(t)^{-p}+\tfrac12\gamma\rho_m(t)+\tfrac12 y(t)^2\ge0.
\end{equation}
Using $G(a)=\kappa^2 a^{-p}$ and $aG'(a)=-p\kappa^2 a^{-p}$ we obtain from \eqref{system}
\begin{equation}
\dot H=-\tfrac12(\gamma\rho_m+y^2)-\tfrac{p}{6}\kappa^2 a^{-p}=-f(t).
\end{equation}
Since $H(t_0)>0$ and $\dot H\le0$, $H(t)$ is non-increasing, and, according to Remark \ref{remII}, cannot cross $H = 0$. Then it converges to some $\eta\ge0$. Hence
\begin{equation}
\int_{t_0}^\infty f(t)\,dt=H(t_0)-\eta<\infty \implies f\in L^1([t_0,\infty)).
\end{equation}
The Friedmann constraint \eqref{constraint}, together with $V\ge0$ and $\Lambda\ge0$, implies that along the forward orbit the quantities $\rho_m(t)$, $y(t)$ and $V(\phi(t))$ remain bounded. By \ref{H2} $V'(\phi)$ is bounded on the relevant $\phi$-range, and by \ref{H4} $\chi'(\phi)/\chi(\phi)$ is uniformly bounded there. Using the evolution equations in \eqref{system} we obtain uniform bounds on $\dot\rho_m(t)$, $\dot y(t)$ and $\dot a(t)$ for $t\ge t_0$. Differentiating \eqref{fdef} and employing these bounds yields a uniform bound on $\dot f(t)$ on $[t_0,\infty)$; therefore, $f$ is uniformly continuous.

By Barbalat's lemma (Lemma~\ref{barbalat}), integrability together with uniform continuity implies $f(t)\to0$ as $t\to\infty$. From \eqref{fdef} we deduce
\begin{equation}
\rho_m(t)\to0,\quad y(t)\to0,\quad \kappa^2 a(t)^{-p}\to0,
\end{equation}
which proves \eqref{first-claim}.

From the constraint \eqref{constraint} and the limits above,
\begin{equation}
V(\phi(t))=3H(t)^2-\Lambda-\rho_m-\tfrac12 y^2-G(a)\rightarrow 3\eta^2-\Lambda.
\end{equation}

\emph{Case A.} The monotonicity assumption on $V$ \ref{H5} implies $V$ is strictly monotone on each side of zero. Hence, any finite accumulation point of $\phi(t)$ must be a critical point of $V$. If $V(\phi(t))\to0$ then $\phi(t)\to0$. Suppose instead $\phi(t)\to\bar\phi\neq0$ finite; since $\rho_m(t)\to0$ and $y(t)\to0$ the interaction term tends to zero, so from \eqref{system} we obtain $\dot y(t)\to -V'(\bar\phi)\neq0$, contradicting $y(t)\to0$. Therefore $\phi(t)\to\pm\infty$ or $0$, proving \eqref{second-claim}. The additional energy bound $3H(t_0)^2<\min\{\lim_{\phi\to\pm\infty}V(\phi)\}$ rules out the runaway limits, leaving $\phi(t)\to0$ and hence $H(t)\to\sqrt{\Lambda/3}$, proving \eqref{third-claim}.

\emph{Case B.} Under \ref{H6} $V$ is strictly decreasing and $\lim_{\phi\to-\infty}V(\phi)=+\infty$; this excludes finite accumulation points of $\phi(t)$, so $\phi(t)\to+\infty$, proving \eqref{fourth-claim}. If additionally $\lim_{\phi\to\infty}V(\phi)=0$, continuity of $V$ yields $V(\phi(t))\to0$ and therefore $H(t)\to\sqrt{\Lambda/3}$, proving \eqref{fifth-claim}.
This completes the proof of Theorem \ref{UnifiedTheorem}.
\end{proof}

These results justify our initial assumption about the asymptotic behaviour of the dynamical variables: expanding solutions with nonnegative potential energy are driven by the dissipative term in $\dot H$ toward regimes in which matter and scalar kinetic energy are negligible, and—when the potential admits an appropriate minimum—toward the corresponding equilibrium state.

Recall that the constraint~\eqref{constraint} implies
\begin{equation}
    3H^{2}
    = \rho_{m} + \tfrac12 y^{2} + V(\phi) + \Lambda + G(a).
\end{equation}
For simplicity, we absorb the cosmological constant into the potential.  
Throughout the remainder of this section, we therefore use the replacement
\begin{equation}
    V(\phi) + \Lambda \;\mapsto\; V(\phi).
\end{equation}
Hence, when we set
\begin{equation}
    H_{*} = \sqrt{\frac{V(\phi_{*})}{3}},
\end{equation}
this corresponds, in the original notation, to
\begin{equation}
    H_{*} = \sqrt{\frac{V(\phi_{*}) + \Lambda}{3}}.
\end{equation} We always assume $\Lambda\geq 0$.

\subsection{Refined decay estimates and local reduction}

We now state refined decay estimates and the local reduction that will be proved in the next sections. These results quantify the dissipation rates and describe how the scalar field relaxes to a minimum of $V$.

\begin{prop}\label{prop:decay}
Under the hypotheses of Theorem~\ref{UnifiedTheorem}, let $x(t)=(H,\rho_m,a,y,\phi)(t)$ be a forward solution with $H(t_0)>0$. There exist constants $C>0$ and $t_1\ge t_0$ such that for all $t\ge t_1$
\begin{equation}
\rho_m(t) \le C(1+t)^{-1}, \quad y(t)^2 \le C(1+t)^{-1}, \quad \kappa^2 a(t)^{-p} \le C(1+t)^{-1}.
\end{equation}
If in addition $V$ has a nondegenerate minimum $\phi_*$ with $V''(\phi_*)>0$, then there exist $C',\lambda>0$ and $t_2\ge t_1$ such that for all $t\ge t_2$
\begin{equation}
|\phi(t)-\phi_*| + |y(t)| \le C' e^{-\lambda t}, \quad |H(t)-\sqrt{V(\phi_*)/3}|\le C' e^{-\lambda t}.
\end{equation}
\end{prop}
\begin{rem}
The polynomial decay in the first part follows from the integrability of the dissipative term $\dot {H} $ together with the uniform‑continuity argument used in the proof of Theorem~\ref{UnifiedTheorem}. Exponential relaxation to the minimum requires the nondegeneracy $V''(\phi_*)>0$; it is obtained by linearizing about the equilibrium and controlling the nonlinear remainder via a standard bootstrap argument.
\end{rem}

\begin{lemma}\label{lemma:center-manifold}
Under the hypotheses of Proposition~\ref{prop:decay}, assume $V$ has a nondegenerate minimum $\phi_*$. There exists a local invariant manifold $\mathcal{M}_{\mathrm{loc}}$ of class $C^k$ (for any finite $k$ determined by the regularity of $V$ and $\chi$) containing the equilibrium $\mathbf p_*$. Solutions with initial data in $\mathcal{M}_{\mathrm{loc}}$ converge exponentially to $\mathbf p_*$.
\end{lemma}

\begin{proof}[Proof strategy]
The argument has three steps.

\textbf{(i) Localization and linearization.} Restrict the flow to a small, compact, positively invariant neighborhood $\Psi_+$ of $\mathbf p_*$ (as in the proof of Theorem~\ref{thm:combined}). Linearize \eqref{system} at $\mathbf p_*$ and verify the linearization has no spectrum with positive real part; the nondegeneracy $V''(\phi_*)>0$ produces the required spectral gap between neutral and stable directions.

\textbf{(ii) Center/stable manifold theorem.} Apply the finite‑dimensional center (stable) manifold theorem (see \cite{Carr1981}) to obtain a local invariant manifold $\mathcal{M}_{\mathrm{loc}}$ tangent to the stable subspace at $\mathbf p_*$. Smoothness of $V$ and $\chi$ yields the stated regularity of $\mathcal{M}_{\mathrm{loc}}$.

\textbf{(iii) Nonlinear stability and exponential decay.} On $\mathcal{M}_{\mathrm{loc}}$ the reduced dynamics are exponentially contracting. Using the spectral gap and Gronwall's inequality to control the nonlinear remainder yields the exponential bounds in Proposition~\ref{prop:decay}.
\end{proof}

The following parts of the paper are organized as follows: Section~\ref{sec:decay} contains the detailed proof of Proposition~\ref{prop:decay}: the polynomial decay estimates are derived from the integrability of $f(t)$ and the uniform continuity argument, and the exponential regime is established under nondegeneracy. Section~\ref{sec:center} constructs the local invariant manifold $\mathcal{M}_{\mathrm{loc}}$ and carries out the finite‑dimensional reduction near $\mathbf p_*$. Section~\ref{sec:perturbations} treats persistence of the attractor under small nonminimal couplings $\chi(\phi)\neq1$ and small anisotropic perturbations of $G(a)$.

\section{Decay estimates, local reduction, and persistence}
\label{sec:decay_center_perturb}

In this section we provide the full statements and proofs of the three technical components announced earlier: (i) refined decay estimates (Section~\ref{sec:decay}), (ii) construction of a local invariant manifold and finite‑dimensional reduction (Section~\ref{sec:center}), and (iii) persistence of the attractor and decay properties under small perturbations of the coupling $\chi$ and the geometric term $G$ (Section~\ref{sec:perturbations}). All results use the notation and hypotheses introduced in Section~2 and Hypotheses \ref{H1}--\ref{H6} used in Theorem~\ref{UnifiedTheorem}.
\subsection{Decay estimates}
\label{sec:decay}

We begin by proving Proposition~\ref{prop:decay} in full. The key steps are (a) a short lemma giving uniform derivative bounds on forward orbits, (b) application of Barbalat's lemma to deduce pointwise decay from integrability, and (c) a bootstrap linearization argument to obtain exponential decay near a nondegenerate minimum.

\begin{lemma}[Uniform derivative bounds]\label{lemma:derivative-bounds}
Let $x(t)=(H,\rho_m,a,y,\phi)(t)$ be a forward solution of \eqref{system} with $H(t_0)>0$. Under Hypotheses \ref{H1}--\ref{H4}, the following hold on any forward orbit contained in a set where $V(\phi)$ is bounded:
\begin{enumerate}
\item The quantities $H(t),\rho_m(t),y(t),\phi(t),a(t)$ remain bounded for $t\ge t_0$.
\item There exist constants $M_1,M_2>0$ such that for all $t\ge t_0$
\begin{equation}
|\dot\rho_m(t)|\le M_1,\quad |\dot y(t)|\le M_1,\quad |\dot a(t)|\le M_2.
\end{equation}
\end{enumerate}
\end{lemma}
\begin{proof}
From the Friedmann constraint \eqref{constraint} and $V(\phi)\ge0$ we have (recall we absorbed $\Lambda\geq 0$ in the potential)
\begin{equation}
\rho_m+\tfrac12 y^2+V(\phi)+G(a)=3H^2,
\end{equation}
so if $H$ is bounded on the forward orbit then $\rho_m,y,V(\phi),G(a)$ are bounded. Boundedness of $V(\phi)$ implies, by Hypothesis \ref{H2}, that $V'(\phi)$ is bounded on the relevant $\phi$-range. By Hypothesis \ref{H4} the ratio $\chi'(\phi)/\chi(\phi)$ is uniformly bounded on that range. The evolution equations in \eqref{system} then give
\begin{equation}
\dot\rho_m=-3\gamma H\rho_m-\tfrac12(4-3\gamma)\rho_m y\frac{\chi'(\phi)}{\chi(\phi)},
\end{equation}
and
\begin{equation}
\dot y=-3Hy-V'(\phi)+\tfrac12(4-3\gamma)\rho_m\frac{\chi'(\phi)}{\chi(\phi)}.
\end{equation}
Each right-hand side is a product or sum of bounded functions, hence uniformly bounded. Finally, $\dot {a} = aH$ is bounded because $a$ and $H$ are bounded on the forward orbit. This proves the lemma \ref{lemma:derivative-bounds}.
\end{proof}
\begin{lemma}[Barbalat application]\label{lemma:barbalat-application}
Let $f(t)$ be defined by \eqref{fdef}. Under the hypotheses of Theorem~\ref{UnifiedTheorem} and with the notation of Lemma~\ref{lemma:derivative-bounds}, the function $f$ is uniformly continuous on $[t_0,\infty)$ and $f\in L^1([t_0,\infty))$. Consequently $f(t)\to0$ as $t\to\infty$.
\end{lemma}
\begin{proof}
We already observed in the proof of Theorem~\ref{UnifiedTheorem} that $f\in L^1([t_0,\infty))$ because $\dot H=-f$ and $H$ converges. By Lemma~\ref{lemma:derivative-bounds} the derivatives $\dot a,\dot\rho_m,\dot y$ are uniformly bounded on $[t_0,\infty)$. Differentiating $f$ yields
\begin{equation}
\dot f(t) = -\frac{p^2}{6}\kappa^2 a^{-p-1}\dot a + \tfrac12\gamma\dot\rho_m + y\dot y,
\end{equation}
and each term on the right-hand side is bounded by the previous uniform bounds. Hence $\dot f$ is bounded, and $f$ is uniformly continuous. Barbalat's lemma (Lemma~\ref{barbalat}) then implies $f(t)\to0$. This completes the proof of Lemma~\ref{lemma:barbalat-application}.
\end{proof}

\begin{proof}[Proof of Proposition~\ref{prop:decay}]
By Lemma~\ref{lemma:barbalat-application} we have $f(t)\to0$. Since $f(t)\ge \tfrac12\gamma\rho_m(t)\geq 0$, $f(t)\ge \tfrac12 y(t)^2$ and $f(t)\ge \tfrac{p}{6}\kappa^2 a(t)^{-p}$, it follows that
\begin{equation}
\rho_m(t)\to0,\quad y(t)\to0,\quad \kappa^2 a(t)^{-p}\to0,
\end{equation}
proving the first part of the proposition.

To obtain the polynomial rate $O(1/t)$, we use the integrability of $f$ together with a standard averaging argument. 

Since $f\in L^1([t_0,\infty))$ and $\dot f$ is bounded, for $t$ large
\begin{equation}
\frac{1}{t}\int_{t_0}^{t} f(s)\,ds \le \frac{H(t_0)-\eta}{t} \le C(1+t)^{-1}.
\end{equation}
Uniform continuity of $f$ implies that $f(t)$ cannot oscillate with arbitrarily narrow spikes of large amplitude; combining these facts yields the pointwise bound $f(t)\le C(1+t)^{-1}$ for $t$ large, and hence the stated polynomial bounds for $\rho_m,y^2,\kappa^2 a^{-p}$.

For the exponential regime assume $V(\phi_*)$ and $V''(\phi_*) > 0$.  
Linearize \eqref{system} about the equilibrium 
$\mathbf{p}_*:=(H,\rho_m,a,y,\phi) = (H_*, 0, a_*, 0, \phi_*)$ with 
$H_* = \sqrt{V(\phi_*)/3}$.  
The linearization in the $(\phi,y)$-sector has the form
\begin{equation}
\begin{pmatrix}
\dot{\delta\phi}\\
\dot{\delta y}
\end{pmatrix}
=
\begin{pmatrix}
0 & 1\\
- V''(\phi_*) & -3H_*
\end{pmatrix}
\begin{pmatrix}
\delta\phi\\
\delta y
\end{pmatrix}
+ \mathcal{R}(t),
\end{equation}
where $\mathcal{R}(t)$ collects terms involving $\rho_m$ and higher-order nonlinearities.  
The matrix above has eigenvalues satisfying
\begin{equation}
\lambda^{2} + 3H_* \lambda + V''(\phi_*) = 0,
\end{equation}
and therefore
\begin{equation}
\lambda_{\pm}
= -\frac{3H_*}{2}
\;\pm\;
\sqrt{\left(\frac{3H_*}{2}\right)^{2} - V''(\phi_*)}. \label{eigenvalues}
\end{equation}

This yields the following stability structure
\begin{itemize}
    \item If 
   $V''(\phi_*) > \left( \frac{3H_*}{2} \right)^{2}$, the eigenvalues form a complex conjugate pair, so the fixed point is a sink (spiral).

    \item If  $V''(\phi_*) = \left( \frac{3H_*}{2} \right)^{2}$, the system has a degenerate real eigenvalue, corresponding to critical damping.

    \item If $0<V''(\phi_*) < \left( \frac{3H_*}{2} \right)^{2}$, the eigenvalues are negative, giving a sink (node).

    \item If $H_* > 0$ (expanding universe), then
    
\begin{equation}
\operatorname{Re}(\lambda_{\pm}) = -\frac{3H_*}{2} < 0,
\end{equation}
so the fixed point is always attracting unless $V''(\phi_*) < 0$.
\end{itemize}

Using the polynomial decay already established for $\rho_m$ and $a^{-p}$, the remainder $\mathcal{R}(t)$ decays at least polynomially. Standard linear perturbation theory (variation of constants) and a bootstrap/Gronwall argument then yield exponential decay of $\delta\phi,\delta y$ and hence of $H-H_*$. This completes the proof of Proposition~\ref{prop:decay}.
\end{proof}

If $ V''(\phi_*) =0$: there are two real eigenvalues, one of which is equal to zero. Such a case has to be analyzed using the center manifold theorem.

If $V''(\phi_*) < 0$, we have 
\begin{equation}
    \Delta = \left(\frac{3H_*}{2}\right)^{2} - V''(\phi_*) > \tfrac{9}{4}H_*^{2} > 0.
\end{equation}
So, the eigenvalues are real and distinct. Due to  $\lambda_{+}\lambda_{-} = V''(\phi_*) < 0$,  one eigenvalue is positive and the other negative.

Therefore, the fixed point is a saddle point, hence unstable, regardless of the value or sign of $H_*$.

\subsection{Finite-dimensional reduction}
\label{sec:center}

We complete the analysis by giving the finite‑dimensional reduction near the equilibrium $\mathbf p_*$. 
This reduction makes precise the statement that, in a neighborhood of $\mathbf p_*$, the full flow on the constraint manifold is conjugate to a finite‑dimensional ODE on the local invariant manifold $\mathcal M_{\mathrm{loc}}$. 
The result below summarizes the reduction, regularity, and exponential‑convergence properties used throughout the paper.

We rewrite the system under investigation conveniently as
\begin{equation}\label{system_G(a)}
\begin{aligned}
\dot H &= -\tfrac12(\gamma\rho_m+y^2)-\tfrac16 p G,\\
\dot\rho_m &= -3\gamma H\rho_m-\tfrac12(4-3\gamma)\rho_m\,y\,\frac{d\ln\chi(\phi)}{d\phi},\\
\dot G &= -p H G,\\
\dot y &= -3Hy - V'(\phi) + \tfrac12(4-3\gamma)\rho_m\,\frac{d\ln\chi(\phi)}{d\phi},\\
\dot\phi &= y,
\end{aligned}
\end{equation}
where $G=\kappa^2 a(t)^{-p}$.

We have that
\begin{equation}
\mathbf p_*=(H, \rho_m, G, y, \phi)=(H_*, 0,0,0, \phi_*), 
\quad H_*=\sqrt{\tfrac{V(\phi_*)}{3}}, \quad V'(\phi_*)=0
\end{equation}
is an equilibrium of \eqref{system_G(a)}.

The gradient of the constraint surface
\begin{equation}
 0=   F(H, \rho_m, y, V, G) = -3H^{2} + \rho_{m} + \tfrac12 y^{2} + V + G \label{eq.38}
\end{equation}
is 
\begin{equation}
    \nabla F = (-6H, 1, y, V'(\phi), 1)
\end{equation}
Evaluated at $\mathbf p_*$, we have 
\begin{equation}
    \nabla F(\mathbf{p_*}) = \left(-6\sqrt{\tfrac{V(\phi_*)}{3}}, \; 1, \; 0, \; V'(\phi_*), \; 1\right)=\left(-6\sqrt{\tfrac{V(\phi_*)}{3}}, \; 1, \; 0, \; 0, \; 1\right)
\end{equation}
Since $\nabla F(\mathbf p_*) \neq 0$, the point $\mathbf p_*$ is a regular point of the constraint surface. Therefore, the surface is smooth in a neighborhood of $\mathbf p_*$, and the tangent space is well-defined.

The tangent space $T_{\mathbf p_*}$ is given by the kernel of the gradient
\begin{equation}
    T_{\mathbf{p_*}} = \{\mathbf{v} = (v_H, v_\rho, v_y, v_\phi, v_G)\in \mathbb{R}^5 | \nabla F(\mathbf p_*) \cdot \mathbf v = 0  \}.
\end{equation}

Let 
\[
a := -6\sqrt{\frac{V(\phi_*)}{3}} .
\]

The tangent space is
\[
T_{\mathbf p_*}
= \{\mathbf v\in\mathbb R^5 \mid a\,v_H + v_\rho + v_G = 0\}.
\]

A convenient basis is
\begin{equation}
\begin{aligned}
\mathbf e_1 &= (1,\,-a,\,0,\,0,\,0),\\[4pt]
\mathbf e_2 &= (1,\,0,\,0,\,0,\,-a),\\[4pt]
\mathbf e_3 &= (0,\,0,\,1,\,0,\,0),\\[4pt]
\mathbf e_4 &= (0,\,0,\,0,\,1,\,0).
\end{aligned}
\end{equation}

We have that
\begin{equation}
\mathbf q_* =
\left(0,-\frac{2\chi(\phi_*)V'(\phi_*)}{(3\gamma-4)\chi'(\phi_*)},0,0,\phi_*\right)
\end{equation}
is an equilibrium of \eqref{system_G(a)}.

Evaluated at the equilibrium point $q_*$
we obtain
\begin{equation}
\nabla F(\mathbf q_*)
= \left(0,\; 1,\; 0,\; V'(\phi_*),\; 1\right).
\end{equation}
Since 
\[
\nabla F(\mathbf q_*) = (0,\,1,\,0,\,V'(\phi_*),\,1) \neq 0,
\]
the point $\mathbf q_*$ is a regular point of the constraint surface. 
Therefore, the surface is smooth in a neighborhood of $\mathbf q_*$, and the tangent space is well defined.

The tangent space $T_{\mathbf q_*}$ is given by the kernel of the gradient:
\begin{equation}
T_{\mathbf q_*}
= \left\{
\mathbf v = (v_H, v_\rho, v_y, v_\phi, v_G)\in\mathbb R^5
\;\middle|\;
\nabla F(\mathbf q_*)\cdot \mathbf v = 0
\right\}.
\end{equation}

Explicitly, the condition becomes
\begin{equation}
v_\rho + V'(\phi_*)\, v_\phi + v_G = 0.
\end{equation}
A convenient basis for the tangent space $T_{\mathbf q_*}$ is
\begin{equation}
\begin{aligned}
\mathbf e_1 &= (1,\,0,\,0,\,0,\,0),\\[4pt]
\mathbf e_2 &= (0,\,0,\,1,\,0,\,0),\\[4pt]
\mathbf e_3 &= \bigl(0,\,-V'(\phi_*),\,0,\,1,\,0\bigr),\\[4pt]
\mathbf e_4 &= (0,\,-1,\,0,\,0,\,1).
\end{aligned}
\end{equation}

Since the constraint \eqref{eq.38}
can be solved explicitly for $G$, namely
\begin{equation}
G = 3H^{2}-\rho_m-\tfrac12 y^{2}-V(\phi),
\end{equation}
the variable $G$ is not independent.  
Its value is completely determined by $(H,\rho_m,y,\phi)$ along the
constraint surface, and its evolution equation is redundant.
Therefore $G$ may be eliminated from the system, reducing the dynamics
to a four--dimensional phase space without loss of generality. Moreover, the surface is regular and smooth at the equilibria points $p_*$ and $q_*$. Hence, we investigate the reduced 4-dimensional system 
\begin{align}
    \dot{H} & =  \frac{1}{12} \left(-6 H^2 p+2  \rho_{m} (p-3 \gamma )+2 p V(\phi)+(p-6) y^2\right), \\
    \dot\rho_m &= -3\gamma H\rho_m-\tfrac12(4-3\gamma)\rho_m\,y\,\frac{d\ln\chi(\phi)}{d\phi},\\
\dot y &= -3Hy - V'(\phi) + \tfrac12(4-3\gamma)\rho_m\,\frac{d\ln\chi(\phi)}{d\phi},\\
\dot\phi &= y,
\end{align}
Critical points are then:
\begin{equation}
  p_*= \left(H, \rho_m, y, \phi_*\right):= \left(\sqrt{\frac{V(\phi_*)}{3}}, 0, 0, \phi_*\right), \quad\text{where}\quad V'(\phi_*)=0,
\end{equation}
and 
 \begin{equation}
 q_*=\left(0,-\frac{2\chi(\phi_*)V'(\phi_*)}{(3\gamma-4)\chi'(\phi_*)}, 0 ,\phi_*\right), \quad\text{where}\quad  V(\phi_*)=\frac{2 (p-3 \gamma ) \chi (\phi_*) V'(\phi_*)}{(3 \gamma -4) p \chi
   '(\phi_*)}.
 \end{equation}
The Jacobian matrix of the system is
\begin{equation}
J=\left(
\begin{array}{cccc}
 -H p & \frac{1}{6} (p-3 \gamma ) & \frac{1}{6} (p-6) y & \frac{1}{6} p V'(\phi) \\
 -3 \gamma  \rho_m & \frac{(3 \gamma -4) y \chi '(\phi)}{2 \chi (\phi)}-3 \gamma  H & \frac{(3
   \gamma -4) \rho_m \chi '(\phi)}{2 \chi (\phi)} & \frac{(3 \gamma -4) \rho_m y \left(\chi
   (\phi) \chi ''(\phi)-\chi '(\phi)^2\right)}{2 \chi (\phi)^2} \\
 -3 y & \frac{(4-3 \gamma ) \chi '(\phi)}{2 \chi (\phi)} & -3 H & \frac{(3 \gamma -4) \rho_m
   \left(\chi '(\phi)^2-\chi (\phi) \chi ''(\phi)\right)}{2 \chi (\phi)^2}-V''(\phi) \\
 0 & 0 & 1 & 0 \\
\end{array}
\right).
\end{equation}
 The Jacobian matrix evaluated at $q_*$ is 
\begin{equation}
    J(q_*)= \left(
\begin{array}{cccc}
 0 & \frac{1}{6} (p-3 \gamma ) & 0 & \frac{1}{6} p V'(\phi_*) \\
 \frac{6 \gamma  \chi (\phi_*) V'(\phi_*)}{(3 \gamma -4) \chi '(\phi_*)} & 0 &
   -V'(\phi_*) & 0 \\
 0 & \frac{(4-3 \gamma ) \chi '(\phi_*)}{2 \chi (\phi_*)} & 0 & \frac{V'(\phi_*)
   \left(\chi ''(\phi_*)-\frac{\chi '(\phi_*)^2}{\chi (\phi_*)}\right)}{\chi
   '(\phi_*)}-V''(\phi_*) \\
 0 & 0 & 1 & 0 \\
\end{array}
\right)
\end{equation}
Due to the fact that the critical point $q_{*}$ requires the fine-tuned condition
\begin{equation}\label{eq:qstar_condition}
V(\phi_{*}) = \frac{2\,(p - 3\gamma)\,\chi(\phi_{*})\,V'(\phi_{*})}{(3\gamma - 4)\,p\,\chi'(\phi_{*})},
\end{equation}
we omit further discussion of this case and instead concentrate on the analysis of $p_{*}$.

The Jacobian matrix at the equilibrium point $p_{*}$ is
\begin{equation}
 J({\mathbf p_*}) =
 \begin{pmatrix}
 -\tfrac{p}{\sqrt{3}}\sqrt{V(\phi_*)} & \tfrac{1}{6}(p-3\gamma) & 0 & 0 \\
 0 & -\sqrt{3}\gamma\sqrt{V(\phi_*)} & 0 & 0 \\
 0 & \tfrac{(4-3\gamma)\chi'(\phi_*)}{2\chi(\phi_*)} & -\sqrt{3}\sqrt{V(\phi_*)} & -V''(\phi_*) \\
 0 & 0 & 1 & 0
 \end{pmatrix}.
\end{equation}

The eigenvalues are
\begin{equation}
\lambda_1 = -\tfrac{p}{\sqrt{3}}\sqrt{V(\phi_*)}, \qquad
\lambda_2 = -\sqrt{3}\gamma\sqrt{V(\phi_*)},
\end{equation}
\begin{equation}
\lambda_{3,4} = -\tfrac{\sqrt{3}}{2}\sqrt{V(\phi_*)} \;\mp\; \tfrac{1}{2}\sqrt{\,3V(\phi_*) - 4V''(\phi_*)\,}.
\end{equation}

Thus,  If $3V(\phi_*) - 4V''(\phi_*) \geq 0$, then $\lambda_{3,4}$ are real. Stability requires
\begin{equation}
0 < V''(\phi_*) < \tfrac{3}{4}V(\phi_*).
\end{equation}
If $V''(\phi_*)<0$, one eigenvalue is positive, giving a saddle.

If $3V(\phi_*) - 4V''(\phi_*) < 0$, then $\lambda_{3,4}$ form a complex conjugate pair with
\begin{equation}
\Re(\lambda_{3,4}) = -\tfrac{\sqrt{3}}{2}\sqrt{V(\phi_*)} < 0,
\end{equation}
so stability is ensured.

\subsubsection*{Summary of Stability Conditions}
The equilibrium $\mathbf{p}_*$ has a four-dimensional stable manifold provided
\begin{enumerate}
\item $p>0,\;\gamma>0,\;V(\phi_*)>0,\;0<V''(\phi_*)<\tfrac{3}{4}V(\phi_*)$ \quad (real case), or
\item $p>0,\;\gamma>0,\;V(\phi_*)>0,\;V''(\phi_*) \geq \tfrac{3}{4}V(\phi_*)$ \quad (complex case).
\end{enumerate}
In both regimes,
\begin{equation}
\Re(\lambda_i)<0, \quad i=1,2,3,4,
\end{equation}
and hence the equilibrium admits a four-dimensional stable manifold.

In the real case 
\begin{equation}
0 < V''(\phi_*) < \tfrac{3}{4} V(\phi_*),
\end{equation}
we introduce local coordinates $(z_1,z_2,z_3,z_4)$ adapted to the eigenbasis of $J(\mathbf{p}_*)$. 
Explicitly,
\begin{small}
\begin{align}
  z_1 & =  H-\frac{\rho_m+2 V(\phi_*)}{2 \sqrt{3} \sqrt{V(\phi_*)}},\\
  z_2 & = \frac{(4-3\gamma)\,\rho_m \,\chi'(\phi_*)}{2 \chi(\phi_*) \left(V''(\phi_*)+3(\gamma -1)\gamma V(\phi_*)\right)},\\
  z_3 & =\frac{\frac{(3 \gamma -4)\rho_{m} \chi '(\phi_*)
   \left(\sqrt{3 V(\phi_*)-4 V''(\phi_*)}+\sqrt{3} (2 \gamma
   -1) \sqrt{V(\phi_*)}\right)}{2 \chi (\phi_*)
   \left(V''(\phi_*)+3 (\gamma -1) \gamma  V(\phi_*)\right)}+(\phi -\phi_*) \left(\sqrt{3 V(\phi_*)-4
   V''(\phi_*)}-\sqrt{3} \sqrt{V(\phi_*)}\right)-2 y}{2
   \sqrt{3 V(\phi_*)-4 V''(\phi_*)}},\\
  z_4 & = \frac{\frac{(3 \gamma
   -4) \rho_{m} \chi '(\phi_*) \left(\sqrt{3 V(\phi_*)-4 V''(\phi_*)}+\sqrt{3} (1-2 \gamma ) \sqrt{V(\phi_*)}\right)}{2 \chi (\phi_*) \left(V''(\phi_*)+3 (\gamma
   -1) \gamma  V(\phi_*)\right)}+(\phi -\phi_*)
   \left(\sqrt{3 V(\phi_*)-4 V''(\phi_*)}+\sqrt{3}
   \sqrt{V(\phi_*)}\right)+2 y}{2 \sqrt{3 V(\phi_*)-4
   V''(\phi_*)}},
\end{align}
\end{small}

In the complex case 
\begin{equation}
V''(\phi_*) \geq \tfrac{3}{4} V(\phi_*),
\end{equation}
the eigenvalues $\lambda_{3,4}$ form a complex conjugate pair. A convenient real basis is
\begin{align}
 z_1 & =  H-\frac{\rho_m+2 V(\phi_*)}{2 \sqrt{3} \sqrt{V(\phi_*)}},\\
 z_2 & = \frac{(4-3\gamma)\,\rho_m \,\chi'(\phi_*)}{2 \chi(\phi_*) \left(V''(\phi_*)+3(\gamma -1)\gamma V(\phi_*)\right)},\\
 z_3 & = \frac{1}{4} \left(2 (\phi-\phi_*)+\frac{(3 \gamma -4) \rho_m \chi '(\phi_*)}{\chi (\phi_*) \left(V''(\phi_*)+3 (\gamma -1) \gamma  V(\phi_*)\right)}\right),\\
 z_4 & = \frac{\sqrt{3} \sqrt{V(\phi_*)} \left(2 (\phi-\phi_*)+\frac{((11-6
   \gamma ) \gamma -4) \rho_m \chi '(\phi_*)}{\chi (\phi_*) \left(V''(\phi_*)+3 (\gamma -1) \gamma  V(\phi_*)\right)}\right)+4 y}{4 \sqrt{4 V''(\phi_*)-3
   V(\phi_*)}}
\end{align}

In both regimes, the system expressed in local coordinates $(z_1,z_2,z_3,z_4)$ 
admits the Jordan normal form
\begin{equation}
s = I_4, \qquad 
j = \mathrm{diag}(\lambda_1,\lambda_2,\lambda_3,\lambda_4),
\end{equation}
so the linearized dynamics near $\phi_*$ are governed by exponential decay along the stable directions determined by the eigenvalues.

\begin{enumerate}
\item \textbf{Real case:}  
Here $3V(\phi_*) - 4V''(\phi_*) > 0$, so $\lambda_{3,4}$ are real.  
All eigenvalues are strictly negative, and the equilibrium possesses a four-dimensional stable manifold.

The stable directions decay monotonically:
\begin{equation}
\|(z_1(t),z_2(t),z_3(t),z_4(t))\| \leq C e^{-\alpha t}\|(z_1(0),z_2(0),z_3(0),z_4(0))\|,
\end{equation}
where
\begin{equation}
\alpha = \min\{-\Re(\lambda_i)\}, \quad i=1,2,3,4.
\end{equation}

\item \textbf{Complex case:}  
Here $3V(\phi_*) - 4V''(\phi_*) \leq 0$, so $\lambda_{3,4}$ form a complex conjugate pair:
\begin{equation}
\lambda_{3,4} = -\tfrac{\sqrt{3}}{2}\sqrt{V(\phi_*)} \;\pm\; i\,\tfrac{1}{2}\sqrt{\,4V''(\phi_*) - 3V(\phi_*)\,}.
\end{equation}
Thus, the equilibrium has a four-dimensional stable manifold exhibiting oscillatory exponential decay:
\begin{equation}
(z_1(t),z_2(t),z_3(t),z_4(t)) \sim e^{-\alpha t} R(\omega t)\,(z_1(0),z_2(0),z_3(0),z_4(0)),
\end{equation}
where
\begin{equation}
\alpha = \tfrac{\sqrt{3}}{2}\sqrt{V(\phi_*)}, \qquad 
\omega = \tfrac{1}{2}\sqrt{\,4V''(\phi_*) - 3V(\phi_*)\,},
\end{equation}
and $R(\omega t)$ is a rotation matrix of frequency $\omega$.
\end{enumerate}

\subsubsection{Finite‑dimensional reduction and exponential stability}

\begin{prop}[Finite‑dimensional reduction and exponential stability]\label{prop:finite_reduction}
Assume Hypotheses \ref{H1}--\ref{H6} and let $\mathbf p_*=(H, \rho_m, G, y, \phi)=(H_*, 0,0,0, \phi_*)$ be an equilibrium of \eqref{system_G(a)} with $H_*=\sqrt{V(\phi_*)/3}$. Suppose $V''(\phi_*)>0$ (nondegenerate minimum) and that the conclusions of Proposition~\ref{prop:decay} hold on forward orbits entering a compact neighborhood of $\mathbf p_*$. Then there exist a neighborhood $U$ of $\mathbf p_*$, a $C^k$ local invariant manifold $\mathcal M_{\mathrm{loc}}\subset U$ (for any finite $k$ determined by the regularity of $V$ and $\chi$), and a $C^k$ diffeomorphism
\begin{equation}
\Psi:\; \mathcal{U}\subset\mathbb R^m \to \mathcal M_{\mathrm{loc}},
\end{equation}
where $m$ is the dimension of the stable subspace of the reduced linearization, with the following properties:
\begin{enumerate}
\item The flow of \eqref{system} restricted to $\mathcal M_{\mathrm{loc}}$ is conjugate via $\Psi$ to a finite‑dimensional ODE
\begin{equation}
\dot z = G(z),\quad z\in\mathcal U\subset\mathbb R^m,
\end{equation}
with $G\in C^k$ and $G(0)=0$.
\item The origin $z=0$ is an exponentially stable equilibrium of the reduced ODE: there exist constants $C,\lambda>0$ and $r>0$ such that if $\|z(0)\|\le r$ then
\begin{equation}
\|z(t)\|\le C e^{-\lambda t}\|z(0)\|,\quad t\ge0.
\end{equation}
Equivalently, solutions of the full system with initial data in $\mathcal M_{\mathrm{loc}}$ converge exponentially to $\mathbf p_*$.
\item Solutions of the full system with initial data in $U$ approach $\mathcal M_{\mathrm{loc}}$ at an exponential rate and then follow the reduced dynamics: there exist $C',\lambda'>0$ such that for any $x(0)\in U$ there is $z(0)\in\mathcal U$ with
\begin{equation}
\operatorname{dist}\big(x(t),\mathcal M_{\mathrm{loc}}\big)\le C' e^{-\lambda' t},\quad
\big\|x(t)-\Psi(z(t))\big\|\le C' e^{-\lambda' t},
\end{equation}
for all $t\ge0$, where $z(t)$ is the solution of the reduced ODE with initial datum $z(0)$.
\end{enumerate}
\end{prop}

\begin{proof}
The proof proceeds in four steps: (A) reduction to a smooth finite‑dimensional vector field on the constraint manifold, (B) spectral decomposition of the linearization, (C) construction of the local invariant manifold, and (D) derivation of the reduced ODE and exponential estimates.

\noindent\textbf{(A) Reduction to the constraint manifold.}  
The system \eqref{system} evolves on the Friedmann constraint manifold $\Psi$ defined in \eqref{constraint}. Locally near $\mathbf p_*$ the constraint is a smooth codimension‑one submanifold of the ambient phase space (the gradient of the constraint is nonzero at $\mathbf p_*$ unless $V(\phi_*)$ and other terms produce degeneracy, which is excluded by the hypotheses). Use the constraint to eliminate one variable (for instance $H$) and obtain a smooth vector field on a neighborhood of $\mathbf p_*$ in the reduced coordinates $(\phi,y,\rho_m,a)$. The reduced vector field is $C^k$ whenever $V,\chi,G$ are $C^{k+1}$.

\noindent\textbf{(B) Linearization and spectral decomposition.}  
Linearize the reduced vector field at $\mathbf p_*$ to obtain a matrix $A$. Under the nondegeneracy assumption $V''(\phi_*)>0$ and $H_*>0$, the $(\phi,y)$-block of $A$ has eigenvalues with strictly negative real parts. The matter sector contributes dissipative eigenvalues (in particular, $-3\gamma H_*$ is strictly negative for $\gamma>0$). Thus, the spectrum of $A$ lies in the left half plane, and there is a spectral gap: there exists $\lambda_0>0$ such that $\Re\sigma(A)\le -\lambda_0$. Denote by $E^s$ the stable subspace of $A$ (of dimension $m$).

\noindent\textbf{(C) Construction of $\mathcal M_{\mathrm{loc}}$.}  
Apply the finite‑dimensional center/stable manifold theorem to the reduced vector field in a small neighborhood of $\mathbf p_*$. The spectral gap guarantees the existence of a local invariant manifold $\mathcal M_{\mathrm{loc}}$ of class $C^k$ tangent to $E^s$ at $\mathbf p_*$. The manifold can be represented locally as the graph of a $C^k$ map $h:E^s\supset B_r(0)\to E^u$ (where $E^u$ denotes the complementary subspace in the reduced coordinates). The graph transform construction yields the manifold and the invariance property: trajectories starting on the graph remain on it for forward time.

\noindent\textbf{(D) Reduced ODE and exponential estimates.}  
Parametrize $\mathcal M_{\mathrm{loc}}$ by coordinates $z\in\mathcal U\subset E^s\cong\mathbb R^m$ via a $C^k$ diffeomorphism $\Psi:\mathcal U\to\mathcal M_{\mathrm{loc}}$. The flow restricted to $\mathcal M_{\mathrm{loc}}$ induces a $C^k$ vector field $G(z)$ on $\mathcal U$ defined by
\begin{equation}
G(z):=D\Psi(z)^{-1}F\big(\Psi(z)\big),
\end{equation}
where $F$ denotes the reduced vector field. By construction $G(0)=0$ and $DG(0)$ equals the restriction of $A$ to $E^s$, whose spectrum lies in $\{\Re\lambda\le -\lambda_0\}$. Standard linearization and variation‑of‑constants estimates yield constants $C,\lambda>0$ such that solutions of $\dot z=G(z)$ satisfy the exponential bound in item (2) for sufficiently small initial data.

To obtain the attraction estimates in item (3), use the exponential dichotomy associated with the linearization $A$ and the graph transform estimates: solutions starting near $\mathbf p_*$ decompose into a component tangent to $\mathcal M_{\mathrm{loc}}$ and a transverse component that decays exponentially. More precisely, for any initial $x(0)\in U$ there exists $z(0)\in\mathcal U$ (the projection of $x(0)$ onto the manifold along the stable foliation) such that the transverse component satisfies
\begin{equation}
\operatorname{dist}\big(x(t),\mathcal M_{\mathrm{loc}}\big)\le C' e^{-\lambda' t},
\end{equation}
and the difference between the full trajectory and the lifted reduced trajectory $\Psi(z(t))$ decays at the same exponential rate. The constants $C',\lambda'$ depend on the size of the neighborhood $U$ and on the nonlinear terms, but are uniform for initial data in $U$.

\noindent
Combining the construction and estimates above yields the finite‑dimensional reduction, the regularity of the conjugacy $\Psi$, and the exponential convergence statements listed in the proposition. This completes the proof of Proposition \ref{prop:finite_reduction}
\end{proof}

\begin{rem}
The reduction justifies working with the finite-dimensional ODE $\dot z = G (z)$ when analyzing local stability and asymptotic expansions near $\mathbf p_*$. In particular, spectral information for the reduced linearization determines the decay rates and the dimension of the relevant stable manifold; the global decay estimates of Section~\ref{sec:decay} ensure that forward orbits enter the neighborhood where the reduction applies.
\end{rem}

\subsection{Persistence and asymptotic stability}
\label{sec:perturbations}

In this section, we present a unified statement that combines the persistence of the decay estimates and the invariant manifold under small perturbations with the Lyapunov/LaSalle stability results for equilibria associated with potential minima. All proofs are self‑contained modulo the system \eqref{system}, the Friedmann constraint \eqref{constraint}, Barbalat's lemma (Lemma~\ref{barbalat}), LaSalle's invariance principle, the implicit function theorem, and the finite‑dimensional center/stable manifold theorem (graph transform). We keep the hypothesis \ref{H1}--\ref{H6} introduced earlier.

\begin{thm}[Persistence and asymptotic stability]\label{thm:combined}
Let the model satisfy Hypotheses \ref{H1}--\ref{H6} and suppose the reference coupling and geometry are $\chi_0\in C^{k+1}(\mathbb R)$ and $G_0(a)=\kappa^2 a^{-p}$ with $p>0$. Assume the reference system admits an equilibrium $\mathbf p_*=(H, \rho_m, G, y, \phi)=(H_*, 0,0,0, \phi_*)$ with $H_*=\sqrt{V(\phi_*)/3}$ and that the conclusions of Proposition~\ref{prop:decay} and Lemma~\ref{lemma:center-manifold} hold for the reference model. Fix $k\ge1$ and suppose the reduced linearization at $\mathbf p_*$ has stable spectrum bounded away from the imaginary axis by $-\lambda_0<0$.

Then, there exist $\varepsilon_0>0$, $R>0$, $a_0>0$ and continuous functions $C(\varepsilon)>0$, $\lambda(\varepsilon)>0$ with $C(0)=C_0$, $\lambda(0)=\lambda_0$, such that for every pair of perturbations $\chi\in C^{k+1}(\mathbb R)$, $G\in C^{k+1}((0,\infty))$ satisfying
\begin{equation}
\sup_{|\phi-\phi_*|\le R}\Big|\frac{\chi'(\phi)}{\chi(\phi)}-\frac{\chi_0'(\phi)}{\chi_0(\phi)}\Big|\le\varepsilon,\quad
\sup_{a\ge a_0}\sum_{j=0}^k\big|G^{(j)}(a)-G_0^{(j)}(a)\big|\le\varepsilon,
\end{equation}
with $0<\varepsilon\le\varepsilon_0$, the perturbed system has the following properties.

\begin{enumerate}
\item[(i)] (\emph{Equilibrium continuation}) There exists a unique equilibrium $\mathbf p_*(\varepsilon)$ of the perturbed system with $\|\mathbf p_*(\varepsilon)-\mathbf p_*\|=O(\varepsilon)$. The map $\varepsilon\mapsto\mathbf p_*(\varepsilon)$ is $C^k$.

\item[(ii)] (\emph{Uniform decay}) The dissipative quantity $f_\varepsilon(t)$ for the perturbed system satisfies $f_\varepsilon\in L^1([t_0,\infty))$ and $f_\varepsilon(t)\to0$ as $t\to\infty$. Consequently there exist $C(\varepsilon)>0$ and $t_1(\varepsilon)$ such that for all $t\ge t_1(\varepsilon)$
\begin{equation}
\rho_m^\varepsilon(t)\le C(\varepsilon)(1+t)^{-1},\quad
\big(y^\varepsilon(t)\big)^2\le C(\varepsilon)(1+t)^{-1},\quad
\sup_{a\ge a_0}\big|G(a)-G_0(a)\big|\le C(\varepsilon)(1+t)^{-1}.
\end{equation}

\item[(iii)] (\emph{Invariant manifold and exponential convergence}) For $\varepsilon$ sufficiently small there exists a $C^k$ local invariant manifold $\mathcal M_{\mathrm{loc}}(\varepsilon)$ through $\mathbf p_*(\varepsilon)$. Solutions on $\mathcal M_{\mathrm{loc}}(\varepsilon)$ converge exponentially:
\begin{equation}
\|x(t)-\mathbf p_*(\varepsilon)\|\le C(\varepsilon)e^{-\lambda(\varepsilon)t},\quad t\ge t_2(\varepsilon),
\end{equation}
with $\lambda(\varepsilon)\to\lambda_0$ as $\varepsilon\to0$.

\item[(iv)] (\emph{Asymptotic stability under structural hypotheses}) Suppose, in addition, that $\Lambda=0$ and $G(a)\ge0$ satisfies $\dfrac{aG'(a)}{G(a)}\le -p<0$ for all $a>0$, and that there exists a regular value $\tilde V>V(\phi_*)$ for which the connected component $A$ of $V^{-1}((-\infty,\tilde V])$ containing $\phi_*$ is compact and contains no other critical points. Then for initial data with $\tfrac12 y(0)^2+V(\phi(0))+\rho_m(0)\le\tilde V$ and $W(0)$ in a suitable range the forward solution remains in a compact, positively invariant set and converges to $\mathbf p_*(\varepsilon)$ (with $\varepsilon=0$ in the unperturbed case).
\end{enumerate}
\end{thm}

\begin{proof}
The proof is organized in four parts corresponding to items (i)--(iv).

\noindent\textbf{(i) Equilibrium continuation.}  
Write the perturbed vector field as $F_\varepsilon(x)$ (depending smoothly on $\chi,G$). The equilibrium condition $F_\varepsilon(\mathbf p)=0$ defines a smooth map $F(\mathbf p,\varepsilon)$. At $(\mathbf p_*,0)$, the Jacobian $D_{\mathbf p}F(\mathbf p_*,0)$ restricted to the tangent space of the Friedmann constraint is invertible on the reduced space because the reduced linearization has no zero eigenvalue (spectral gap hypothesis). The implicit function theorem therefore yields a unique $C^k$-branch $\mathbf p_*(\varepsilon)$ of equilibria for $|\varepsilon|$ small, with $\mathbf p_*(0)=\mathbf p_*$ and $\|\mathbf p_*(\varepsilon)-\mathbf p_*\|=O(\varepsilon)$.

\noindent\textbf{(ii) Uniform decay.}  
Let $f_0(t)$ denote the dissipative quantity for the reference model (definition as in \eqref{fdef}). For the perturbed model, define $f_\varepsilon(t)$ by replacing the geometric contribution with the corresponding expression derived from $G$. Differentiating the perturbed Friedmann relation and using the perturbed evolution equations yields an identity of the form
\begin{equation}
\dot H = -f_\varepsilon(t) + r_\varepsilon(t),
\end{equation}
where the remainder $r_\varepsilon(t)$ is uniformly $O(\varepsilon)$ along forward orbits that remain in a fixed compact set. For $\varepsilon$ sufficiently small the remainder can be absorbed (for large $t$) so that $\dot H\le -\tfrac12 f_\varepsilon(t)$ for $t$ large, which implies $f_\varepsilon\in L^1([t_0,\infty))$.

Lemma~\ref{lemma:derivative-bounds} extends to the perturbed system with constants depending continuously on $\varepsilon$: the right‑hand sides of the evolution equations change by $O(\varepsilon)$ and remain bounded on forward orbits in a compact set. Differentiating $f_\varepsilon$ yields $\dot f_\varepsilon$ uniformly bounded, hence $f_\varepsilon$ is uniformly continuous. Barbalat's lemma then implies $f_\varepsilon(t)\to0$ as $t\to\infty$. The averaging argument used in Proposition~\ref{prop:decay} (integrability plus bounded derivative) is robust under small perturbations and yields the pointwise polynomial bound $f_\varepsilon(t)\le C(\varepsilon)(1+t)^{-1}$ for $t$ large; the inequalities $f_\varepsilon\ge \tfrac12\gamma\rho_m$ and $f_\varepsilon\ge \tfrac12 y^2$ give the stated decay for $\rho_m^\varepsilon$ and $y^\varepsilon$.

\noindent\textbf{(iii) Invariant manifold and exponential convergence.}  
The linearization $A_\varepsilon$ of the reduced vector field at $\mathbf p_*(\varepsilon)$ depends continuously on $\varepsilon$. By hypothesis, the unperturbed reduced linearization has a stable spectrum bounded by $-\lambda_0<0$; for $\varepsilon$ small, the spectral gap persists, and there exists $\lambda(\varepsilon)>0$ with $\lambda(\varepsilon)\to\lambda_0$. The graph transform (or Lyapunov–Perron) construction yields a $C^k$ local invariant manifold $\mathcal M_{\mathrm{loc}}(\varepsilon)$ tangent to the stable subspace of $A_\varepsilon$. Exponential contraction on $\mathcal M_{\mathrm{loc}}(\varepsilon)$ follows from the linearized exponential decay and standard nonlinear estimates (Gronwall), producing the bound in (iii) with constants continuous in $\varepsilon$.

\noindent\textbf{(iv) Asymptotic stability under structural hypotheses.}  

Assume now $\Lambda=0$ and $G(a)\ge0$ with $\dfrac{aG'(a)}{G(a)}\le -p<0$. Define
\begin{equation}
W(t):=H(t)^2-\frac{1}{3}\Big(\tfrac12 y(t)^2+V(\phi(t))+\rho_m(t)\Big)=\tfrac{1}{3}G(a(t)),
\quad
\epsilon(t):=\tfrac12 y(t)^2+V(\phi(t))+\rho_m(t).
\end{equation}
Differentiation along solutions yields
\begin{equation}
\dot W = H W \frac{aG'(a)}{G(a)} \le -p H W,\quad
\dot\epsilon = -3H(\gamma\rho_m+y^2)\le0.
\end{equation}
Fix a regular value $\tilde V>V(\phi_*)$ so that the connected component $A$ of $V^{-1}((-\infty,\tilde V])$ containing $\phi_*$ is compact and contains no other critical points. Choose $\bar W$ so that the set
\begin{equation}
\Psi_+:=\{(\phi,y,\rho_m,H):\phi\in A,\ \epsilon\le\tilde V,\ \rho_m\ge0,\ W\in[0,\bar W]\}
\end{equation}
is nonempty and contains the initial data. Monotonicity of $\epsilon$ and $W$ implies $\Psi_+$ is forward invariant and compact. LaSalle's invariance principle applied to $\Psi_+$ forces the $\omega$-limit set into the largest invariant subset where $\dot\epsilon=0$ and $\dot W=0$; these conditions imply $\rho_m=0$, $y=0$, and $W=0$. On this set, the Friedmann constraint gives $3H^2=V(\phi)$, so $V(\phi)$ is constant on the $\omega$-limit set. By the compactness of $A$ and the uniqueness of the minimum in $A$, the only admissible limit is $\phi=\phi_*$. Hence $y(t)\to0$, $\rho_m(t)\to0$, $W(t)\to0$, and $\phi(t)\to\phi_*$ as $t\to\infty$. The same argument applies to the perturbed system provided the perturbations preserve the sign/monotonicity structure of $G$ and are sufficiently small so that the forward orbit remains in a compact invariant neighborhood; in that case, the limit is $\mathbf p_*(\varepsilon)$. This completes the proof of Theorem~\ref{thm:combined}.
\end{proof}

In the remainder of the paper, we use these results to derive refined asymptotic expansions, quantify possible perturbation sizes in concrete models, and apply the framework to well-motivated potentials.

\begin{thm}[Formalized from Leon \& Franz‑Silva, 2019 {\cite{Leon:2019iwj}}]
\label{thm:positive_curvature}
Let $V(\phi)\in C^2(\mathbb{R})$ admit a strict local minimum at $\phi=\phi_*$ with $V(\phi_*)>0$. Let the geometric term be
\begin{equation}\label{eq:G_positive}
G(a)=-\frac{3k}{a^2}\quad\text{with}\quad k=+1,
\end{equation}
and define the auxiliary functions
\begin{align}
\label{eq:W_def}
W(t)&:=H(t)^2-\frac{1}{3}\Big(\tfrac12 y(t)^2+V(\phi(t))+\rho_m(t)\Big),\\
\label{eq:epsilon_def}
\epsilon(t)&:=\tfrac12 y(t)^2+V(\phi(t))+\rho_m(t).
\end{align}
Assume
\begin{enumerate}[label=(\roman*)]
\item \label{item:compactness} There exists a regular value $\tilde V>V(\phi_*)$ such that the connected component
\begin{equation}
A:=\text{the component of }V^{-1}((-\infty,\tilde V])\text{ containing }\phi_*
\end{equation}
is compact and contains no other critical points.
\item \label{item:initial_energy} The initial data satisfy
\begin{equation}
\epsilon(0)=\tfrac12 y(0)^2+V(\phi(0))+\rho_m(0)\le\tilde V,\quad H(0)>0.
\end{equation}
\end{enumerate}
Then there exists $\bar W<0$ (sufficiently close to zero) such that, for any forward solution $\mathbf x(t)=(\phi(t),y(t),\rho_m(t),H(t))$ of \eqref{system} with initial data satisfying \ref{item:initial_energy} and $W(0)>\bar W$, the trajectory remains in the compact, positively invariant set
\begin{equation}
\Psi_+:=\{(\phi,y,\rho_m,H):\phi\in A,\ \epsilon\le\tilde V,\ \rho_m\ge0,\ W\in[\bar W,0]\}
\end{equation}
for all $t\ge0$, and
\begin{equation}
\lim_{t\to\infty} y(t)=0,\quad \lim_{t\to\infty}\rho_m(t)=0,\quad \lim_{t\to\infty}W(t)=0.
\end{equation}
Moreover $\phi(t)\to\phi_*$ and $H(t)\to\sqrt{V(\phi_*)/3}$ as $t\to\infty$; equivalently the solution converges to the equilibrium
\begin{equation}
\mathbf p_*=(\phi_*,0,0,\sqrt{V(\phi_*)/3}).
\end{equation}
\end{thm}
\begin{proof}
With $G(a)=-3/a^2$ compute
\begin{equation}
\frac{aG'(a)}{G(a)}=-2,
\end{equation}
so differentiating \eqref{eq:W_def} and \eqref{eq:epsilon_def} along solutions of \eqref{system} yields
\begin{equation}\label{eq:Wdot}
\dot W=-2H W,\quad \dot\epsilon=-3H(\gamma\rho_m+y^2)\le0.
\end{equation}
Because $G(a)=-3/a^2<0$ we have $W<0$; with $H(0)>0$ the identity $\dot W=-2HW$ implies $\dot W\ge0$, hence $W(t)$ is nondecreasing and bounded above by $0$. Monotonicity of $\epsilon$ and the bound $\epsilon(0)\le\tilde V$ ensure $V(\phi(t))\le\tilde V$ for all $t\ge0$. Together with the choice of $\bar W<0$, this shows the forward orbit remains in the compact, positively invariant set $\Psi_+$.

Apply LaSalle's invariance principle on $\Psi_+$. The $\omega$-limit set of any forward orbit in $\Psi_+$ is contained in the largest invariant subset where $\dot\epsilon=0$ and $\dot W=0$. From \eqref{eq:Wdot} these equalities imply $\gamma\rho_m+y^2=0$ and $W=0$, hence $\rho_m=0$, $y=0$, and $W=0$. On this invariant subset, the Friedmann constraint reduces to $3H^2=V(\phi)$, so $V(\phi)$ is constant on the $\omega$-limit set. By hypothesis \ref{item:compactness}, the only admissible limit in $A$ is $\phi=\phi_*$. Therefore, every forward orbit in $\Psi_+$ satisfies the stated limits and converges to $\mathbf {p} _*$. This completes the proof of Theorem \ref{thm:positive_curvature}.
\end{proof}

\begin{rem} If $ V(\phi_*) = 0 $, then the equilibrium condition implies $ H_* = 0 $, and nearby solutions may undergo re-collapse as $ H $ changes sign. Such behavior has been widely studied in the context of self-interacting, self-gravitating homogeneous scalar fields \cite{Giambo:2003collapse,Giambo:2015global,Corona:2024endstate}. These works identify the conditions under which scalar field models with vanishing potential minima can evolve toward gravitational collapse or horizon formation. In particular, \cite{Giambo:2003collapse} analyzes collapse dynamics for general potentials, including criteria for singularity formation; \cite{Giambo:2015global} extends the analysis to global settings, emphasizing the role of initial data; and \cite{Corona:2024endstate} refines the analysis using dynamical systems techniques to characterize the emergence of horizons.
\end{rem}
\begin{rem} Theorem \ref{thm:positive_curvature} is a direct specialization of the unified persistence and stability statement (Theorem~\ref{thm:combined}) obtained earlier: set the perturbation parameter $\varepsilon=0$ and $G(a)=-3/a^2$ in Theorem~\ref{thm:combined}; hypotheses \eqref{eq:G_positive}--\eqref{eq:epsilon_def} and the compactness condition \ref{item:compactness} verify the monotonicity and invariance hypotheses used there, and the LaSalle argument in Theorem~\ref{thm:combined}(iv) reproduces the convergence conclusion above.
\end{rem}
\begin{rem}\label{rem:minimalcoupling}
Under minimal coupling $\chi(\phi)\equiv1$, the interaction term vanishes, and Hypotheses \ref{H1}--\ref{H4} reduce to standard regularity and growth conditions on $V$ and $G$. In that setting Theorem~\ref{UnifiedTheorem} implies solutions enter the alternatives in \eqref{second-claim}, and Theorem~\ref{thm:combined} guarantees asymptotic convergence to $\mathbf p_*$ when $V$ has a strict minimum satisfying its hypotheses.
\end{rem}
\begin{cor}[Flat and vacuum limits]
If $G\equiv0$ and $\Lambda=0$ (flat FLRW, no cosmological constant), then Hypotheses \ref{H1}--\ref{H4} together with the hypotheses of Theorem~\ref{thm:combined} imply forward solutions with $H(t_0)>0$ converge to $\mathbf p_*$. In particular, in the vacuum case $\rho_m\equiv0$ the attractor persists and the scalar field relaxes to $\phi_*$ while $H^2\to V(\phi_*)/3$.
\end{cor}
\begin{proof}[Sketch of proof]
Set $G\equiv0$ and $\Lambda=0$. Then $\dot H=-\tfrac12(\gamma\rho_m+y^2)$ and the integrability/Barbalat argument yields $\rho_m,y\to0$. The LaSalle construction used in Theorem~\ref{thm:combined} provides a compact, positively invariant neighborhood of $\mathbf {p} _*$ and forces convergence to $\mathbf {p} _*$.
\end{proof}
\begin{rem}[On the role of the geometric term]
The geometric term $G(a)=\kappa^2 a^{-p}$ contributes an additional dissipative channel (through $\tfrac{p}{6}\kappa^2 a^{-p}$ in $\dot H$) that aids decay of kinetic and matter energy. When $G$ is nonzero, the equilibrium value of $H$ (in the absence of $\Lambda$) is shifted by the residual geometric energy; Hypothesis (iv) in Theorem~\ref{thm:combined} ensures $W$ serves as a Lyapunov quantity on the chosen invariant set.
\end{rem}
\begin{rem}[Practical verification of hypotheses]
Typical cosmological potentials (polynomial, suitably signed exponential, plateau potentials with a single minimum) satisfy Hypotheses \ref{H1}--\ref{H2} and conditions (i)–(iii) above. The geometric condition (Hypothesis \ref{H3} or (iv)) holds for the curvature and shear terms listed earlier (flat, open FLRW, Bianchi I). The coupling bound (Hypothesis \ref{H4}) is mild and holds for conformal couplings $\chi(\phi)=e^{\alpha\phi}$ with $\alpha$ bounded on the relevant $\phi$-range.
\end{rem}
Several works support the asymptotic picture described here. In particular, under Hypotheses \ref{H1}--\ref{H4} with $\chi\equiv1$, Leon and Franz‑Silva \cite{Leon:2019iwj} prove convergence to the alternatives in \eqref{second-claim} for a broad class of potentials and geometries. In the flat case $G\equiv0$ with $\Lambda=0$, Miritzis \cite[Prop.~3]{Miritzis:2003ym} establishes analogous long‑time behavior under standard energy conditions. The LaSalle and center‑manifold techniques used here follow classical treatments; see LaSalle \cite{LaSalle1968}, Wiggins \cite{Wiggins2003}, and Carr \cite{Carr1981}.

\section{Discussion and connection with the averaging framework}
\label{sec:discussion_averaging}

This section explains how the averaging framework developed in our work-in-progress \cite{Leon:2026vov} integrates with the global decay, finite-dimensional reduction, and persistence results established in this paper. For simplicity, we make $\Lambda=0$. We (i) present correspondence between the hypotheses of the averaging theorem contained in \cite{Leon:2026vov}, and the assumptions used in Sections~\ref{sec:decay}--\ref{sec:perturbations}, (ii) show how the quantitative averaging error is used to transfer integrability, LaSalle/Barbalat conclusions, and spectral‑gap persistence from the averaged model to the full oscillatory system, and (iii) state corollaries and a practical recipe for verifying the hypotheses in concrete cosmological models.
\subsection{Summary of the averaging result used}
\label{subsec:summary_averaging}

In \cite{Leon:2026vov}, we formulate an averaging theorem tailored to oscillatory scalar‑field cosmologies. In general, we use normalized variables of the form
\begin{equation}
\label{normalized-vars}
\Omega = \sqrt{\frac{\omega^2\phi^2+\dot\phi^2}{H^2}},\quad \Phi=\omega t-\arctan\!\Big(\frac{\omega\phi}{\dot\phi}\Big), \quad
\Omega_G=\frac{G(a)}{3 H^2},\quad \Omega_m=\frac{\rho_m}{3 H^2},
\end{equation}
the full system is expanded for small $H$ as
\begin{equation}
\label{eq:quasi_standard} 
\dot{\mathbf{x}}=\mathbf{f}^0(\mathbf{x},\theta)+H\,\mathbf{f}^1(\mathbf{x},\theta)+R(\mathbf{x},\theta,H),\quad
\dot H=f^{[2]}(\mathbf{x},\theta)\,H^2+S(\mathbf{x},\theta,H),
\end{equation}
with $\theta=\omega t$ and remainders $R,S$ higher order in $H$. Under hypotheses of smoothness/periodicity, controlled remainders, well‑defined Lipschitz averages, matched small initial data, asymptotic stability of the averaged frozen field, and a frequency–amplitude scaling that ties $\omega^{-1}$ to the small parameter $H$, \cite{Leon:2026vov} proves the key estimate
\begin{equation}
\|\mathbf{x}(t)-\bar{\mathbf{x}}(t)\|=\mathcal{O}(H(t))\quad(t\to\infty),
\end{equation}
where $\bar{\mathbf{x}}(t)$ solves the averaged system obtained by time averaging in $\theta$. The theorem also yields convergence of both the averaged and full trajectories to the same asymptotically stable equilibrium when the averaged field admits such an equilibrium.

The corresponding time averages are
\begin{equation}\label{timeavrg}
\bar{\mathbf{f}}^i(\mathbf{x})
:= \frac{1}{2\pi}\int_0^{2\pi}\mathbf{f}^i(\mathbf{x},\theta)\,d\theta,
\quad i=0,1,
\quad
\bar{f}^{[2]}(\mathbf{x})
:= \frac{1}{2\pi}\int_0^{2\pi} f^{[2]}(\mathbf{x},\theta)\,d\theta.
\end{equation}
\begin{thm}[Averaging for Scalar-Field Cosmologies {\cite[Thm.~3.1]{Leon:2026vov}}]
\label{thm:averaging_scalar_cosmology}
Let $H:[t_x,\infty) \to (0,\infty)$ be $C^1$, strictly decreasing, and satisfy $\lim_{t\to\infty} H(t) = 0$. Consider the system~\eqref{eq:quasi_standard} with $\theta = \omega t$ and $\mathbf{x}$ as in~\eqref{normalized-vars}. Assume: 
\begin{enumerate}[label=(A\arabic*), ref=A\arabic*]
\item \label{A1} \textbf{Smoothness and periodicity.}  
$\mathbf{f}^0$, $\mathbf{f}^1$, and $f^{[2]}$ are $C^1$ in $\mathbf{x}$ on an open set $U \subset \mathbb{R}^4$, and $2\pi$-periodic in $\theta$.
\item \label{A2} \textbf{Controlled remainders.}  
$R = \mathcal{O}(H^2)$ and $S = \mathcal{O}(H^3)$ uniformly on compact subsets of $U \times \mathbb{S}^1$.
\item \label{A3} \textbf{Well-defined averages.}  
The time averages defined in~\eqref{timeavrg} exist and are Lipschitz continuous on $U$.
\item \label{A4} \textbf{Matched initial data.}  
The full and averaged systems share initial data at $t = t_x$, with $H(t_x) = \varepsilon \ll 1$.
\item \label{A5} \textbf{Asymptotic stability.}  
The averaged system $\dot{\bar{\mathbf{x}}} = \bar{\mathbf{f}}^0(\bar{\mathbf{x}})$ admits an asymptotically stable equilibrium $\mathbf{x}_* \in U$.
\item \label{A6} \textbf{Frequency–amplitude scaling.}  
The frequency satisfies $\omega^{-1} = \mathcal{O}(\varepsilon) = \mathcal{O}(H(t_x))$, ensuring that boundary terms from integration by parts are absorbed into the remainders.
\end{enumerate}
Then the solutions $\mathbf{x}(t)$ and $\bar{\mathbf{x}}(t)$ with common initial data satisfy
\begin{equation}
\|\mathbf{x}(t) - \bar{\mathbf{x}}(t)\| = \mathcal{O}(H(t)) \quad \text{as } t \to \infty,
\label{error}
\end{equation}
and both converge to $\mathbf{x}_*$.
\end{thm}
\begin{proof}
Decompose $\mathbf{f}^i = \bar{\mathbf{f}}^i + \tilde{\mathbf{f}}^i$ for $i = 0,1$, where $\tilde{\mathbf{f}}^i$ has zero average in $\theta$ as in~\eqref{timeavrg}. Define $\bar{R}(\mathbf{x}, H) := \tfrac{1}{2\pi} \int_0^{2\pi} R(\mathbf{x}, \theta, H)\,d\theta = \mathcal{O}(H^2)$ by~\eqref{A2}.
Let $\bar{\mathbf{x}}(t)$ solve the averaged system
$\dot{\bar{\mathbf{x}}} = \bar{\mathbf{f}}^0(\bar{\mathbf{x}}) + H(t)\,\bar{\mathbf{f}}^1(\bar{\mathbf{x}}) + \bar{R}(\bar{\mathbf{x}}, H(t)),$
with initial data $\bar{\mathbf{x}}(t_x) = \mathbf{x}(t_x)$ by~\eqref{A4}, and  error $\mathbf{e}(t) := \mathbf{x}(t) - \bar{\mathbf{x}}(t)$. 

Subtracting the systems yields
\begin{equation}
\dot{\mathbf{e}} = A(t)\,\mathbf{e} + \tilde{\mathbf{f}}^0(\mathbf{x}, \omega t) + H(t)\,\tilde{\mathbf{f}}^1(\mathbf{x}, \omega t) + \Delta_R(t),
\end{equation}
where $A(t)$ is the Jacobian of the averaged vector field, with $\|A(t)\| \le L$ by~\eqref{A3}, and $\|\Delta_R(t)\| = \mathcal{O}(H^2) + \mathcal{O}(\|\mathbf{e}\| H)$.
By variation of constants,
\begin{equation}
\mathbf{e}(t) = \int_{t_x}^t \Phi(t,s)\left[\tilde{\mathbf{f}}^0(\mathbf{x}(s), \omega s) + H(s)\,\tilde{\mathbf{f}}^1(\mathbf{x}(s), \omega s) + \Delta_R(s)\right]\,ds,
\end{equation}
where $\Phi(t,s)$ is the evolution operator of the linearized system, satisfying $\|\Phi(t,s)\| \le e^{L(t-s)}$, and $\mathbf{e}(t_x) = 0$. The oscillatory integrals are estimated by integration by parts, using the $C^1$ regularity in $\theta$ from~\eqref{A1} and boundedness of trajectories. This yields convolution terms bounded by $C_1 \int_{t_x}^t e^{L(t-s)} H(s)\,ds$, plus boundary terms of order $H/\omega$, which are $\mathcal{O}(H^2)$ by~\eqref{A6} and absorbed into $\Delta_R$.
Thus, for some constant $C > 0$,
\begin{equation}
\|\mathbf{e}(t)\| \le C \left( \int_{t_x}^t e^{L(t-s)} H(s)\,ds + \int_{t_x}^t e^{L(t-s)} \|\mathbf{e}(s)\| H(s)\,ds \right).
\end{equation}

Choose $t_0 \ge t_x$ such that $C H(t_0) \le \tfrac{1}{2}L$, which is possible since $H(t) \to 0$. For $t \ge t_0$, monotonicity of $H$ implies
$\int_{t_0}^t e^{L(t-s)} H(s)\,ds \le H(t) \int_{t_0}^t e^{L(t-s)}\,ds \le \tfrac{1}{L} H(t),$
so
\begin{equation}
\sup_{s \in [t_0, t]} \|\mathbf{e}(s)\| \le \tfrac{C}{L} H(t) + \tfrac{C}{L} H(t) \sup_{s \in [t_0, t]} \|\mathbf{e}(s)\|.
\end{equation}
Choosing $t_0$ so that $\tfrac{C}{L} H(t_0) \le \tfrac{1}{2}$, we obtain
\begin{equation}
\sup_{s \in [t_0, t]} \|\mathbf{e}(s)\| \le \tfrac{2C}{L} H(t),
\end{equation}
\noindent
hence $\|\mathbf{e}(t)\| = \mathcal{O}(H(t))$ as $t \to \infty$. On the compact interval $[t_x, t_0]$, the error remains uniformly bounded by continuity and the smallness of $\varepsilon = H(t_x)$.

Finally, by~\eqref{A5}, the averaged trajectory $\bar{\mathbf{x}}(t)$ converges to $\mathbf{x}_*$, and since $\|\mathbf{x}(t) - \bar{\mathbf{x}}(t)\| = \mathcal{O}(H(t))$, it follows that $\mathbf{x}(t) \to \mathbf{x}_*$ as well.
\end{proof}
\begin{cor}[Degenerate Averaging Regime {\cite[Cor.~3.2]{Leon:2026vov}}]
\label{corollary}
Assume \eqref{A1}–\eqref{A6} and suppose $\mathbf{f}^0 \equiv 0$. If the model parameters determine $\omega$ such that $\omega^{-1} = \mathcal{O}(\varepsilon)$, then for common initial data at $t = t_x$, the solutions $\mathbf{x}(t)$ and $\bar{\mathbf{x}}(t)$ satisfy the estimate~\eqref{error} and both converge to the attracting equilibrium $\mathbf{x}_*$ of the frozen slow field $\dot{\bar{\mathbf{x}}} = \bar{\mathbf{f}}^1(\bar{\mathbf{x}})$.
\end{cor}
\begin{proof}
When $\mathbf{f}^0 \equiv 0$, the leading-order dynamics are governed by the averaged field $H \bar{\mathbf{f}}^1$. The same argument as in Theorem~\ref{thm:averaging_scalar_cosmology} applies, with the roles of $\bar{\mathbf{f}}^0$ and $\bar{\mathbf{f}}^1$ interchanged. The scaling condition $\omega^{-1} = \mathcal{O}(H(t_x))$ ensures that boundary terms from integration by parts remain $\mathcal{O}(H^2)$ and are absorbed into the remainders. The result follows.
\end{proof}

\subsection{Technical estimates: uniform derivative bounds}
\label{sec:uniform-bounds}

Here, we present a self‑contained proof of Proposition~\ref{prop:derivative_bounds}.
\begin{prop}[Uniform derivative bounds]
\label{prop:derivative_bounds}
Let $(\phi,y,\rho_m,a,H)$ be a forward solution of the constrained Friedmann system on $[t_0,\infty)$ satisfying the structural hypotheses of Section~\ref{sect:0}: $V,\chi, G\in C^2$ on an open set containing the forward invariant region $\Psi_+$, the matter equation of state satisfies $0<\gamma\le2$, and initial data lie in a compact subset $\mathcal K$ of the admissible phase space. Assume $H(t)>0$ for all $t\ge t_0$ and $\lim_{t\to\infty}H(t)=0$.

Then there exist $t_1\ge t_0$ and constants $C,C_1,C_2>0$, depending only on $\mathcal K$ and the $C^2$-norms of $V,\chi,G$ on a neighbourhood of $\mathcal K$, such that for all $t\ge t_1$
\begin{equation}
|\dot\phi(t)| + |\ddot\phi(t)| + |\dot\rho_m(t)| + |\dot H(t)| \le C,
\end{equation}
and, for any fixed frequency scale $\omega>0$ used to define the phase in \eqref{normalized-vars},
\begin{equation}
\Big|\frac{d}{dt}\Big(\omega t-\arctan\!\frac{\omega\phi(t)}{\dot\phi(t)}\Big)\Big|
\le C_1 + C_2\frac{|\ddot\phi(t)|+|\dot\phi(t)|}{\dot\phi(t)^2+(\omega\phi(t))^2},
\end{equation}
with the right‑hand side uniformly bounded on $[t_1,\infty)$ (interpreting the phase derivative by continuous extension at isolated zeros of $\dot\phi$). Consequently the normalized slow variables $\mathbf{x}$ in \eqref{normalized-vars} satisfy
\begin{equation}
\|\dot{\mathbf{x}}(t)\| \le C \quad\text{for all } t\ge t_1.
\end{equation}
\end{prop}

\begin{proof}
The proof is organized in five steps. Constants denoted by $C, C_i$ depend only on the compact set $\mathcal K$ and the $C^2$-norms of $V,\chi, G$ on a neighbourhood of $\mathcal K$; their values may change from line to line.

\noindent\textbf{Step 1. Uniform bounds on the state variables.}
By forward invariance, the trajectory remains in $\Psi_+\subset\mathcal K$. The Friedmann constraint (written schematically)
\begin{equation}
3H^2 = \tfrac12 y^2 + V(\phi) + \rho_m + G(a)
\end{equation}
(with $G(a)$ the geometric contribution) implies uniform bounds on $y^2$, $V(\phi)$, $\rho_m$, and $\mathcal G(a)$ on $\Psi_+$. Since $V$ has compact sublevel sets on the relevant domain, $|\phi(t)|$ is uniformly bounded. The scale factor $a(t)$ likewise remains in a compact subset of $(0,\infty)$ for forward times considered. Hence there exists $C_0>0$ with
\begin{equation}
|\phi(t)| + |y(t)| + \rho_m(t) + |a(t)| \le C_0 \quad (t\ge t_0).
\end{equation}

\noindent\textbf{Step 2. Uniform bounds on first derivatives.}
Write the scalar field and matter equations in first‑order form
\begin{equation}
\dot\phi = y,
\quad
\dot y = -3H y - V'(\phi) + \mathcal C(\phi,\rho_m,a),
\end{equation}
\begin{equation}
\dot\rho_m = -3\gamma H \rho_m + \mathcal M(\phi,y,\rho_m,a),
\end{equation}
where $\mathcal C,\mathcal M$ collect coupling and geometric correction terms; by hypothesis these are $C^1$ and bounded on $\Psi_+$. Using the uniform bounds from Step 1 and boundedness of $V'$, $\chi$ and geometric coefficients on $\mathcal K$, we obtain
\begin{equation}
|\dot\phi(t)| = |y(t)| \le C_0,
\quad
|\dot\rho_m(t)| \le 3\gamma H(t)\rho_m(t) + |\mathcal M(\cdot)| \le C_1,
\end{equation}
and
\begin{equation}
|\dot y(t)| \le 3H(t)|y(t)| + |V'(\phi(t))| + |\mathcal C(\cdot)| \le C_2.
\end{equation}
Choose $t_1\ge t_0$ large enough that the orbit has entered the small‑$H$ regime used in subsequent estimates; then the above bounds hold uniformly for $t\ge t_1$.

\noindent\textbf{Step 3. Uniform bounds on $\dot H$.}
Differentiate the Friedmann constraint or use the Raychaudhuri equation to express $\dot H$ in terms of the state variables
\begin{equation}
\dot H = -\tfrac12 y^2 - \tfrac12\gamma\rho_m + \mathcal R(\phi,y,\rho_m,a),
\end{equation}
with $\mathcal R$ bounded on $\Psi_+$. Using Step 1 bounds we obtain $|\dot H(t)|\le C_3$ for $t\ge t_1$.

\noindent\textbf{Step 4. Uniform bounds on second derivatives.}
Differentiate the $\dot y$ equation to obtain
\begin{equation}
\ddot y = -3\dot H\,y -3H\dot y - V''(\phi)\dot\phi + \partial_t\mathcal C(\phi,y,\rho_m,a).
\end{equation}
Each term on the right is controlled: $y,\dot y,\dot\phi$ are bounded by Step 2; $V''$ and derivatives of $\mathcal C$ are bounded on $\mathcal K$; and $\dot H$ is bounded by Step 3. Hence $|\ddot y(t)|\le C_4$ for $t\ge t_1$. Since $\ddot\phi=\dot y$, this yields the claimed bound on $|\ddot\phi|$.

\noindent\textbf{Step 5. Control of the oscillatory phase derivative and Lipschitz bound for $\mathbf{x}$.}
Define the phase
\begin{equation}
\Phi(t)=\omega t - \arctan\!\frac{\omega\phi(t)}{\dot\phi(t)}.
\end{equation}
A direct differentiation (quotient rule) gives
\begin{equation}
\dot\Phi(t)
= \omega - \frac{\omega\big(\dot\phi^2 - \omega\phi\ddot\phi\big)}{\dot\phi^2+(\omega\phi)^2}.
\end{equation}
Using the uniform bounds on $|\phi|,|\dot\phi|,|\ddot\phi|$ from Steps 1–4, the numerator and denominator are uniformly controlled on $[t_1,\infty)$. If $\dot\phi$ vanishes at isolated instants, interpret $\dot\Phi$ by continuous extension; the denominator $\dot\phi^2+(\omega\phi)^2$ prevents singular behaviour except at genuine degeneracies excluded by the forward invariant compactness assumption. Therefore there exist $C_1,C_2>0$ with
\begin{equation}
|\dot\Phi(t)| \le C_1 + C_2\frac{|\ddot\phi(t)|+|\dot\phi(t)|}{\dot\phi(t)^2+(\omega\phi(t))^2},
\end{equation}
and the right‑hand side is uniformly bounded for $t\ge t_1$.

Finally, the normalized variables $\mathbf{x}=(\Omega,\Sigma,\Omega_k,\Phi)^T$ are smooth functions of $(\phi,y,\rho_m,a,H)$ away from degenerate points. Differentiating each coordinate produces combinations of $\dot\phi,\ddot\phi,\dot\rho_m,\dot H$ and algebraic factors of $\phi,y,\rho_m, a, H$; by the uniform bounds established above, these derivatives are uniformly bounded on $[t_1,\infty)$. Hence there exists $C>0$ such that
\begin{equation}
\|\dot{\mathbf{x}}(t)\| \le C \quad (t\ge t_1),
\end{equation}
which completes the proof.
\end{proof}

\subsection{Correspondence between the averaging hypotheses and our assumptions}
\label{subsec:mapping}

Each averaging hypothesis in \cite{Leon:2026vov} has a direct counterpart among the conditions used in this paper.
\begin{itemize}
\item \textbf{Smoothness and periodicity.} The $C^1$ regularity and $2\pi$-periodicity in the fast phase assumed in \cite{Leon:2026vov} correspond to the regularity hypotheses invoked in Proposition~\ref{prop:derivative_bounds} and in the center/manifold construction (Section~\ref{sec:center}). They guarantee the existence of uniform derivative bounds and of Lipschitz constants for averaged vector fields.
\item \textbf{Controlled remainders.} The small‑$H$ expansion $R=\mathcal{O}(H^2)$, $S=\mathcal{O}(H^3)$ matches the asymptotic ordering used in our decay analysis and ensures that averaged dynamics capture the leading slow behaviour while remainders are negligible in the small‑$H$ regime exploited by LaSalle/Barbalat.
\item \textbf{Well‑defined averages.} Lipschitz continuity of the averaged fields supplies the uniform Jacobian bound $L$ used in variation‑of‑constants, and Gronwall estimates that control the averaging error.
\item \textbf{Matched initial data and frequency scaling.} Choosing initial data with $H(t_x)=\varepsilon\ll1$ and imposing $\omega^{-1}=\mathcal{O}(\varepsilon)$ ensures boundary terms arising in integration by parts are absorbed into higher‑order remainders; this is the same smallness regime in which the dissipative integrability and LaSalle region $\Psi_+$ are simultaneously valid.
\item \textbf{Averaged asymptotic stability.} An asymptotically stable equilibrium of the averaged frozen field plays the role of the nondegenerate minimum $V''(\phi_*)>0$ and the uniqueness/compactness conditions used in Theorem~\ref{thm:positive_curvature} and Theorem~\ref{UnifiedTheorem}.
\end{itemize}

When these conditions hold, the averaging theorem in \cite{Leon:2026vov} and the global/persistence results of this paper are fully compatible.

\subsection{How the averaging error is used in the global analysis}
\label{sec:avg-error-global}

The estimate $\|\mathbf{x}-\bar{\mathbf{x}}\|=\mathcal{O}(H)$ allows one to apply the conclusions of the averaged model to the full oscillatory system in three steps, which will be outlined:

\subsubsection{Integrability and Barbalat.} 
\label{sec:barbalat}
If the averaged dissipation $\overline{f}(t)$ is integrable (as shown for the averaged system), then
\begin{equation}
f(t)=\overline{f}(t)+\mathcal{O}(H(t)),
\end{equation}
so $f$ remains integrable because $H\in L^1_{\mathrm{loc}}$ in the regimes considered and the $\mathcal{O}(H)$ term is negligible for large times. Uniform derivative bounds persist under the small additive error, hence Barbalat's lemma applies to the full system and yields $f(t)\to0$ and the $O(1/t)$ transient established in Proposition~\ref{prop:decay}.

\subsubsection{LaSalle invariance and compactness.} 
\label{sec:lasalle}
Averaging shows that the averaged invariant region $\bar\Psi_+$ and the true forward invariant set $\Psi_+$ differ by an $\mathcal{O}(H)$ tubular neighbourhood for sufficiently small $H$. Since the $\omega$-limit of the averaged flow lies at the averaged equilibrium $\mathbf{x}_*$, the error bound implies the $\omega$-limit of the full flow is contained in an arbitrarily small neighbourhood of $\mathbf{x}_*$ for large times; compactness and uniqueness of the limit then force convergence of the full trajectory to the same equilibrium.

\subsubsection{Spectral gap and manifold persistence.} 
\label{sec:spectral-gap}
The linearization of the averaged reduced field at $\mathbf{x}_*$ determines the spectral gap used to obtain exponential convergence on $\mathcal M_{\mathrm{loc}}$. The $\mathcal{O}(H)$ closeness of the full and averaged vector fields implies the linearizations differ by a small operator perturbation; standard perturbation theory yields continuous dependence of eigenvalues and invariant subspaces on the perturbation, hence persistence of the spectral gap, of the local invariant manifold, and of exponential rates for the full oscillatory system.

\subsection{Combined corollaries}
\label{subsec:combined_corollaries}
Combining the averaging theorem of \cite{Leon:2026vov} with Theorems~\ref{UnifiedTheorem} and \ref{thm:combined} yields the following corollaries (informally stated; precise versions follow by conjunction of the respective hypotheses).

\begin{cor}[Averaged LaSalle convergence]
Suppose the hypotheses of Theorem~\ref{UnifiedTheorem} (or Theorem~\ref{thm:positive_curvature}) hold together with the averaging hypotheses in \cite{Leon:2026vov}. Then, for initial data with $H(t_x)=\varepsilon\ll1$ the full oscillatory system satisfies the decay conclusions of Proposition~\ref{prop:decay} and converges to the same equilibrium $\mathbf p_*$ determined by the averaged dynamics.
\end{cor}
\begin{cor}[Averaged persistence and exponential regime]
Under the same combined hypotheses, if the averaged equilibrium $\mathbf{x}_*$ corresponds to a nondegenerate minimum (spectral gap) then there exists a local invariant manifold for the full oscillatory system and solutions starting sufficiently close converge exponentially to $\mathbf p_*$; the exponential rate and manifold depend continuously on the small parameters (averaging parameter and perturbation size).
\end{cor}

To apply the combined averaging and global framework in a concrete model, follow these steps in order.
\begin{enumerate}
\item \textbf{Structural checks.} Verify compactness and uniqueness conditions for sublevel sets of $V$, and confirm the sign/monotonicity properties of the geometric term $G$ used in Sections~\ref{sec:decay}--\ref{sec:perturbations}. Ensure $V,\chi, G$ have the required regularity on the forward invariant region $\Psi_+$.
\item \textbf{Normalization and expansion.} Introduce normalized variables $\mathbf{x}$ (cf.\ \eqref{normalized-vars}) and expand the vector field in the small‑$H$ form
\begin{equation}
\dot{\mathbf{x}}=\mathbf{f}^0(\mathbf{x},\theta)+H\,\mathbf{f}^1(\mathbf{x},\theta)+R(\mathbf{x},\theta,H),\quad
\dot H=f^{[2]}(\mathbf{x},\theta)\,H^2+S(\mathbf{x},\theta,H).
\end{equation}
Identify $\mathbf{f}^0,\mathbf{f}^1,f^{[2]}$ and compute the orders of the remainders $R,S$.
\item \textbf{Check the Averaging hypotheses.} Periodicity in $\theta$, remainder orders $R=\mathcal{O}(H^2)$, $S=\mathcal{O}(H^3)$, Lipschitz continuity of the time averages, matched small initial data $H(t_x)=\varepsilon\ll1$, and the frequency–amplitude scaling $\omega^{-1}=\mathcal{O}(H(t_x))$.
\item \textbf{Averaged stability.} Compute the averaged vector field $\bar{\mathbf{f}}^0$ (time average in $\theta$) and verify asymptotic stability of its equilibrium $\mathbf{x}_*$ by linearization and spectral gap estimates.
\item \textbf{Apply averaging and global theorems.} Use the averaging theorem in \cite{Leon:2026vov} to obtain the $\mathcal{O}(H)$ error bound between full and averaged trajectories, then invoke the global decay and persistence theorems of this paper to conclude convergence and, when the spectral gap holds, exponential rates.
\end{enumerate}
\subsection{Limitations and caveats}
\label{subsec:limitations_remarks}
\begin{itemize}
\item \textbf{Scaling requirement.} The frequency–amplitude scaling is essential: if $\omega$ does not scale with the small parameter $H$, boundary terms from integration by parts need not be negligible and a resonant normal‑form treatment is required.
\item \textbf{Resonances.} Strong resonances between the fast scalar oscillations and the background expansion can change effective dissipation; resonant averaging or normal‑form reductions must be performed before applying persistence arguments.
\item \textbf{Small‑$H$ regime.} The combined approach is most effective when $H$ decays sufficiently fast. If $H$ decays too slowly (or does not decay), the averaging window may be too short to control the error uniformly in time.
\item \textbf{Degenerate minima.} If the potential minimum is degenerate (e.g., $V''(\phi_*)=0$), the spectral gap may vanish and exponential convergence need not hold; higher‑order normal forms or blow‑up analyses are then necessary.
\end{itemize}

\section{Application to the quadratic potential}
\label{sec:quad}
To illustrate the procedures, consider the quadratic potential
\begin{equation}
V(\phi)=\tfrac12 m^2\phi^2,\quad m>0.
\end{equation}
Its derivatives are
\begin{equation}
V'(\phi)=m^2\phi,\quad V''(\phi)=m^2>0,
\end{equation}
so the unique minimum at $\phi_*=0$ is nondegenerate, and the local linearization has mass $m$. 
Sublevel sets $\{V\le E\}$ are compact for any finite energy $E$. 
Since all derivatives are smooth, Proposition~\ref{prop:derivative_bounds} applies. 
For
\begin{equation}
\chi(\phi)=1+\alpha\phi,\quad |\alpha|\ll1,
\end{equation}
the remainder estimates in (A2) follow directly from Taylor expansion.

\subsection{Normalized variables and expansion.}  
Using \eqref{normalized-vars} with a chosen frequency scale $\omega$ (natural choice $\omega=m$ for near‑harmonic motion), the normalized energy variable $\Omega$ and phase $\Phi$ capture the fast oscillation. 

The vector field $ (H,\, \Phi,\,\Omega,\,\Omega_G,\,\Omega_m)^{\top} $ evolve in   time according to 

\begin{align}
\dot{H} &= -\frac{1}{2} H^2 \left(3 \gamma  \Omega_{m}+p \Omega_G+3 \Omega ^2 \cos (2 (\Phi -t \omega ))+3 \Omega ^2\right),\\
\dot{\Phi} &= \frac{\sin
   (\Phi -t \omega ) \left(\frac{\sqrt{6} \alpha  (3 \gamma -4) H \omega ^2 \Omega_{m}}{\Omega  \left(\omega -\sqrt{6} \alpha  H \Omega  \sin (\Phi -t
   \omega )\right)}+12 H \omega  \cos (\Phi -t \omega )+4 (\omega^2 -m^2)  \sin (\Phi -t \omega )\right)}{4 \omega },\\
\dot{\Omega} &= \frac{1}{12} \Bigg[6 H \Omega 
   \left(3 \gamma  \Omega_{m}+p \Omega_G+3 \Omega ^2 \cos (2 (\Phi -t \omega ))+3 \Omega ^2\right)\nonumber\\
   & -\frac{3 \sqrt{6} \alpha  (3 \gamma
   -4) H \omega  \Omega_{m} \cos (\Phi -t \omega )}{\omega -\sqrt{6} \alpha  H \Omega  \sin (\Phi -t \omega )}-36 H \Omega  \cos ^2(\Phi -t \omega
   )+\frac{6 \Omega  (m^2-\omega^2) \sin (2 (\Phi -t \omega ))}{\omega }\Bigg],\\
 \dot{\Omega_G} &= H \Omega_G \left(3 \gamma  \Omega_{m}+p
   (\Omega_G-1)+6 \Omega ^2 \cos ^2(\Phi -t \omega )\right),\\
\dot{\Omega_m} &= H \Omega_{m} (3 \gamma  (\Omega_{m}-1)+p \Omega_G) \nonumber\\
& +\frac{1}{2} H \Omega 
   \Omega_{m} \cos (\Phi -t \omega ) \left(\frac{\sqrt{6} \alpha  (3 \gamma -4) \omega }{\omega -\sqrt{6} \alpha  H \Omega  \sin (\Phi -t \omega )}+12
   \Omega  \cos (\Phi -t \omega )\right).
\end{align}
Expanding the full constrained system for small $H$ yields the schematic form
\begin{equation}
\dot{\mathbf{x}}=\mathbf{f}^0(\mathbf{x},\theta)+H\,\mathbf{f}^1(\mathbf{x},\theta)+\mathcal{O}(H^2),
\quad
\dot H = f^{[2]}(\mathbf{x},\theta)\,H^2+\mathcal{O}(H^3),
\end{equation}
with $\theta=\omega t$.

Let $\mathbf{x}=(\Phi,\,\Omega,\,\Omega_G,\,\Omega_m)^{\top}$, be the state vector for the dimensionless variables defined in \eqref{normalized-vars}. Then $\mathbf{f}^{0}(\mathbf{x},t)=\mathbf{0}$ and
\begin{equation}
    \mathbf{f}^{1}(\mathbf{x},\theta)= \left(
\begin{array}{c}
  \frac{1}{4} \left[\sqrt{6} \alpha  (4-3 \gamma ) \Omega _m\Omega^{-1} \sin (\theta -\Phi )-6 \sin (2 \theta -2\Phi )\right]\\
\frac{1}{4} \{2 p \Omega  \Omega _G+\Omega _m \left[6 \gamma  \Omega +\sqrt{6} \alpha  (4-3 \gamma ) \cos (\theta-\Phi )\right]+12 \Omega  \left(\Omega ^2-1\right) \cos
   ^2(\theta-\Phi )\}\\
 \Omega _G \left[p (\Omega _G-1)+3 \gamma  \Omega _m+6 \Omega ^2 \cos ^2(\theta-\Phi )\right]\\
\Omega _m \left[-3 \gamma + p \Omega _G+3 \gamma  \Omega _m+\Omega  \cos (\theta-\Phi ) \left(\sqrt{6} \alpha  (\frac{3}{2} \gamma -2)+6 \Omega  \cos (\theta-\Phi
   )\right)\right]
\end{array}
\right)
\end{equation}
where $\theta= m t $.
We also get
\begin{equation}
    f^{[2]}(\mathbf{x},\theta)=-\frac{1}{2} \left[p \Omega _G+3 \gamma  \Omega _m+6 \Omega ^2 \cos ^2(\theta-\Phi)\right].
\end{equation}

In this representation, the Friedmann restriction takes the form
\begin{equation}
    \Omega^2+\Omega_m+\Omega_G=1.
\end{equation}
The Friedmann condition can be used to reduce the problem's dimensionality by one. Furthermore, after applying the averaging procedure described in section \ref{sec:discussion_averaging}, we obtain
\begin{equation}
  \left(\begin{array}{c}
    {\overline\Omega}^{\prime}\\
    {{\overline\Omega}_G}^{\prime}\\
    {\overline{\Omega}_m}^{\prime}
    \end{array}\right):= \frac{1}{H} \left(\begin{array}{c}
    \dot{\overline\Omega}\\
    \dot{\overline\Omega}_G\\
    \dot{\overline\Omega}_m
    \end{array}\right)=   \left(\begin{array}{c}
    \frac{1}{2} \overline\Omega  \left(p \overline\Omega _G+3 \gamma  \overline\Omega _m+3 \overline\Omega ^2-3\right)\\
     \overline\Omega _G \left[p \left( \overline\Omega _G-1\right)+3 \left(\gamma  \overline\Omega _m+\overline\Omega ^2\right)\right]\\
    \overline\Omega _m \left[p \overline\Omega _G+3 \gamma  \left(\overline\Omega _m-1\right)+3 \overline\Omega ^2\right]
    \end{array}\right).
\end{equation}
Notice that, in the previous system of equations, the dynamics for the state variable $\Phi$ are absent. This follows from the fact that, after choosing $\omega=m$, we obtain $\dot{\overline\Phi}=0$.

Using the Friedmann equation as a definition of $\overline\Omega_G$, we obtain the two-dimensional system
\begin{equation}
\label{averaging}
  \left(\begin{array}{c}
    {\overline\Omega}^{\prime}\\
    {\overline{\Omega}_m}^{\prime}
    \end{array}\right)=   \left(\begin{array}{c}
    \frac{1}{2} \overline\Omega  \left[p \left(1- \overline \Omega^2-\overline \Omega_m\right)+3 \gamma  \overline\Omega _m+3 (\overline\Omega ^2-1)\right]\\
    \overline\Omega _m \left[p \left(1- \overline \Omega^2-\overline \Omega_m\right)+3 \gamma  \left(\overline\Omega _m-1\right)+3 \overline\Omega ^2\right]
    \end{array}\right).
\end{equation}

\subsubsection{Equilibrium points and stability}

To understand the local phase-space structure of the system \eqref{averaging}, we perform a stability analysis within the framework of qualitative dynamical systems theory. This requires identifying the critical points and analyzing the eigenvalues of the Jacobian matrix of the linearized flow around each of them.

Let $u=\overline\Omega$ and $v=\overline\Omega_m$. Then the system \eqref{averaging} can be written as
\begin{equation}
u'=\tfrac12 u\,A(u,v),\qquad 
v'=v\,B(u,v),    
\end{equation}
where
\begin{equation}
A(u,v)=-(p-3)u^2+(3\gamma-p)v+(p-3), \qquad
B(u,v)=A(u,v)+3(1-\gamma).
\end{equation}

The equilibrium (or critical) points are obtained by imposing
\begin{equation}
    u'=0, \qquad v'=0.
\end{equation}
Since the right-hand sides are factorized, these conditions are equivalent to
\begin{equation}
    u=0 \ \text{or}\ A(u,v)=0, \qquad v=0 \ \text{or}\ B(u,v)=0.
\end{equation}
Therefore, the critical points must arise from the possible simultaneous combinations:
\begin{equation}
(u=0,\;v=0), \qquad (u=0,\;B=0), \qquad (v=0,\;A=0),
\end{equation}
together with any eventual solutions of
\begin{equation}
A(u,v)=0, \qquad B(u,v)=0,
\end{equation}
which only occurs if $\gamma=1$.

For generic parameters $p\neq 3$, $p\neq 3\gamma$, and $\gamma\neq 1$, the system admits the following isolated equilibrium points:
\begin{equation}
    (0,0), \qquad (0,1), \qquad (1,0), \qquad (-1,0).
\end{equation}

The Jacobian matrix of the system is
\begin{equation}
    J(u,v)=
\frac{1}{2}\begin{pmatrix}
3(3-p)u^2+(3\gamma-p)v+(p-3) & (3\gamma-p)u \\[6pt]
-4(p-3)uv & 2(3-p)u^2-4(p-3\gamma)v+2(p-3\gamma)
\end{pmatrix}.
\end{equation}

We now evaluate the Jacobian at each equilibrium point:

\begin{enumerate}
\item \textbf{Point $(0,0)$}  
\begin{equation}
J(0,0)=\frac{1}{2}\begin{pmatrix} p-3 & 0 \\ 0 & 2(p-3\gamma) \end{pmatrix},
\end{equation}
with eigenvalues
\begin{equation}
\lambda_1=\tfrac12(p-3), \qquad \lambda_2=p-3\gamma.
\end{equation}
Classification:
\begin{itemize}
\item sink if $p<3$ and $p<3\gamma$,
\item source if $p>3$ and $p>3\gamma$,
\item saddle otherwise.
\end{itemize}

\item \textbf{Point $(0,1)$}  
\begin{equation}
J(0,1)=\frac{1}{2}\begin{pmatrix} 3(\gamma-1) & 0 \\ 0 & 2(3\gamma-p) \end{pmatrix},
\end{equation}
with eigenvalues
\begin{equation}
\lambda_1=\tfrac{3}{2}(\gamma-1), \qquad \lambda_2=(3\gamma-p).
\end{equation}
Classification:
\begin{itemize}
\item sink if $\gamma<1$ and $p>3\gamma$,
\item source if $\gamma>1$ and $p<3\gamma$,
\item saddle otherwise.
\end{itemize}

\item \textbf{Point $(1,0)$}  
\begin{equation}
J(1,0)= \begin{pmatrix} 3-p & 3\gamma-p \\ 0 & 3(1-\gamma) \end{pmatrix},
\end{equation}
with eigenvalues
\begin{equation}
\lambda_1=3-p, \qquad \lambda_2=3(1-\gamma).
\end{equation}
Classification:
\begin{itemize}
\item sink if $p>3$ and $\gamma>1$,
\item source if $p<3$ and $\gamma<1$,
\item saddle otherwise.
\end{itemize}

\item \textbf{Point $(-1,0)$}  
The Jacobian has the same eigenvalues as at $(1,0)$, so the stability classification is identical.
\end{enumerate}

The origin $(0,0)$ corresponds to negligible scalar and matter contributions; its stability depends on whether the effective coupling $p$ lies below or above the critical values $3$ and $3\gamma$.  
The point $(0,1)$ represents a pure matter state, while $(\pm 1,0)$ corresponds to scalar–dominated configurations.  
The signs of the eigenvalues determine which component governs the late–time dynamics: matter, scalar field, or a transient saddle regime.

\subsubsection{Special parameter cases.}
In addition to the four isolated equilibria obtained for generic values of 
$p\neq 3$, $p\neq 3\gamma$, and $\gamma\neq 1$, the system 
\eqref{averaging} develops continuous families of equilibrium points when the
parameters take special values:

\begin{enumerate}

\item[\textbf{(i)}] \textbf{Case $p = 3$.}  
For $v=0$ one has
\begin{equation}
    A(u,0)=-(p-3)u^2+(p-3),
\end{equation}
so when $p=3$ the condition $A(u,0)=0$ holds for all $u$.   
Since $v'=vB(u,v)$ vanishes identically for $v=0$, the entire line 
\begin{equation}
\{(u,0): u\in\mathbb{R}\}
\end{equation}
consists of equilibrium points.

\item[\textbf{(ii)}] \textbf{Case $p = 3\gamma$.}  
For $u=0$ one finds
\begin{equation}
    B(0,v)=(3\gamma-p)v+(p-3).
\end{equation}
If $p=3\gamma$ and simultaneously $\gamma=1$, then $B(0,v)\equiv 0$ and the entire line 
\begin{equation}
\{(0,v): v\in\mathbb{R}\}
\end{equation}
is equilibrium. Otherwise, only isolated points remain.

\item[\textbf{(iii)}] \textbf{Case $\gamma = 1$.}  The system is reduced to 
\begin{align}
 u'   & = \frac{1}{2} (p-3) u \left(1-u^2-v\right), \label{eq165}\\
 v'   & =(p-3) v \left(1-u^2-v\right) \label{eq166}. 
\end{align}
In this case $B(u,v)=A(u,v)$, so the equilibrium conditions reduce to
\begin{equation}
   uA(u,v)=0, \qquad vA(u,v)=0.
\end{equation}
Thus any point with $u\neq 0$, $v\neq 0$, and $A(u,v)=0$
is an equilibrium. This yields a \emph{curve} of equilibria in the $(u,v)$–plane, given implicitly by
\begin{equation}
    (p-3)(1-u^2-v)=0.
\end{equation}
Hence:
\begin{itemize}
\item If $p=3$, then $A(u,v)\equiv 0$ and every point is an equilibrium.
\item If $p\neq 3$, the equilibria form the parabola $v = 1 - u^2$, 
which corresponds to the invariant surface 
\begin{equation}
\overline\Omega^2+\overline\Omega_m=1, \qquad \overline\Omega_G=0.
\end{equation}
Along this curve, the Jacobian has one zero eigenvalue (reflecting the continuous family of equilibria) and one transverse eigenvalue $\lambda = p-3$.
Therefore, 
\begin{itemize}
\item for $p>3$ the parabola is transversely stable (attracting),
\item for $p<3$ the parabola is transversely unstable (repelling),
\item for $p=3$ the system is fully degenerate, with every point an equilibrium.
\end{itemize}
\end{itemize}
\end{enumerate}
As illustrated in Figure~\ref{fig:phase-plane}, the phase portrait provides a geometric visualization of the dynamical system \eqref{eq165}--\eqref{eq166}. The trajectories capture the qualitative behavior of solutions, while equilibrium sets and invariant manifolds delineate the physically admissible region of phase space. 

\begin{figure}[H]
    \centering
    \includegraphics[width=0.8\textwidth]{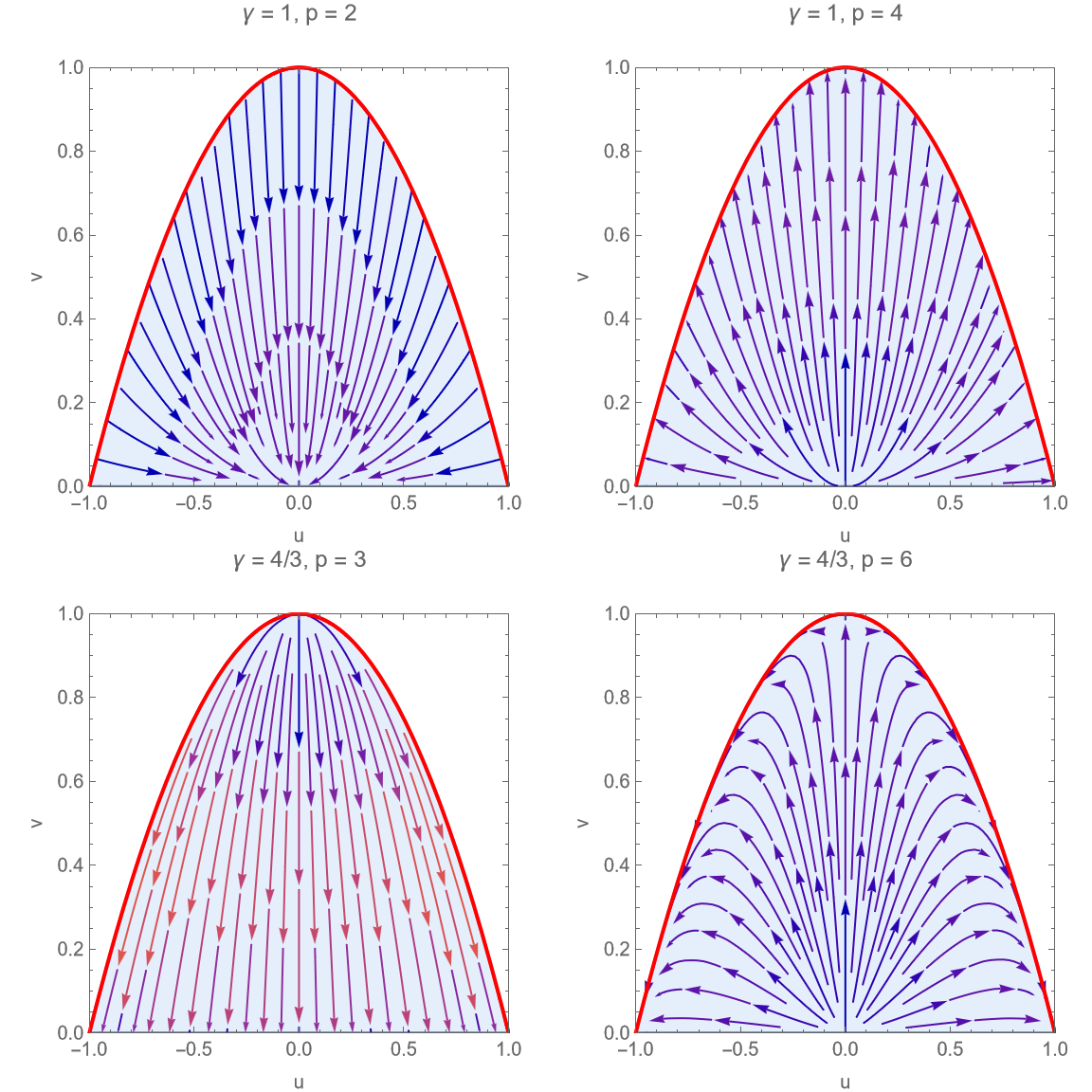}
    \caption{Phase plane of the dynamical system \eqref{eq165}--\eqref{eq166}. 
    The physically relevant region is bounded by the parabola $v=1-u^2$ and the positive $v$ axis. 
    Streamlines indicate the flow of $(u,v)$, highlighting the role of the equilibrium curve and axes as invariant sets.}
    \label{fig:phase-plane}
\end{figure}
For concreteness, the top row corresponds to a dust background ($\gamma=1$), with $p=2$ representing a cosmic–string–like fluid and $p=4$ a radiation–like component $\rho_G$. 
The bottom row corresponds to a radiation background ($\gamma=4/3$), with $p=3$ representing a dust–like fluid and $p=6$ a stiff–like component $\rho_G$.

\subsection{Averaged dissipation.}  
We can express the energy density of the scalar field as $E=\tfrac12\dot\phi^2+\tfrac12 m^2\phi^2$. For small $H$ and weak coupling $\chi(\phi)=1+\alpha\phi+\mathcal{O}(\alpha^2)$, we
have 
\begin{equation}
    \dot E/E = -\Gamma(\mathbf{x}) H + \mathcal{O}(H^2),
\end{equation}
where
\begin{equation}
 \Gamma(\mathbf{x}) =  \frac{\cos (m t-\Phi ) \left(3 \sqrt{6} \alpha  \gamma   \Omega_{m}-4 \sqrt{6} \alpha  \Omega_{m}+12 \Omega  \cos (m t-\Phi )\right)}{2 \Omega }.
\end{equation}
The averaged energy loss per fast period is
\begin{equation}
\dot{\overline{E}} = -\Gamma(\bar{\mathbf{x}}) \, H \,\overline{E} + \mathcal{O}(H^2),
\end{equation}
with the leading coefficient
\begin{equation}
\Gamma(\bar{\mathbf{x}})=3.
\end{equation}
In particular, $\Gamma>0$ for dissipative couplings, leading to $\overline{E}=\overline{E}_0 a^{-3}$.

The fact that $\Gamma(\bar{\mathbf{x}})=3$ demonstrates that, to first order in $H$, the scalar field loses energy exactly as a pressureless component, independently of the microscopic details of the coupling. This reinforces the interpretation of coherent scalar oscillations as an effective fluid of matter and shows that weak dissipative interactions do not destabilize this correspondence. Moreover, the positivity of $\Gamma$ ensures that the averaged dynamics remain monotonic and integrable, guaranteeing that the energy density redshifts as $a^{-3}$ throughout the slow-damping regime.

\subsection{Numerical Integration}

To analyze the dynamical behaviour of the model, we solved both the full
(non-averaged) system and its corresponding averaged system for a set of
initial conditions satisfying the Friedmann constraint
$
\Omega^{2} + \Omega_{m} + \Omega_{G} = 1.
$
The full system consists of five coupled nonlinear differential equations
for the variables
$
(H,\,\Phi,\,\Omega,\,\Omega_{G},\,\Omega_{m}),
$
and contains fast oscillatory terms driven by the phase
$
\theta = \omega t.
$
The averaged system reduces the dynamics to the slow variables
$
(\Omega,\,\Omega_{m},\,H),
$
after removing the rapid oscillations associated with the combination
$\Phi - \theta$.

Both systems were integrated forward and backward in time over the interval
$[-T,\,T]$ with $T = 10^{4}$.
We employed the \texttt{solve\_ivp} routine from the SciPy library, which
implements adaptive step-size control suitable for stiff and multi-scale
systems.
A dense temporal sampling (up to $2\times 10^{5}$ evaluation points) was
used to ensure accurate resolution of the fast oscillations in the full
system and to allow a direct comparison with the averaged dynamics.

The numerical solutions were represented in two complementary phase spaces.
The first is the two-dimensional plane $(\Omega,\,\Omega_{m})$, where the
physically allowed region is given by
$
\Omega^{2} + \Omega_{m} \le 1.
$
The second is the three-dimensional space
$(\Omega,\,\Omega_{m},\,H)$, which illustrates the full dynamical evolution
of the slow variables.
Trajectories of the full system are plotted as solid curves, while those of the averaged system are shown as dashed curves.
Star markers indicate initial conditions. Table~\ref{tab:ICs} lists the initial values of $(H_{0},\Omega_{0},\Omega_{m0})$ used in the simulations, together with the corresponding $\Omega_{G0}$ obtained from the Friedmann constraint.

\textbf{\begin{table}[h!]
\centering
\caption{\label{tab:ICs}  Initial conditions used in the numerical integrations. 
All sets satisfy the Friedmann constraint 
$\Omega^{2} + \Omega_{m} + \Omega_{G} = 1$.}
\begin{tabular}{c c c c c}
\hline
\hline
Case & $H_{0}$ & $\Omega_{0}$ & $\Omega_{m0}$ & $\Omega_{G0}=1-\Omega_{0}^{2}-\Omega_{m0}$ \\
\hline
1 & $10^{-3}$ & $0.80$ & $0.10$ & $0.26$ \\
2 & $10^{-3}$ & $0.70$ & $0.20$ & $0.31$ \\
3 & $10^{-3}$ & $0.60$ & $0.30$ & $0.34$ \\
4 & $10^{-3}$ & $0.50$ & $0.40$ & $0.35$ \\
5 & $10^{-3}$ & $0.90$ & $0.05$ & $0.14$ \\
\hline
\hline
\end{tabular}
\end{table}}

Figures~\ref{fig:phase_space_p1}--\ref{fig:phase_space_p2} summarise the behaviour of the full and averaged dynamical systems for the representative values $p=1$ and $p=2$. Each figure contains a two-dimensional projection in the $(\Omega,\Omega_{m})$ plane and a three-dimensional representation in $(\Omega,\Omega_{m},H)$. 
The solutions of the full system exhibit high-frequency oscillations
superimposed on a slow evolution. One can easily compare the slow dynamics and the oscillatory structure of the complete model.

Figure~\ref{fig:phase_space_p1} corresponds to the case $p=1$. The full system (solid curves) features fast oscillations superimposed on a slow drift, while the averaged system (dashed curves) follows a smooth trajectory that illustrates the long-term behaviour of the full dynamics. We observe that the full system approaches the average system over time. The three-dimensional figure also shows that the deviations of the averaged system from the full system become less pronounced as $H$ decreases.

Figure~\ref{fig:phase_space_p2} shows the behaviour for $p=2$. Here, the influence of the $p$--dependent terms becomes more evident. The trajectories display a stronger drift toward the boundary of the allowed region.

\begin{figure}[H]
    \centering
    \includegraphics[width=\textwidth]{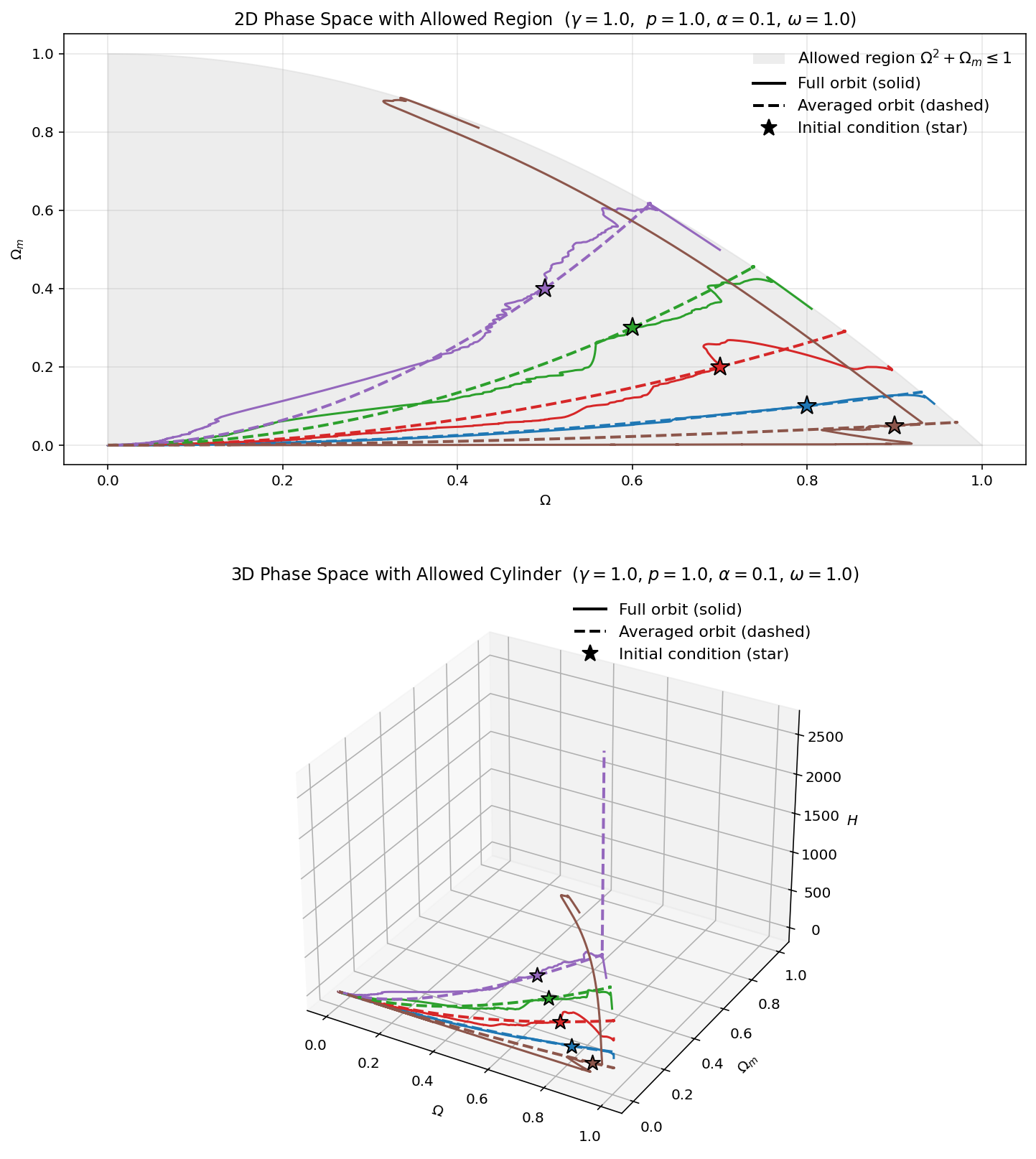}
    \caption{
    Phase--space trajectories for the parameter set 
    $(\gamma, p, \alpha, \omega)=(1.0, 1.0, 0.1, 1.0)$.
    \textbf{Top panel:} Two--dimensional phase space $(\Omega,\Omega_{m})$,
    showing the allowed region $\Omega^{2}+\Omega_{m}\leq 1$.
    Solid curves denote the full system, dashed curves the averaged system,
    and star markers indicate the initial conditions listed in 
    Table~\ref{tab:ICs}.
    \textbf{Bottom panel:} Three--dimensional trajectories in 
    $(\Omega,\Omega_{m},H)$, illustrating the slow evolution of $H$ and the
    close agreement between the full and averaged dynamics.
    }
    \label{fig:phase_space_p1}
\end{figure}

\begin{figure}[H]
    \centering
    \includegraphics[width=\textwidth]{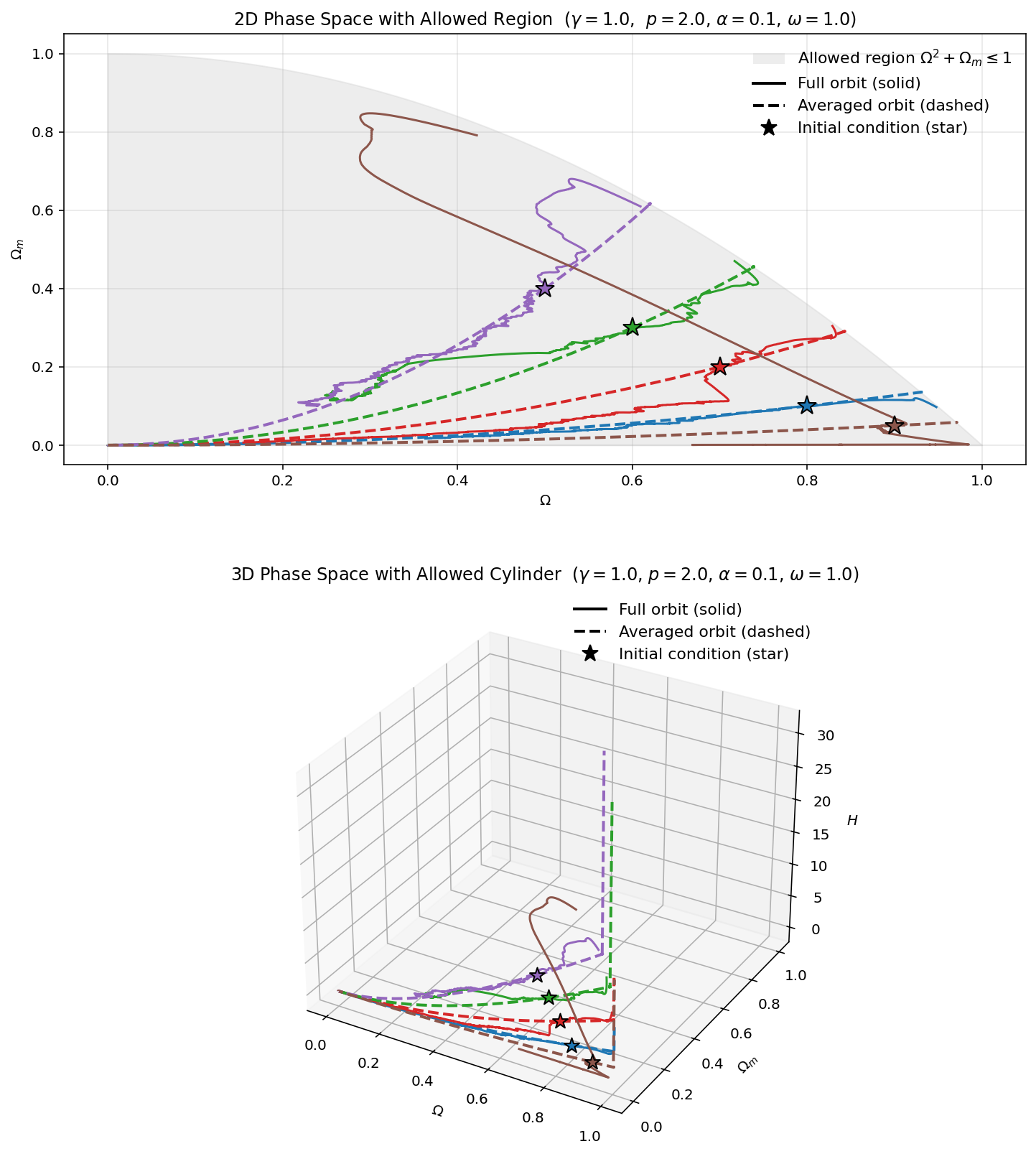}
    \caption{
    Phase--space trajectories for the case $p=2$ with parameters
    $(\gamma, p,\alpha,\omega)=(1.0, 2.0,0.1,1.0)$.
    \textbf{Top panel:} Two--dimensional phase space $(\Omega,\Omega_{m})$,
    showing how the increased value of $p$ modifies the flow structure.
    \textbf{Bottom panel:} Three--dimensional trajectories in 
    $(\Omega,\Omega_{m},H)$, highlighting the stronger drift toward the
    boundary of the allowed region and the continued agreement between the
    full and averaged systems.
    }
    \label{fig:phase_space_p2}
\end{figure}

Nevertheless, the averaged system continues to reproduce the slow component of the full dynamics, both in the two-dimensional projection and in the three-dimensional figure.

Taken together, Figures~\ref{fig:phase_space_p1}--\ref{fig:phase_space_p2} demonstrate that the averaged system provides a good late-time approximation to the slow dynamics of the full model for the two choices of the $p$ parameter and for various initial conditions. 

We have considered $p<3\gamma$ and $\gamma=1$, hence, both integrations show the late-time behavior
\begin{equation}
     \Omega_{G} = 1 \implies 3H^2=\rho_{G}(a) = \kappa^{2} a^{-p} \implies 3{\dot{a}}^2= \kappa^{2} a^{2-p}. \label{eq.96}
\end{equation}
Starting from \eqref{eq.96}
we obtain
\begin{equation}
\dot{a}^2 = \frac{\kappa^2}{3} a^{2-p}
\quad \implies \quad
\dot{a} = \pm \frac{\kappa}{\sqrt{3}} a^{1 - \tfrac{p}{2}}.
\end{equation}
Hence
\begin{equation}
\frac{da}{dt} = \pm \frac{\kappa}{\sqrt{3}} a^{1 - \tfrac{p}{2}}
\quad \implies \quad
\frac{da}{a^{1 - \tfrac{p}{2}}} = \pm \frac{\kappa}{\sqrt{3}} dt.
\end{equation}
Integrating,
\begin{equation}
\int a^{\tfrac{p}{2} - 1} \, da = \pm \frac{\kappa}{\sqrt{3}} \int dt
\quad \implies \quad
\frac{2}{p} a^{p/2} = \pm \frac{\kappa}{\sqrt{3}} t + C.
\end{equation}
Therefore,
\begin{equation}
a^{p/2}(t) = \frac{p}{2} \left( \pm \frac{\kappa}{\sqrt{3}} t + C \right)
\quad \implies \quad
a(t) = \left[ \frac{p}{2} \left( \pm \frac{\kappa}{\sqrt{3}} t + C \right) \right]^{2/p}, \quad p \neq 0.
\end{equation}

For the special case $p=0$,
\begin{equation}
\dot{a} = \pm \frac{\kappa}{\sqrt{3}} a
\quad \implies \quad
a(t) = A \exp\!\left(\pm \frac{\kappa}{\sqrt{3}} t\right).
\end{equation}

The scale factor solutions allow us to compute the deceleration parameter,
\begin{equation}
q \equiv - \frac{a \ddot{a}}{\dot{a}^2}.
\end{equation}
From
\begin{equation}
 \dot{a} = \pm \frac{\kappa}{\sqrt{3}} a^{1 - \tfrac{p}{2}} \implies   \ddot{a} = \pm \frac{\kappa}{\sqrt{3}} \left(1 - \frac{p}{2}\right) a^{- \tfrac{p}{2}} \dot{a}.
\end{equation}
Substituting into the definition of $q$ gives
\begin{equation}
q = - \left(1 - \frac{p}{2}\right).
\end{equation}

Hence the sign of $q$ depends directly on the parameter $p$:
\begin{equation}
q < 0 \quad \Leftrightarrow \quad p < 2,
\end{equation}
which corresponds to accelerated expansion, while
\begin{equation}
q > 0 \quad \Leftrightarrow \quad p > 2,
\end{equation}
indicates decelerated expansion. 

The critical case $p=2$ yields $q=0$, i.e.\ a coasting solution with linear growth of the scale factor.
In general, values of $p$ below $2$ lead to accelerated expansion driven by the geometric term. In particular, the case $p=1$ corresponds to $q=-\tfrac{1}{2}$, showing acceleration, while $p=2$ corresponds to the transition to non-accelerating expansion.

\section{Exact Scalar Field solutions for a minimally coupled scalar field in General Relativity}
\label{Asect:2}
In this Section, we switch our attention to cosmic evolution scenarios in the context of scalar-tensor gravity admitting exact solutions. The previous Sections focused on approximate methods for solving the evolution equations under suitable assumptions; this Section completes the analysis by considering models that admit analytic solutions. We start by outlining the development of the quadrature approach described in the literature \cite{Chimento:1995da}. 
 Unlike in the previous part of the paper, we consider here a cosmological model governed by a minimally coupled scalar field, $\phi$, and a perfect-fluid matter component with a barotropic index, $\gamma$. The evolution equations reduce to 
\begin{subequations}
\begin{align}
\ddot{\phi} + 3H\,\dot{\phi} + \frac{dV}{d\phi} &= 0, \label{eq:KG_GR} \\
\dot{\rho}_m + 3\gamma H\,\rho_m &= 0, \label{eq:matter_GR} \\
3H^2 &= \frac{1}{2}\dot{\phi}^2 + V(\phi) + \rho_m + G(a),\label{eq:Fried_GR}
\end{align}
\end{subequations}
with $G(a)$ being the geometric term, defined as:
\begin{equation}
    G(a) = \begin{cases}
         -\frac{3k}{a^2},k=0, \pm 1, & \quad \text{FLRW metric,} \\
         \frac{\sigma_0^2}{a^6}, & \quad \text{Bianchi I}.
    \end{cases}
\end{equation}
Differentiating Eq.~\eqref{eq:Fried_GR} and using Eqs.~\eqref{eq:KG_GR}–\eqref{eq:matter_GR}, we obtain an auxiliary identity
\begin{equation}
\frac{1}{2}\dot{\phi}^2 - V(\phi) + (\gamma - 1)\rho_m - G(a) - \frac{1}{3}a G'(a)= -2\dot{H} - 3H^2. \label{eq:aux_energy_GR}
\end{equation}

The Raychaudhuri equation then reads
\begin{equation}
\dot{H} = -\frac{1}{2} \left( \gamma \rho_m + \dot{\phi}^2 -\frac{1}{3}a G'(a)\right), \label{eq:Ray_GR}
\end{equation}
highlighting the contribution of kinetic and matter energy to the expansion deceleration.

\subsection{Standard cosmology in \texorpdfstring{$\tau$}{τ}-time parameterization}
\label{Asect:2.A}
To simplify the evolution system, we adopt a new time variable $ \tau $ defined by
\begin{equation}
dt = a^3(\tau)\, d\tau,
\quad \implies \quad
\frac{d}{dt} = \frac{1}{a^3(\tau)}\, \frac{d}{d\tau}.
\end{equation}

From this parametrization, the time function becomes
\begin{equation}
t(\tau) = \int a^3(\tau)\, d\tau + C_1. \label{eq:t_tau_GR}
\end{equation}

\subsubsection{Dynamical equations in \texorpdfstring{$\tau$}{τ}-time}
\label{Asect:2.A.1}

The dynamical equations in \texorpdfstring{$\tau$}{τ}-time are the following
\begin{itemize}
  \item \textbf{Klein--Gordon equation}
  \begin{equation}
  \phi''(\tau) + a^6(\tau)\, \frac{dV}{d\phi} = 0. \label{eq:KG_GR_tau}
  \end{equation}

  \item \textbf{Matter continuity equation}
  \begin{equation}
  \rho_m'(\tau) + 3\gamma\, \frac{a'(\tau)}{a(\tau)}\, \rho_m(\tau) = 0 \implies \rho_m(\tau) = \rho_{m,0} \left( \frac{a(\tau)}{a_0} \right)^{-3\gamma}. \label{eq:matter_GR_tau}
  \end{equation}

  \item \textbf{Friedmann equation}
  \begin{equation}
  3H^2(\tau) = \frac{1}{2 a^6(\tau)} \left( \phi'(\tau) \right)^2 + V(\phi(\tau)) + \rho_{m,0} \left( \frac{a(\tau)}{a_0} \right)^{-3\gamma}+G\left(a(\tau)\right). \label{eq:Friedmann_GR_tau}
  \end{equation}

  \item \textbf{Raychaudhuri equation}
  \begin{equation}
  H'(\tau) = -\frac{1}{2} \left[ \gamma \rho_{m,0} \left( \frac{a(\tau)}{a_0} \right)^{3(1-\gamma)} + \frac{ \left( \phi'(\tau) \right)^2 }{ a^3(\tau) } -\frac{1}{3}a(\tau)G'(a(\tau))\right]. \label{eq:Ray_GR_tau}
  \end{equation}
\end{itemize}

\subsubsection{Scalar field parametrization and quadrature formalism}
\label{Asect:2.A.2}
Defining the potential via \begin{equation}
\label{quad}
V[\phi(a)]=\frac{\mathcal{F}(a)}{a^6},
\end{equation} we obtain the scalar field energy
\begin{equation}
\mathcal{E}(\tau) = \frac{1}{a^6(\tau)} \left[ \frac{1}{2} \left( \phi'(\tau) \right)^2 + \mathcal{F}(a(\tau)) \right]. \label{eq:energy_GR_tau}
\end{equation}

Integrating in standard GR, this energy splits into
\begin{equation}
\frac{1}{2}\dot{\phi}^2 + V(\phi) = \frac{C}{a^6} + \frac{6}{a^6} \int \frac{\mathcal{F}(a)}{a} \, da. \label{eq:phi_integral_mu1}
\end{equation}

The effective Friedmann function becomes:
\begin{equation}
3H^2\equiv \mathcal{G}(a) = G(a) + \frac{\rho_{m,0}}{a^{3\gamma}} + \Lambda + \underbrace{\frac{C}{a^6} + \frac{6}{a^6} \int \frac{\mathcal{F}(a)}{a} \, da}_{\text{Scalar field energy density}}. \label{eq:G_mu1}
\end{equation}

The kinetic contribution reads
\begin{equation}
\frac{1}{2}{\dot{\phi}}^2\equiv \mathcal{L}(a) = -\frac{\mathcal{F}(a)}{a^6} + \underbrace{\frac{C}{a^6} + \frac{6}{a^6} \int \frac{\mathcal{F}(a)}{a} \, da}_{\text{Scalar field energy density}}.\label{eq:L_mu1}
\end{equation}

The corresponding quadratures are
\begin{subequations}
	\label{quad-chimento}
	\begin{align}
	&t-t_0=\sqrt{3}\int \frac{d a}{a\sqrt{\mathcal{G}(a)}},\label{quad1}\\
	&\phi-\phi_0=\sqrt{6}\int\sqrt{\frac{\mathcal{L}(a)}{\mathcal{G}(a)}}\frac{d a}{a}.\label{quad2}
	\end{align}
\end{subequations}

Eqs. \eqref{quad}, \eqref{quad1} and \eqref{quad2}, allows to obtain a solution $a(t), \phi(t)$ for a potential $V(\phi),$ the functions $\mathcal{F}(a), \chi(a)$, the parameters $\rho_{m,0}, \gamma, k, \sigma_0, $ and $\Lambda$, and the integration constant $C$.
The solution \eqref{quad1} and \eqref{quad2} has three integration constants, $t_0,\phi_0, C$; however, it does not provide the general solution for a given potential, but it provides a special solution for each member of a family of potentials.   
 
To obtain physical solutions it is required that $\mathcal{G}(a)$ and $\mathcal{L}(a)$ are nonnegative on an interval $(a_1,a_2),$ $0\leq a_1\leq a_2\leq \infty$ \cite{Chimento:1995da}:
\begin{itemize}
	\item sufficient conditions: $a>0, \dot a \neq 0, \dot{\phi}\neq 0$.
	\item to relax the above conditions, we need to study the behavior of the scale factor, determined by the function $\mathcal{G}(a)$:
	\begin{itemize}
		\item $a(t)$ is monotonic ($\dot a\neq 0$) when $\mathcal{G}(a)>0$,
		\item it is bounded if $\mathcal{G}(a_0)=0$ for $0<a_0< \infty$,
		\item a growing monotonic solution starts from $a=0$ and expands without bound
		\item a solution is {\em{singular}} when $a(t)$ vanishes at a finite time, that is, when the integral
		\begin{equation}
		T(a_1,a_2)=\int_{a_1}^{a_2} \frac{d a}{a \sqrt{\mathcal{G}(a)}}
		\end{equation}
		converges in the limit $a_1\rightarrow 0$.
	\end{itemize}  
\end{itemize}

\subsubsection{Inflationary Parameters in Minimal Coupling FLRW}
\label{Asect:2.A.3}
Let us discuss the slow-roll approximation. Two parameters quantify the length of the accelerated expansion and the smallness of the inflaton's kinetic energy relative to its potential, which dominates the scalar field energy in slow-roll inflation. For this purpose, the slow-roll parameters are introduced, as expressed with the time derivatives of the Hubble function and the scalar field \cite{Copeland:1993jj, Lidsey:1995np, Copeland:1998fz}: 
\begin{align}
    \epsilon_H (a) = -\frac{\dot{H}}{H^2}\equiv \frac{\dot{\phi}^2}{2V + \dot{\phi}^2}, \quad \eta_H(a) = -\frac{\ddot{\phi}}{\dot{\phi}H}.
\end{align}
Written in terms of the functions $\mathcal{L}(a)$ and $\mathcal{F}(a)$, they are
\begin{subequations}
\begin{align}
\epsilon_H(a) &= 1 - \frac{\mathcal{F}(a)}{\mathcal{F}(a) + a^6 \mathcal{L}(a)} = 1 - \frac{\mathcal{F}(a)}{C + 6 \int \frac{\mathcal{F}(a)}{a} \, da}, \label{eq:Epsilon_def} \\
\eta_H(a) & = \frac{a \mathcal{L}'(a)}{2 \mathcal{L}(a)} = -\frac{6 \left(-2 \mathcal{F}(a) + 6 \int \frac{\mathcal{F}(a)}{a} \, da + C \right) + a \mathcal{F}'(a)}{2 \left( -\mathcal{F}(a) + 6 \int \frac{\mathcal{F}(a)}{a} \, da + C \right)}, \label{eq:eta_def}
\end{align}
\end{subequations}
where $\epsilon_H$ measures the relative contribution of the kinetic field energy to its total energy density and $\eta_H$ measures the ratio of the field acceleration relative to the friction term. The slow-roll approximation is valid when $|\epsilon_H|\ll 1$ and $|\eta_H|\ll 1$. These requirements impose constraints on the form of the potential and the value of the initial conditions \cite{Chimento:1995da}.

On the other hand, the slow-roll conditions, assuming $M^2_\text{Pl} = 1$, can also be expressed using the scalar field potential (and its derivatives)
\begin{align}
   & \epsilon_V(a) = \frac{1}{2}\left(\frac{V_{,\phi}}{V}\right)^2, \quad \eta_V(a) = \frac{V_{,\phi\phi}}{V}.
\end{align}
Or, equivalently
\begin{subequations}
    \begin{align}
    & \epsilon_V(a) = \frac{1}{2}\frac{\mathcal{G}(a)}{6\mathcal{L}(a)}\left(\frac{a \mathcal{F}'(a) -6 \mathcal{F}(a)}{\mathcal{F}(a)}\right)^2,\\
    & \eta_V(a) = \frac{a^2 \mathcal{G}(a) \mathcal{F}''(a)}{6 \mathcal{F}(a) \mathcal{L}(a)}+\frac{a^2 \mathcal{F}'(a)
   \mathcal{G}'(a)}{12 \mathcal{F}(a) \mathcal{L}(a)}-\frac{a^2 \mathcal{G}(a) \mathcal{F}'(a) \mathcal{L}'(a)}{12
   \mathcal{F}(a) \mathcal{L}(a)^2}\nonumber\\
   &\quad\quad\quad-\frac{11 a \mathcal{G}(a) \mathcal{F}'(a)}{6 \mathcal{F}(a)
   \mathcal{L}(a)}-\frac{a \mathcal{G}'(a)}{2 \mathcal{L}(a)}+\frac{a \mathcal{G}(a) \mathcal{L}'(a)}{2
   \mathcal{L}(a)^2}+\frac{6 \mathcal{G}(a)}{\mathcal{L}(a)}.
\end{align}
\end{subequations}
Up to the first order, they can be related to the slow-roll parameters \eqref{eq:Epsilon_def} and \eqref{eq:eta_def} in the following way:
\begin{align}
    \epsilon_H(a)\approx\epsilon_V(a),\quad \eta_H(a)\approx \eta_V(a) - \epsilon_V(a).
\end{align}
The usual formula gives the number of e-foldings:
\begin{equation}
    N =\int_t^{t_{\text{end}}} H \text{d}t = \int_{\phi_{\text{end}}}^\phi\frac{d\phi}{\sqrt{2\epsilon_V(a(\phi))}},
\end{equation}
where $\phi$ stands for the value of the inflaton at the beginning of the inflation, and $\phi_{\text{end}}$ is the value of the field at the end. We will always assume that $\phi_{\text{end}} \ll \phi$. Moreover, during the slow-roll inflation, the following identity holds
\begin{equation}
    \frac{6\mathcal{L}(a)}{\mathcal{G}(a)} =- \frac{a\mathcal{F}'(a) - 6\mathcal{F}(a)}{\mathcal{F}(a)}.
\end{equation}
One can use this fact to write (notice the change of the integration limits):
\begin{equation}
    N = \int_{\phi}^{\phi_\text{end}} \sqrt{\frac{G(a(\phi))}{6L(a(\phi))}}d\phi \equiv \int_{\phi}^{\phi_\text{end}} \frac{1}{a(\phi)\frac{d\phi}{da}}d\phi.
\end{equation}
This formula can be used to compute the number of e-foldings in terms of the scalar field.

\subsection{Dynamics and inflationary parameters from given \texorpdfstring{$\mathcal{F}(a)=B a^s$}{\mathcal{F}(a)=\mathcal{F}(a)=B a^s}}
\label{Asect:2.B}
We consider a minimally coupled scalar field with $\mathcal{F}(a)=B a^s$
\begin{equation}
\mathcal{L}(a) = \frac{C + B\left( \frac{6}{s} - 1 \right) a^s}{a^6}, \quad
\mathcal{G}(a) = G(a) + \frac{\rho_{m,0}}{a^{3\gamma}} + \Lambda + \frac{C}{a^6} + \frac{6B}{s} a^{s - 6}.
\end{equation}
The cosmic time is given by 
\begin{equation}
t(a) = \int \frac{da}{a \sqrt{\mathcal{G}(a)}} = \int \frac{\sqrt{3}da}{a \sqrt{ G(a) + \frac{\rho_{m,0}}{a^{3\gamma}} + \Lambda + \frac{C}{a^6} + \frac{6B}{s} a^{s - 6} }}.
\end{equation}
The scalar field is then
\begin{align}
\phi(a) &= \phi_0 \pm \int \sqrt{ \frac{6 \mathcal{L}(a)}{a^2 \mathcal{G}(a)} } \, da \nonumber\\
&= \phi_0 \pm \int \sqrt{ \frac{6\left[ C + B\left( \frac{6}{s} - 1 \right) a^s \right] }{ a^8 \left[ G(a) + \frac{\rho_{m,0}}{a^{3\gamma}} + \Lambda + \frac{C}{a^6} + \frac{6B}{s} a^{s - 6} \right] } } \, da.
\end{align}

For the exact integral, the Hubble function is given by 
\begin{equation}\label{eq:hubble}
H(a) = \frac{1}{\sqrt{3}} \sqrt{ G(a) + \frac{\rho_{m,0}}{a^{3\gamma}} + \Lambda + \frac{C}{a^6} + \frac{6B}{s} a^{s - 6} }.
\end{equation}
The deceleration parameter is given by 
\begin{equation}
q(a) = -1 - \frac{a}{2\mathcal{G}(a)} \frac{d\mathcal{G}}{da}, \quad
\frac{d\mathcal{G}}{da} = -\frac{3\gamma \rho_{m,0}}{a^{3\gamma + 1}} - \frac{6C}{a^7} + \frac{6B(s - 6)}{s} a^{s - 7} + G'(a).
\end{equation}

\subsubsection{Slow-roll parameters}
\label{Asect:2.B.1}
Assuming a potential $ V(a) = B a^{s-6} $, we write
\begin{align}
   \epsilon_H(a) =\frac{Ba^s(6-s) + Cs}{6Ba^s + Cs},\quad  \eta_H(a) = \frac{Ba^s(s-6)^2+6Cs}{2Ba^s(s-6)-2Cs}.
\end{align}
Alternatively, written in terms of the scalar field potential
\begin{subequations}
\begin{align}
   & \epsilon_V(a) =\frac{(s-6)^2 a^{-3 \gamma }
   \left(a^{3 \gamma }
   \left(-\left(6 B a^s+s
   \left(a^6 G(a)+a^6
   \Lambda
   +C\right)\right)\right)-a^6
 \rho_{m, 0}  
   s\right)}{12 \left(B (s-6)
   a^s-C s\right)},\\
& \eta_V(a) = \frac{a^{-3\gamma}}{12(Ba^s(s-6) - Cs)^2}\Bigg(\rho_{0,m} s a^6(-Ba^s (s-6)(s-3(\gamma+2)) + Cs(2s - 3(\gamma+2))) \nonumber \\
& \quad a^{3\gamma}\Big(-12B^2a^{2s}(s-6)^2 + Cs^2\big(2c(s-6) + 2\Lambda a^6(s-3) + 2G(a)a^6(s-3) + a^7 G'(a)\big) \nonumber \\
&\quad-Bsa^s(c(144 +s(s-36)) +\Lambda a^6(s-6)^2 + a^6(s-6)((s-6)G(a) + a G'(a)))\Big)\Bigg).
\end{align}
\end{subequations}
The spectral tilt can be computed using the following definition
\begin{equation}
n_s(a)-1 =-4\epsilon_H(a) +2\eta_H(a) = - 6 \epsilon_V(a) + 2 \eta_V(a), \quad r(a) = 16 \epsilon(a).
\end{equation}
We can write slow-roll parameters in terms of the e-foldings, $N$
\begin{align}
 \epsilon_H(N) =\frac{C s+B (6-s) e^{N s}}{6 B e^{N s}+C s},\quad \eta_H(N) =\frac{B (s-6)^2 e^{N s}+6 C s}{2 B (s-6) e^{N s}-2 C s}.
\end{align}
Therefore, the scalar spectral index and the tensor-to-scalar ratio become
\begin{subequations}
    \begin{align}
       & n_s(N) =1+ \frac{4 \left(B (s-6) e^{N s}-C s\right)}{6 B e^{N s}+C s}+\frac{B (s-6)^2 e^{N
   s}+6 C s}{B (s-6) e^{N s}-C s}, \\
 &  r(N) = \frac{16 \left(C s-B (s-6) e^{Ns}\right)}{6 B e^{N s}+C s}.
    \end{align}
\end{subequations}
It turns out that, for this particular setup with $C\neq 0$, one can obtain an explicit relation between these two parameters
\begin{equation}
    n_s(r) = -\frac{16 (s-6)}{r}-\frac{r}{4}+s-11.
\end{equation}
In Fig. \ref{fig:inflation}, we present different $n_r$-$r$ relations for a range of values of the parameter $s$ with parameters $B$, $C$ fixed to be $B=1$ and $C=1$. We find that, for successful inflation and correct values of the scalar spectral index and the tensor-to-scalar ratio, $s$ must be slightly less than $s=6$. Any value above it yields a negative tensor-to-scalar ratio, providing additional justification for the previous assumption that $0 < s < 6$. 
\begin{figure}[h]
    \centering
    \includegraphics[width=15cm]{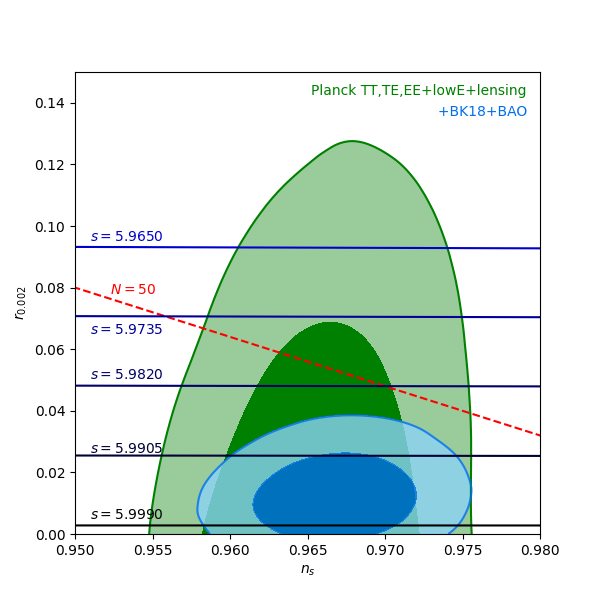}
    \caption{Constraints on the scalar spectral index and the tensor-to-scalar ratio in the $\Lambda$CDM model using the Planck TT, TE, EE+lowE+lensing (green color, with $1\sigma$ and $2\sigma$ confidence regions) \cite{Planck:2018vyg}, and joint constraints with BAO and BICEP2/Keck (blue color, with $1\sigma$ and $2\sigma$ confidence regions) \cite{BICEP2:2018kqh}. Solid lines represent the $n_s$-$r$ relations for different values of the parameter $s$, with the remaining two parameters chosen as $B = 1$, $C=1$. The dashed line represents the values of $n_s$ and $r$ for $N=50$ e-foldings, plotted for different values of the parameter $s$; the line for $ N=60$ overlaps with it entirely. 
The figure was created using publicly available Python code downloadable from \href{http//bicepkeck.org/bk18_2021_release.html}{http//bicepkeck.org/bk18$\_$2021$\_$release.html}.}
    \label{fig:inflation}
\end{figure}

\subsubsection{Explicit dynamical quantities for scalar source model}
\label{Asect:2.C}

Assume late-time regime dominated by $ \mathcal{G}(a) \approx G(a) + \Lambda + \frac{6B}{s} a^{s - 6} $. Then, 
\begin{align}
\mathcal{G}(a) &\sim \begin{cases}
    G_* a^{s - 6}, \quad G_* = \frac{6B}{s}, & \text{if}\; s>6\\
    \Lambda, & \text{if}\; s<6
\end{cases} 
\nonumber\\
&\implies t(a) \approx \begin{cases} \int \frac{\sqrt{3}da}{a \sqrt{G_* a^{s - 6}}}
= \frac{\sqrt{3}}{\sqrt{G_*}} \int a^{-1 - (s - 6)/2} da
= \frac{\sqrt{3}}{\sqrt{G_*}} \cdot \frac{a^{(6 - s)/2}}{(6 - s)/2},& \text{if}\; s>6\\
  \sqrt{\frac{3}{\Lambda}}\ln a, & \text{if}\; s<6
\end{cases} .
\end{align}
From above
\begin{equation}
a(t) \sim 
\begin{cases}\left[ \frac{|6 - s|}{2} \sqrt{\frac{G_*}{3}} \, t \right]^{2/(6 - s)} , & \text{if}\; s>6\\
\exp\left(\sqrt{\frac{\Lambda}{3}} t\right) , & \text{if}\; s<6
\end{cases}.
\end{equation}
If $s > 6$, the scalar field is
\begin{small}
\begin{align}
\phi(a) &= \phi_0 \pm \int \sqrt{ \frac{6 \mathcal{L}(a)}{a^2 \mathcal{G}(a)} } da
= \phi_0 \pm \int \sqrt{ \frac{6 [C + B(6/s - 1)a^s] }{ a^8 \mathcal{G}(a) } } da \nonumber\\
&\approx \phi_0 \pm \frac{1}{\sqrt{G_*}} \int \sqrt{\frac{a^{-s-2}(6Ba^{s}(6-s) + 6Cs)}{s}} da \nonumber\\
&= \phi_0 \mp \frac{2 \sqrt{6} \left(\sqrt{C s a^{-s}-B(s -6)}-\sqrt{B} \sqrt{s-6} \arctan
 \left(\frac{\sqrt{Cs a^{-s}-B( s-6)}}{\sqrt{B}
   \sqrt{s-6}}\right)\right)}{\sqrt{G_*} s^{3/2}}.
\end{align}
\end{small}
On the other hand, when $s < 6$, one gets
\begin{small}
\begin{align}
\phi(a) &= \phi_0 \pm \int \sqrt{ \frac{6 \mathcal{L}(a)}{a^2 \mathcal{G}(a)} } da
\approx \phi_0 \pm \int \sqrt{ \frac{6 [C + B(6/s - 1)a^s] }{ a^8 \Lambda } } da \nonumber\\
& = \phi_0 \pm \frac{12 \sqrt{C s-B (s-6) a^s}-2 \sqrt{c} s^{3/2} \,
   _2F_1\left(\frac{1}{2},-\frac{3}{s};\frac{s-3}{s};\frac{a^s B (s-6)}{c
   s}\right)}{a^3 \sqrt{\Lambda } (s-6) \sqrt{s}}.
\end{align}
\end{small}

From Eq \eqref{eq:hubble}, we have
\begin{equation}
H(t) \sim 
\begin{cases} \frac{1}{\sqrt{3}} \sqrt{\frac{6B}{s} a(t)^{s - 6} }, & \text{if}\; s>6\\
\sqrt{\frac{\Lambda}{3}}, & \text{if}\; s<6
\end{cases}, \quad \text{as}\; a\rightarrow \infty. 
\end{equation}

\subsubsection{Scalar field parametrization and quadrature formalism in the early-time limit}
The early-time limit will be dominated by the term $\mathcal{G}(a)\approx\beta a^{-p}$, where $\beta$ represents some constant, and $p = \max\{3\gamma, 6, 6-s\}$. Then, one gets
\begin{equation}
    t(a) \approx \frac{2 \sqrt{3} a^{p/2}}{\sqrt{\beta } p} \implies a(t) \approx 12^{-1/p} \left(\sqrt{\beta } p \right)^{2/p}t^{2/p}.
\end{equation}
The scalar field is 
\begin{align}
    \phi(a) \approx 
    \begin{cases}\frac{12 a^{\frac{p}{2}-3} \left(\sqrt{c} s^{3/2} \,
   _2F_1\left(\frac{1}{2},\frac{p-6}{2 s};\frac{p-6}{2 s}+1;\frac{a^s B (s-6)}{c
   s}\right)+(p-6) \sqrt{c s-B (s-6) a^s}\right)}{\sqrt{\beta } (p-6) \sqrt{s}
   (p+s-6)}, & \text{if}\; p\neq 6\\
   12(\sqrt{\beta } s^{3/2})^{-1} \left(\sqrt{c s-B (s-6) a^s}-\sqrt{c} \sqrt{s} \tanh
   ^{-1}\left(\frac{\sqrt{c s-B (s-6) a^s}}{\sqrt{c}
   \sqrt{s}}\right)\right), & \text{if}\; p=6
   \end{cases}.
\end{align}
And finally, the deceleration parameter is
\begin{equation}
    q = -1+\frac{p}{2}.
\end{equation}

\subsection{FLRW models}
\label{Asect:2.D}
Setting the input function as $\mathcal{F}(a)=B a^s$, $G(a) = -\frac{3k}{a^2}  $, $\rho_{r,0}=\rho_m=\Lambda=k=C=0$, that is, using the flat FRW metric, no matter, and only the scalar field, we obtain
\begin{subequations}
	\begin{align}
	&t=t_0+\frac{\sqrt{2s}}{\sqrt{B}(6-s)} a^{3-\frac{s}{2}},\\
	&\phi =\phi_0+\sqrt{6-s}\ln a,\\
	&V[\phi(a)]=B a^{s-6}.
	\end{align}
\end{subequations}
After inversion, we obtain
\begin{subequations}
	\begin{align}
	&a(t)=\left[\frac{\sqrt{B}|6-s|}{\sqrt{2s}}(t-t_0)\right]^{\frac{2}{6-s}},\\
	&\phi =\phi_0+\ln\left[\frac{\sqrt{B|6-s|}}{\sqrt{2s}}(t-t_0)\right]^{\frac{2}{\sqrt{6-s}}},\\
	&V[\phi(a)]=B e^{-\sqrt{6-s}(\phi-\phi_0)}.
	\end{align}
\end{subequations}
Consider the parametric solution for the minimal coupling case, where $
\mathcal{F}(a) = B a^s, s \leq 6, C > 0,
$
and we assume all background terms vanish
$\rho_{m,0} = \Lambda = k = 0$.
The solutions for the time quadrature, scalar field, and potential are given by
\begin{subequations}
\begin{align}
t(a) &= \frac{a^3\sqrt{6Ba^s + Cs} \, _2F_1\left(1, \frac{1}{2}+\frac{3}{s}; \frac{s + 3}{s}; -\frac{6 B a^s}{C s}\right)}{\sqrt{3s}C} + t_0, \\
\phi(a) &= \frac{\sqrt{6 - s}}{s} \ln\left(2 \sqrt{6 B a^s + C s} \sqrt{C s - B (s - 6) a^s} - \frac{12 B (s - 6) a^s + C (s - 12) s}{\sqrt{6} \sqrt{6 - s}}\right) \nonumber \\
&\quad -\frac{\sqrt{6}}{s} \ln\left(2 \sqrt{6 B a^s + C s} \sqrt{C s - B (s - 6) a^s} - B (s - 12) a^s + 2 C s \right) \nonumber \\
&\quad + \sqrt{6} \ln a + \phi_0, \\
V(\phi(a)) &= B a^{s - 6}.
\end{align}
\end{subequations}

\subsubsection{Asymptotic analysis of flat FLRW models}
\label{Asect:2.D.1}

\textbf{Early-time limit ($ a \to 0 $).} At early times, the behavior is characterized by
\begin{itemize}
    \item $ V(\phi) \to \infty $ for $ s < 6 $; $ V(\phi) \to B $ for $ s = 6 $.
    \item Scalar field evolves as $ \phi(a) \sim \sqrt{6} \ln a + \phi_0 + \mathcal{O}(1) $, implying $ a \sim e^{\phi / \sqrt{6}} $ as $ \phi \to -\infty $.
    \item The hypergeometric function tends to unity as its argument approaches zero, so the time evolution simplifies to
    \begin{equation}
    t(a) \sim \frac{a^3}{\sqrt{3 C}} + t_0 + \mathcal{O}(a^{3 + s}),
    \end{equation}
    leading to the inverse relation $ a(t) \sim (t - t_0)^{1/3} $.
    \item The inflationary parameters are:
    \begin{align}
\epsilon_H(a) &\to 1, \\
\eta_H(a) &  \to -3.
\end{align}
This regime is non-inflationary; the universe is dominated by scalar kinetic energy ($ \epsilon_H \to 1 $) and quickly varying ($ \eta_H \to -3 $).
\end{itemize}

\textbf{Late-time limit ($ a \to \infty $).} For large-scale factor
\begin{itemize}
    \item $ V(\phi) \to 0 $ for $ s < 6 $; $ V(\phi) \to B $ for $ s = 6 $.
    \item Scalar field grows logarithmically as
    \begin{equation}
    \phi(a) \sim \sqrt{6-s}\ln a + \phi_0 + \mathcal{O}(1),
    \quad \text{so} \quad a \sim e^{\phi /\sqrt{6-s
    }}.
    \end{equation}
    \item The hypergeometric function $_2 F_1(a,b,c,z)$ admits an asymptotic expansion for $ |z| \to -\infty $ according to the formula
    \begin{equation}_2 F_1(a,b,c,z)\approx \frac{\Gamma(b-a)\Gamma(c)}{\Gamma(b)\Gamma(c-a)}(-z)^{-a} + \frac{\Gamma(a-b)\Gamma(c)}{\Gamma(a)\Gamma(c-b)}(-z)^{-b},
    \end{equation}
    yielding
    \begin{equation}
    t(a) \sim A_1 a^{3-\frac{s}{2} } + A_2
    \quad \implies \quad a(t) \sim t^{1 /(3-\frac{s}{2})} \text{ for } s < 6.  
     \end{equation}
\end{itemize}

\paragraph{Discussion.}

The asymptotic behavior of the system separates naturally into two dynamical regimes, summarized in Table~\ref{tab:asymptotics}. 

\begin{table}[h]
\caption{Asymptotic behavior of the system}
\label{tab:asymptotics}
\begin{center}
\renewcommand{\arraystretch}{1.3}
\begin{tabular}{|c|c|c|c|c|}
\hline
Regime & $ V(\phi) $ & $ \phi(a) \sim$ & $ a(t) \sim $ & Interpretation \\
\hline
$ a \to 0 $ & $ B a^{s - 6} \to \infty $ & $ \sqrt{6} \ln a $ & $ t^{1/3} $ & Stiff-matter early universe \\
\hline
$ a \to \infty $ & $ B a^{s - 6} \to 0 $ & $ \sqrt{6-a}\ln a $ & $ t^{1/(3-\frac{s}{2})} $ & Kinetic-dominated scalar decay \\
\hline
\end{tabular}
\end{center}
\end{table}

In the limit $a \to 0$, the potential scales as $V(\phi) \sim B a^{\,s-6}$, which diverges for $s < 6$. In this regime the scalar field evolves logarithmically,
\begin{equation}
\phi(a) \sim \sqrt{6}\,\ln a,
\end{equation}
and the scale factor follows the expansion law
\begin{equation}
a(t) \sim t^{1/3}.
\end{equation}

At late times, $a \to \infty$, the potential becomes negligible, $V(\phi) \to 0$, and the dynamics transition into a kinetic-dominated regime. The field retains a logarithmic dependence on the scale factor, though with a modified coefficient, while the expansion evolves as
\begin{equation}
a(t) \sim t^{1/(3 - s/2)}.
\end{equation}
The deceleration parameter for this model is $q = 2 - \frac{s}{2}$, which will correspond to accelerated expansion  for $q > 4$. The current value of the parameter of about $q_0\approx -\frac{1}{2}$ \cite{Camarena:2020} corresponds to $s = 5$. 


\subsection{Bianchi I models}
\label{Asect:2.E}
Now, let's consider $\mathcal{F}(a)=B a^s, s\leq 6$, $\rho_{r,0}=\rho_m=\Lambda= C=0$, and the Bianchi I metric, then
\begin{subequations}
\begin{align}
&t=\frac{a^3 \sqrt{6 B a^s+s \sigma_0^2} \, _2F_1\left(1,\frac{1}{2}+\frac{3}{s};\frac{s+3}{s};-\frac{6 a^s B}{s \text{$\sigma $}^2_0}\right)}{\sqrt{3} \sqrt{s} \sigma_0^2}+t_0,\\
&\phi= \phi_0+\frac{2 \sqrt{6-s} \ln \left(6 B a^{s/2}+\sqrt{6} \sqrt{B} \sqrt{6 B
		a^s+s \sigma_0^2}\right)}{s},\\
&V[\phi(a)]=B a^{s-6}.
\end{align}
\end{subequations}

Let us also consider the case with non-vanishing cosmological constant and $\mathcal{F}(a)=B a^s, s=3$, $\rho_{r,0}=\rho_m=C=0$,
\begin{small}
\begin{subequations}
	\begin{align}
	&t=\frac{\ln \left(a^3 \Lambda +\sqrt{\Lambda } \sqrt{a^6 \Lambda +2 a^3 B+\sigma_0^2}+B\right)}{\sqrt{3} \sqrt{\Lambda }}+t_0,\\
	&\phi =\phi_0-\frac{2 i \sqrt{\frac{2}{3}} \sqrt{B} \sqrt{\sqrt{B^2-\Lambda  \sigma_0^2}+B} \mathcal{F}\left(i \sinh ^{-1}\left(\frac{a^{3/2}
			\sqrt{\Lambda }}{\sqrt{B+\sqrt{B^2-\Lambda  \sigma_0^2}}}\right)|\frac{2 B \left(B+\sqrt{B^2-\Lambda  \sigma_0^2}\right)-\Lambda 
			\sigma_0^2}{\Lambda  \sigma_0^2}\right)}{\sqrt{\Lambda } \sqrt{\sigma_0^2}},\\
	&V[\phi(a)]=B a^{s-6}.
	\end{align}
\end{subequations}
\end{small}
where ${\mathcal{F}}$ gives the elliptic integral of the first kind.

Inverting the above expressions, we get  
\begin{small}
\begin{subequations}
\begin{align}
&a(t)=\frac{\sqrt[3]{e^{\sqrt{3} \sqrt{\Lambda } (t_0-t)} \left(\left(B-e^{\sqrt{3} \sqrt{\Lambda } (t-t_0)}\right)^2-\Lambda  \sigma^2_0\right)}}{\sqrt[3]{2} \sqrt[3]{\Lambda }},\\
&\phi(a(t))=\phi_0-\frac{2 i \sqrt{\frac{2}{3}} \sqrt{B} \sqrt{\sqrt{B^2-\Lambda  \sigma_0^2}+B} \mathcal{F}\left(i \sinh ^{-1}\left(\frac{a(t)^{3/2}
		\sqrt{\Lambda }}{\sqrt{B+\sqrt{B^2-\Lambda  \sigma_0^2}}}\right)|\frac{2 B \left(B+\sqrt{B^2-\Lambda  \sigma_0^2}\right)-\Lambda 
		\sigma_0^2}{\Lambda  \sigma_0^2}\right)}{\sqrt{\Lambda } \sqrt{\sigma_0^2}},\\
& V(\phi)=-\frac{B \Lambda  \text{ns}\left(\frac{i \sqrt{\frac{3}{2}} \sqrt{\Lambda } \sqrt{\sigma_0^2} (\phi -\phi_0)}{2 \sqrt{B}
		\sqrt{B+\sqrt{B^2-\Lambda  \sigma_0^2}}}|\frac{2 B \left(B+\sqrt{B^2-\Lambda  \sigma_0^2}\right)-\Lambda  \sigma_0^2}{\Lambda 
		\sigma_0^2}\right)^2}{\sqrt{B^2-\Lambda  \sigma_0^2}+B},\\
&H(t)=\frac{\sqrt{\Lambda } \left(\left(B^2-\Lambda  \sigma_0^2\right) e^{2 \sqrt{3} \sqrt{\Lambda } (t_0-t)}-1\right)}{\sqrt{3} \left(\left(\Lambda 
	\sigma_0^2-B^2\right) e^{2 \sqrt{3} \sqrt{\Lambda } (t_0-t)}+2 B e^{\sqrt{3} \sqrt{\Lambda } (t_0-t)}-1\right)},
\end{align}
\end{subequations}
\end{small}
where $\text{ns}$ denotes the Jacobi elliptic function $\text{ns}(u|m)$.
This solution gives an accelerating expansion for a positive cosmological constant since for $\Lambda>0$, $q\rightarrow -1$ as $t\rightarrow\infty$, 
$H\rightarrow \sqrt{\frac{\Lambda}{3}}$ for $\Lambda>0$ as  $ t\rightarrow\infty$ 

\subsection{Asymptotic analysis of Bianchi I models}
\label{Asect:2.F}
We consider two analytic solutions for Bianchi I cosmologies with scalar source $ \mathcal{F}(a) = B a^s $, where $ s \leq 6 $, and analyze their asymptotic behavior in the limits $ a \to 0 $ and $ a \to \infty $.

\vspace{0.5em}
\subsubsection{Case 1: vanishing background terms (\texorpdfstring{$\rho_{m,0} = \Lambda = C = 0 $}{ρm0 = Λ = C = 0})}
\label{Asect:2.F.1}
\textbf{Early-time limit ($ a \to 0 $).}
\begin{itemize}
    \item $ V(\phi) \to \infty $ for $ s < 6 $; $ V(\phi) \to B $ for $ s = 6 $.
    \item $ \phi(a) \sim \text{const}$.
    \item $ t(a) \sim \frac{a^3}{\sqrt{3 s \sigma_0^2}} + t_0 $, so $ a(t) \sim (t - t_0)^{1/3} $.
\end{itemize}

\textbf{Late-time limit ($ a \to \infty $).}
\begin{itemize}
    \item $ V(\phi) \to 0 $ for $ s < 6 $; $ V(\phi) \to B $ for $ s = 6 $.
    \item $ \phi(a) \sim \sqrt{6 - s} \ln a + \phi_0 $, so $ a \sim e^{(\phi - \phi_0)/\sqrt{6 - s}} $.
    \item $ t(a) \sim \sigma_0^{\frac{1}{2} + \frac{3}{s}} a^{3-\frac{s}{2}} $, hence $ a(t) \sim t^{1/(3 -\frac{s}{2})} $.
\end{itemize}

\vspace{0.5em}
\subsubsection{Case 2: nonzero cosmological constant (\texorpdfstring{$ \Lambda > 0, s = 3 $)}{Λ > 0, s = 3 }}
\label{Asect:2.F.2}

\textbf{Late-time limit ($ t \to \infty $).}
\begin{itemize}
    \item $ a(t) \to \infty $, with exponential growth $ a(t) \sim e^{\sqrt{\Lambda/3} \, t} $.
    \item $ H(t) \to \sqrt{\Lambda/3} $. 
    \item $ q(t) \to -1 $, confirming accelerated expansion.
    \item $ V(\phi) \to 0 $, consistent with scalar field dilution.
\end{itemize}

\textbf{Early-time limit ($ t \to t_0 $).}
\begin{itemize}
    \item In general $ a(t) \to 0 $, with $ V(\phi) \to \text{const} $ if $t_0 = \frac{\ln( B-\sqrt{\Lambda}\sigma_0)}{\sqrt {3\Lambda}}$.
    \item $\phi(t) \rightarrow \infty$ if $t_0 = \frac{\ln( B-\sqrt{\Lambda}\sigma_0)}{\sqrt {3\Lambda}}$.
    \item Dynamics resemble stiff-matter or Kasner-like behavior.
\end{itemize}

\paragraph{Discussion.} The asymptotic structure of the solutions summarized in Table~\ref{Asect:2.F.3} reveals two distinct dynamical regimes for Bianchi~I cosmologies with the cosmological constant included. 

\begin{table}[h]
\caption{Asymptotic behavior}
\label{Asect:2.F.3}
\begin{center}
\renewcommand{\arraystretch}{1.3}
\begin{tabular}{|c|c|c|c|c|}
\hline
Regime & $ V(\phi) $ & $ \phi(a) $ & $ a(t) $ & Interpretation \\
\hline
$ a \to 0 $ & $ \to \infty $ & $ \ln a $ & $ t^{1/3} $ & Anisotropic stiff-like early universe \\
\hline
$ a \to \infty $ & $ \to 0 $ & $ \ln a $ & $ e^{\sqrt{\Lambda/3} \, t} $ & Accelerated isotropic expansion \\
\hline
\end{tabular}
\end{center}
\end{table}

In the limit $a \to 0$, the potential diverges, $V(\phi)\to\infty$, and the scalar field evolves logarithmically as $\phi(a)\sim \ln a$. The corresponding expansion law, $a(t)\sim t^{1/3}$, reflects a stiff-like early universe. This behavior is consistent with the well-known dominance of shear and kinetic terms near the initial singularity in anisotropic models.

At late times, $a \to \infty$, the potential decays to zero and the scalar field again exhibits a logarithmic dependence on the scale factor. In this regime the scale factor grows exponentially, $a(t)\sim e^{\sqrt{\Lambda/3}\, t}$, signaling the onset of accelerated expansion.

Overall, these results confirm that Bianchi~I models with scalar sources naturally evolve from an anisotropic, stiff-like early phase toward a late-time isotropic state, with the cosmological constant driving the system toward accelerated expansion.

\subsubsection{Inflationary parameters in Bianchi I cosmology}
\label{Asect:2.F.4}
We analyze the slow-roll parameters for the Bianchi I solution with scalar source $ \mathcal{F}(a) = B a^s$, assuming minimal coupling and vanishing background terms
\begin{equation}
\rho_{m,0} = \Lambda = C = 0.
\end{equation}
The scalar field kinetic term is given by
\begin{equation}
\mathcal{L}(a) = \frac{6 - s}{s} B a^{s - 6},
\end{equation}
and the Friedmann function includes anisotropic shear:
\begin{equation}
\mathcal{G}(a) = \frac{6 B}{s a^{6 - s}} + \frac{\sigma_0^2}{a^6}.
\end{equation}

The inflationary parameters are:
\begin{subequations}
\begin{align}
\epsilon_H(a) &= 1- \frac{s}{6}, \label{eq:epsilon_bianchi} \\
\eta_H(a) &=  \frac{1}{2}(s-6). \label{eq:eta_bianchi}
\end{align}
\end{subequations}
Obviously, with vanishing constant $C$, the slow-roll parameters do not depend on the geometric term $G(a)$, since it does enter the definition of the $\mathcal{L}(a)$ function, and are the same for both the FLRW model, as well as for the Bianchi I models. The fact that the slow-roll parameters are constant means that the inflationary period does not come to a halt (one expects the period of an accelerated expansion to stop when $\epsilon_H(a) \approx 1$).

\section{Exact Scalar field solutions for nonminimally coupled scalar field}
\label{Asect:3}

Let's consider a nonminimally coupled scalar field cosmology given by \cite{Gonzalez:2007ht, Fadragas:2014mra}:
\begin{subequations}
	\label{non_min}
	\begin{align}
	&\dot\phi\left[\ddot\phi+3 H \dot \phi +\frac{d V(\phi)}{d\phi}\right]=\frac{1}{2}(4-3\gamma)\rho_m H a\frac{d\ln\chi}{d a},\label{Fried1b}\\
	&\dot{\rho_m}+3\gamma H\rho_m=-\frac{1}{2}(4-3\gamma)\rho_m H a\frac{d\ln\chi}{d a},\label{consb}\\
	& 3H^2=\frac{1}{2}\dot\phi^2+V(\phi)+\Lambda+G(a) + \rho_m.\label{Fried2b}
	\end{align}
\end{subequations}
Integrating \eqref{consb} it follows
\begin{equation}
\rho_m=\frac{\rho_{m,0}}{a^{3\gamma}}\chi(a)^{-2+\frac{3\gamma}{2}}.
\end{equation}

Introducing the new time variable $d\tau=a^{-3}dt$, and parameterizing the potential according to \eqref{quad} the equation \eqref{Fried1b} can be written as
\begin{subequations}
\begin{align}
&\frac{d}{d\tau}\left[\frac{1}{2}\left(\frac{d\phi}{d\tau}\right)^2+\mathcal{F}(a)\right]=6\frac{\mathcal{F}}{a}\frac{d a}{d\tau}+\frac{1}{2}\rho_{m,0}(4-3\gamma)a^{3(2-\gamma)} \chi^{-3+\frac{3\gamma}{2}}\frac{d \chi}{d \tau},\\
&\frac{1}{2}\left(\frac{d\phi}{d\tau}\right)^2+\mathcal{F}(a)=C + 6 \int\frac{\mathcal{F}}{a} d a +\frac{1}{2}\rho_{m,0}(4-3\gamma)\int a^{3(2-\gamma)} \chi^{-3+\frac{3\gamma}{2}} d \chi \nonumber\\
&=C-\rho_{m,0}
a^{3(2-\gamma)} \chi^{-2+\frac{3\gamma}{2}} + 6 \int\frac{\mathcal{F}}{a} d a  +3\rho_{m,0}(2-\gamma) \int \chi^{-2+\frac{3\gamma}{2}} a^{3(2-\gamma)}\frac{d a}{a}.
\end{align}
\end{subequations}
so that, finally,  
\begin{equation}
    \frac{1}{2}\dot{\phi}^2+\frac{\mathcal{F}(a)}{a^6}=\frac{C}{a^6}-\rho_{m,0}a^{-3\gamma}\chi^{-2+\frac{3\gamma}{2}} +\frac{3}{a^6}\int\left[2\mathcal{F}+(2-\gamma)\rho_{m,0}a^{3(2-\gamma)} \chi(a)^{-2+\frac{3\gamma}{2}} \right]\frac{d a}{a},
\end{equation}
and the kinetic part,
$\frac{1}{2}\dot{\phi}^2=\mathcal{L}(a)$,
becomes
\begin{align}\label{def_L_nonminimal}
\mathcal{L}(a)=-\frac{\rho_{m,0}}{a^{3\gamma}}\chi(a)^{-2+\frac{3\gamma}{2}}-\frac{\mathcal{F}(a)}{a^6}+\frac{C}{a^6}+\frac{3}{a^6}\int\left[2\mathcal{F}+(2-\gamma)\rho_{m,0}a^{3(2-\gamma)} \chi(a)^{-2+\frac{3\gamma}{2}} \right]\frac{d a}{a}.
\end{align}
Moreover, in analogy to the previous part of the paper, we introduce:
\begin{align}\label{def_G_nonminimal}
&3H^2=\mathcal{G}(a)=G(a)+\Lambda+\frac{C}{a^6}+\frac{3}{a^6}\int\left[2\mathcal{F}+(2-\gamma)\rho_{m,0}a^{3(2-\gamma)} \chi(a)^{-2+\frac{3\gamma}{2}} \right]\frac{d a}{a}. 
\end{align}

As before, we obtain the quadratures \eqref{quad-chimento}
but with the definitions of $\mathcal{L}(a)$ and $\mathcal{G}(a)$ given by \eqref{def_L_nonminimal} and \eqref{def_G_nonminimal}.  $t_0$ and $\phi_0$ are the other two integration constants. As compared to the previous case, the theory is now enriched by one function $\chi$, describing the non-minimal coupling between the scalar field and matter sector.

Identically to the previous case, to obtain physical solutions it is required that $\mathcal{G}(a)$ and $\mathcal{L}(a)$ are nonnegative on an interval $(a_1,a_2),$ $0\leq a_1\leq a_2\leq \infty.$ \cite{Chimento:1995da}

Setting $\chi(a)\equiv 1$ and taking integration by parts, the equations \eqref{quad-chimento} (the minimally coupled case) are recovered.

\subsection{Flat FLRW metric}
\label{Asect:3.A}
Selecting $\mathcal{F}(a) = B a^s$ and $\chi(a) = \chi_0 a^r$, and imposing the parameter choices $C = k = \Lambda = 0$, with $r = \frac{2(s - 2)}{3\gamma - 4} + 2$, we obtain
\begin{equation}
\Delta t = -\frac{2 \sqrt{s} \chi_0 a^{3 - s/2}}{(s - 6) \sqrt{2B \chi_0^2 - (\gamma - 2) \rho_{m,0} \chi_0^{3\gamma/2}}},
\end{equation}
\begin{equation}
\Delta\phi = \frac{\ln(a) \sqrt{B(s - 6) \chi_0^2 + \rho_{m,0} \chi_0^{3\gamma/2}(3\gamma + s - 6)}}{\sqrt{\frac{1}{2}(\gamma - 2) \rho_{m,0} \chi_0^{3\gamma/2} - B \chi_0^2}}.
\end{equation}

Applying inverse function properties yields
\begin{subequations}
\begin{align}
a(t) &= 2^{\frac{2}{s - 6}} \left(-\frac{(s - 6)\, \Delta t\, \sqrt{2B \chi_0^2 - (\gamma - 2) \rho_{m,0} \chi_0^{3\gamma/2}}}{\sqrt{s}\, \chi_0}\right)^{-\frac{2}{s - 6}}, \\
\phi(t) &= \frac{\sqrt{2B(s - 6) \chi_0^2 + 2\rho_{m,0} \chi_0^{3\gamma/2}(3\gamma + s - 6)}\, \ln(a(t))}{\sqrt{(\gamma - 2)\rho_{m,0} \chi_0^{3\gamma/2} - 2B \chi_0^2}} + \phi_0,
\end{align}
\begin{equation}
V(\phi) = B \exp\left(
\frac{(s - 6)\, \Delta\phi\, \sqrt{\frac{1}{2}(\gamma - 2) \rho_{m,0} \chi_0^{3\gamma/2} - B \chi_0^2}}
{\sqrt{B(s - 6) \chi_0^2 + \rho_{m,0} \chi_0^{3\gamma/2}(3\gamma + s - 6)}}
\right).
\end{equation}
\end{subequations}

This solution is non-accelerating. However, for $ \Lambda \neq 0 $, we find:
\begin{small}
\begin{equation}
\Delta t = \frac{2\sqrt{3}}{\sqrt{\Lambda}(6 - s)} \sinh^{-1} \left( 
\frac{\sqrt{\Lambda s} a^{3 - s/2}}{\sqrt{6B - 3(\gamma - 2)\rho_{m,0} \chi_0^{3\gamma/2 - 2}}} 
\right),
\end{equation}
\begin{equation}
\Delta\phi = K(a)\, \csc^{-1} \left( 
\frac{\sqrt{\Lambda s} \chi_0 a^{3 - s/2}}{\sqrt{3(\gamma - 2)\rho_{m,0} \chi_0^{3\gamma/2} - 6B \chi_0^2}} 
\right),
\end{equation}
\end{small}

where
\begin{small}
\begin{equation}
K(a) = \frac{2\sqrt{2\Lambda s}\chi_0 a^{\frac{3\gamma}{2} + \frac{2s}{3\gamma - 4} + 3} \sqrt{B(s - 6)\chi_0^2 + \rho_{m,0} \chi_0^{3\gamma/2}(3\gamma + s - 6)} \sqrt{\frac{a^{s - 6}(6B \chi_0^2 - 3(\gamma - 2)\rho_{m,0} \chi_0^{3\gamma/2})}{\Lambda s \chi_0^2} + 1}}{(s - 6) \sqrt{(\gamma - 2)\rho_{m,0} \chi_0^{3\gamma/2} - 2B \chi_0^2} \sqrt{a^{3\gamma(\frac{s}{3\gamma - 4} + 1)}(3(\gamma - 2)\rho_{m,0} \chi_0^{3\gamma/2} - 6B \chi_0^2) - \Lambda s \chi_0^2 a^{3\gamma + \frac{4s}{3\gamma - 4} + 6}}}.
\end{equation}
\end{small}
Using inverse function identities again, we obtain
\begin{equation}
a(t) = \left(-\frac{B' \sinh\left(\frac{\sqrt{\Lambda}(s - 6)\Delta t}{2\sqrt{3}}\right)}{\sqrt{s\Lambda}}\right)^{-\frac{2}{s - 6}},
\end{equation}
\begin{equation}
H(t) = -\frac{\sqrt{\Lambda}}{\sqrt{3}} \coth\left( \frac{\sqrt{\Lambda}(s - 6)\Delta t}{2\sqrt{3}} \right),
\end{equation}
\begin{equation}
q(t) = -1 + \frac{6 - s}{\cosh\left( \frac{\sqrt{\Lambda}(s - 6)(t - t_0)}{\sqrt{3}} \right) + 1},
\end{equation}
where $ B' = \sqrt{6B - 3(\gamma - 2)\rho_{m,0} \chi_0^{3\gamma/2 - 2}} $.

\subsubsection{Slow-roll parameters in flat FLRW cosmology with interaction}
\label{Asect:3.A.1}
Assuming $k=\Lambda=C=0$, we consider the interacting scalar field model in a flat FLRW background. The slow-roll parameters are defined as
\begin{small}
\begin{subequations}
\begin{align}
\epsilon_H(a) &= 1 + \frac{\mathcal{F}(a)}{
\rho_{m,0} a^{6 - 3\gamma} \chi(a)^{\frac{3\gamma}{2} - 2}
- 3 \int \left( \frac{2 \mathcal{F}(a)}{a} - (\gamma - 2) \rho_{m,0} a^{5 - 3\gamma} \chi(a)^{\frac{3\gamma}{2} - 2} \right) da
}, \\
\eta_H(a) &= \frac{a \mathcal{F}'(a)}{
2 \left[
\rho_{m,0} a^{6 - 3\gamma} \chi(a)^{\frac{3\gamma}{2} - 2}
- 3 \int \left( \frac{2 \mathcal{F}(a)}{a} - (\gamma - 2) \rho_{m,0} a^{5 - 3\gamma} \chi(a)^{\frac{3\gamma}{2} - 2} \right) da
+ \mathcal{F}(a)
\right]
} \nonumber \\
&\quad - \frac{3 a^{3\gamma} \mathcal{F}(a) \chi(a)^2}{
a^6 \rho_{m,0} \chi(a)^{3\gamma/2}
- a^{3\gamma} \chi(a)^2 \left[
3 \int \left( \frac{2 \mathcal{F}(a)}{a} - (\gamma - 2) \rho_{m,0} a^{5 - 3\gamma} \chi(a)^{\frac{3\gamma}{2} - 2} \right) da
- \mathcal{F}(a)
\right]
} \nonumber \\
&\quad + \frac{a^7 (3\gamma - 4) \rho_{m,0} \chi'(a)}{
4 a^6 \rho_{m,0} \chi(a)
- 4 a^{3\gamma} \chi(a)^{3 - \frac{3\gamma}{2}} \left[
3 \int \left( \frac{2 \mathcal{F}(a)}{a} - (\gamma - 2) \rho_{m,0} a^{5 - 3\gamma} \chi(a)^{\frac{3\gamma}{2} - 2} \right) da
- \mathcal{F}(a)
\right]
} - 3.
\end{align}
\end{subequations}
\end{small}
The integral $\mathcal{I}(a) =\int \left( \frac{2 \mathcal{F}(a)}{a} - (\gamma - 2) \rho_{m,0} a^{5 - 3\gamma} \chi(a)^{\frac{3\gamma}{2} - 2} \right) da$ becomes
\begin{equation}
\mathcal{I}(a) = \frac{a^s \left(2 B \chi_0^2-(\gamma -2)\rho_{m,0} \chi_0^{3 \gamma /2}\right)}{s \chi_0^2}.
\end{equation}

Then, the slow-roll parameters are
\begin{subequations}
\begin{align}
\epsilon_H(a) &= \frac{B (s-6) \chi_0^2+\rho_{m,0} \chi_0^{3 \gamma /2}
   (3 \gamma +s-6)}{\rho_{m,0} \chi_0^{3 \gamma /2} (3 \gamma
   +s-6)-6 B \chi_0^2}, \\
\eta_H(a) &= \frac{1}{2}(s-6).
\end{align}
\end{subequations}
For this particular set of parameters, the slow-roll parameters are constant, meaning that the inflation (for an appropriate combination of the remaining constants) will not stop at any moment.

\subsection{Bianchi I metric}
\label{Asect:3.B}
Selecting $ \mathcal{F}(a) = B a^s $ and $ \chi(a) = \chi_0 a^r $, and setting parameters $ C = k = \Lambda = 0 $, with $ r = \frac{2(s - 2)}{3\gamma - 4} + 2 $, the resulting solution is
\begin{small}
\begin{subequations}
\begin{align}
\Delta t &= \frac{a^3 \sqrt{ \frac{3 a^s (2B \chi_0^2 - (\gamma - 2)\rho_{m,0} \chi_0^{3\gamma/2}) }{s \chi_0^2} + \sigma_0^2 }\, {}_2F_1\left(1, \frac{1}{2} + \frac{3}{s}; \frac{s + 3}{s}; \frac{3 a^s((\gamma - 2)\rho_{m,0} \chi_0^{3\gamma/2} - 2B \chi_0^2)}{s \sigma_0^2 \chi_0^2} \right)}{\sqrt{3}\sigma_0^2}, \\
\Delta\phi &= K_2(a)\, \csc^{-1} \left( \frac{\sqrt{s\sigma_0^2} \chi_0 a^{-s/2}}{\sqrt{3(\gamma - 2)\rho_{m,0} \chi_0^{3\gamma/2} - 6B \chi_0^2}} \right),
\end{align}
\end{subequations}
\end{small}
with
\begin{small}
\begin{equation}
K_2(a) = \frac{2 \sqrt{2} a^{\frac{1}{2}(3\gamma + \frac{4s}{3\gamma - 4})} \sqrt{B(s - 6)\chi_0^2 + \rho_{m,0} \chi_0^{3\gamma/2}(3\gamma + s - 6)} \sqrt{a^s(6B \chi_0^2 - 3(\gamma - 2)\rho_{m,0} \chi_0^{3\gamma/2}) + s\sigma_0^2 \chi_0^2}}{s \sqrt{2B \chi_0^2 - (\gamma - 2)\rho_{m,0} \chi_0^{3\gamma/2}} \sqrt{a^{3\gamma(\frac{s}{3\gamma - 4} + 1)}(6B \chi_0^2 - 3(\gamma - 2)\rho_{m,0} \chi_0^{3\gamma/2}) + s\sigma_0^2 \chi_0^2 a^{3\gamma + \frac{4s}{3\gamma - 4}}}}.
\end{equation}
\end{small}
Because the inflationary parameters do not depend on the geometric term $G(a)$ (since it does not enter the function $\mathcal{L}(a)$, which is the key component of the slow-roll parameters), the Bianchi models will yield precisely the same values of them as the FRLW model. Obviously, the scalar spectral index and the tensor-to-scalar ratio will also be equal. This, however, does not mean that the dependence of these quantities on the cosmic time or the scalar field will be the same.

\section{Exact Scalar Field solutions in Brane Cosmology}
\label{Asect:4}

In this section, we analyze scalar field trapping in brane-world cosmologies with FLRW and Bianchi~I embeddings. The scalar field $\phi$ evolves under a potential $V(\phi)$ and interacts with a matter fluid of energy density $\rho_m$, governed by an equation of state $p_m = (\gamma - 1)\rho_m$. The total energy density $\rho_T$ includes contributions from the scalar field, matter, and a dark radiation term $\rho_{r,0}$, with brane corrections introduced via the parameter $\sigma_b = 1/(2\lambda_b)$.

\subsection{Scalar field trapping in brane universes: FLRW and Bianchi I cosmologies}
\label{Asect:4.A}
Let us consider a brane universe given either by an FLRW or a Bianchi I model, where a scalar field potential and a matter fluid energy density $\rho_m$ with equation of state $p_m=(\gamma-1)\rho_m$ are trapped.
The equations field is given by 
\begin{subequations}
	\label{branes_syst}
	\begin{align}
	&\ddot{\phi}+3H\dot{\phi}+V'(\phi)=0,\label{Fried1c}\\
	&\dot{\rho_m}+3\gamma H\rho_m=0\label{cons3},\\
	&3 H^2=\rho_T(1+\sigma_b \rho_T)+3\frac{\rho_{r,0}}{a^4}+\Lambda+G(a),
	\end{align}
\end{subequations}
where 
\begin{equation}
\rho_T=\frac{1}{2}\dot{\phi}^2+V(\phi)+\rho_m,
\end{equation}
$\rho_{r,0}$ is related with the so-called dark radiation term, and $\sigma_b=\frac{1}{2\lambda_b}$, where $\lambda_b$ denotes the brane tension.
Integrating \eqref{cons3} we obtain 
\begin{equation}
\rho_m=\frac{\rho_{m,0}}{a^{3\gamma}},
\end{equation}
Introducing the new time variable $dt=a^3 d\tau$ in \eqref{Fried1c}, and given by \eqref{quad}, we obtain the first integral
\begin{equation}\label{eq2c}
\frac{1}{2}\dot{\phi}^2+V(\phi)=\frac{C}{a^6}+\frac{6}{a^6}\int \frac{\mathcal{F}(a)}{a}da,
\end{equation}
where $C$ is an arbitrary integration constant.
That is, 
\begin{align}
\rho_T=\frac{\rho_{m,0}}{a^{3\gamma}}+\frac{C}{a^6}+\frac{6}{a^6}\int \frac{\mathcal{F}(a)}{a}da,
\end{align}
Hence, as before, the system \eqref{branes_syst}
can be written as
\begin{align}
3H^2=\mathcal{G}(a),\\
\frac{1}{2}\dot{\phi}^2=\mathcal{L}(a),
\end{align}
where 
\begin{subequations}
	\label{L_G_branes}
	\begin{align}
	&\mathcal{G}(a)=\frac{\rho_{m,0}}{a^{3\gamma}}+\frac{C}{a^6}+\frac{6}{a^6}\int \frac{\mathcal{F}(a)}{a}da+3\frac{\rho_{r,0}}{a^4}+\Lambda+G(a) +\sigma_b \left[\frac{\rho_{m,0}}{a^{3\gamma}}+\frac{C}{a^6}+\frac{6}{a^6}\int \frac{\mathcal{F}(a)}{a}da\right]^2,\\
	&\mathcal{L}(a)=-\frac{\mathcal{F}(a)}{a^6}+\frac{C}{a^6}+\frac{6}{a^6}\int \frac{\mathcal{F}(a)}{a}da.
	\end{align}
\end{subequations}
As before, we obtain the quadratures \eqref{quad1} and \eqref{quad2}, 
where $t_0$ and $\phi_0$ are other two integration constants and $\mathcal{G}(a)$ and $\mathcal{L}(a)$ are given by \eqref{L_G_branes}. 

Eqs. \eqref{quad}, \eqref{quad1} and \eqref{quad2}, allows to obtain a solution $a(t), \phi(t)$ for a potential $V(\phi),$ the functions $\mathcal{F}(a), \chi(a)$, the parameters $\rho_{m,0},\rho_{r,0}, \gamma, k, \sigma_0, \sigma_b$ and $\Lambda$, the choice of the primitives involved and the integration constant $C$.
The solution \eqref{quad1} and \eqref{quad2} has three integration constants, $t_0,\phi_0, C$.

As before, to obtain physical solutions it is required that $\mathcal{G}(a)$ and $\mathcal{L}(a)$ are nonnegative on an interval $(a_1,a_2),$ $0\leq a_1\leq a_2\leq \infty.$ \cite{Chimento:1995da}

Now, in the low energy limit $\rho_T\ll\frac{1}{\sigma_b}$, and assuming that no dark-radiation is present, the equations \eqref{quad-chimento} are recovered.

We consider scalar field–dominated cosmologies with brane corrections, governed by a source term
\begin{equation}
\mathcal{F}(a) = B a^s, \quad \gamma = 1, \quad C = \Lambda = \rho_{r,0} = G(a) = 0.
\end{equation}

We introduce
\begin{align}
G_{\text{base}}(a) &= \frac{\rho_{m,0}}{a^3} + \frac{6B}{s a^{6 - s}}, \\
G_{\text{brane}}(a) &= \sigma_b \left( \frac{\rho_{m,0}}{a^3} + \frac{6B}{s a^{6 - s}}\right)^2, \\
\mathcal{G}(a) &= G_{\text{base}}(a) + G_{\text{brane}}(a), \\
\mathcal{L}(a) &= \frac{B(6 - s)}{s} a^{s - 6}.
\end{align}

The scalar field integral becomes
\begin{equation}
\phi(a) = \phi_0 +  \sqrt{6Bs(6-s)}  \int \frac{a^{2-\frac{s}{2}}}{\sqrt{6 B
   a^s+a^3 \rho_{m,0}  s} \sqrt{6 B \sigma_b a^s+a^6 s+a^3
   \rho_{m,0} s \sigma_b}} \, da.
\end{equation}

\subsection{Late-time behavior \texorpdfstring{$(a \to \infty)$}{(a -> ∞)}}
\label{Asect:4.B}
In the asymptotic regime, if $0 < s < 3$, brane corrections dilute and the dominating term is
\begin{equation}
\mathcal{G}(a) \approx \frac{\rho_{m,0}}{a^3}, \quad 
\mathcal{L}(a) = \frac{B(6 - s)}{s a^{6 - s}}.
\end{equation}

Then the evolution proceeds as
\begin{align}\label{eq:late-time-brane}
t(a)-t_0 &\approx \frac{2 a^{3/2}}{ \sqrt{3\rho_{m,0}}}, \quad 
a(t) \approx \frac{\sqrt[3]{3\rho_{m,0}} (t-t_0)^{2/3}}{2^{2/3}}, \quad 
H(t) = \frac{2}{3(t - t_0)}, \quad q(t) = \frac{1}{2}.
\end{align}

The scalar field evolves approximately as
\begin{equation}
\phi(t) \approx \phi_0 + \frac{\sqrt{B} 2^{\frac{5}{2}-\frac{s}{3}} 3^{s/6} \sqrt{\frac{6-s}{s}}
   \rho_{m,0}^{\frac{s}{6}-1} (t-t_0)^{\frac{s}{3}-1}}{s-3}.
\end{equation}
Hence, we have the potential to be reconstructed as 
\begin{equation}
a(\phi) \sim \left(\frac{24B (6-s)}{\rho_{m,0} (s-3)^2 s (\phi - \phi_0)
   ^2}\right)^{\frac{1}{3-s}} \quad \implies \quad
V(\phi) \sim B \left(\frac{24B (6-s)}{\rho_{m,0} (s-3)^2 s (\phi-\phi_0)
   ^2}\right)^{\frac{s-6}{3-s}}.
\end{equation}

\subsection{Early-time behavior \texorpdfstring{$(a \to 0)$}{(a -> 0)}}
\label{Asect:4.C}
As $a \to 0$, brane nonlinearities dominate. Let
\begin{equation}
\mathcal{G}(a) \sim \sigma_b a^{-p}, \quad 
p = \max\{6, 2(6 - s)\}.
\end{equation}

Then, 
\begin{equation}
a(t) \sim 12^{-1/p} \left(p \sqrt{\sigma_b } t\right)^{2/p}, \quad 
H(t) = \frac{2}{pt}, \quad 
q(t) = \frac{p}{2} - 1.
\end{equation}
We obtain the scalar field evolution
\begin{align}
\phi(t)&\sim 2 \sqrt{6} \sqrt{\frac{B (6-s) \left(12^{-1/p} \left(\frac{1}{p^2 \sigma_b 
   }\right)^{-1/p}\right)^{p+s-6}}{s \sigma_b (p+s-6)^2}}(t-t_0)^{\frac{p+s-6}{p}}, \\ 
V(\phi) &\sim B \left(\frac{24B (6-s)}{s \sigma_b  (\phi-\phi_0) ^2
   (p+s-6)^2}\right)^{\frac{s-6}{-p-s+6}}.
\end{align}

\subsection{Example: quadratic source \texorpdfstring{$(\mathcal{F}(a) = B a^2)$}{(\mathcal{F}(a) = B a**2)}}
\label{Asect:4.D}
Fixing $ s = 2 $:
\begin{align}
\mathcal{G}(a) &= \frac{3 B}{a^4}+\frac{\rho_{m,0}}{a^3}+\sigma_b \left(\frac{3
   B}{a^4}+\frac{\rho_{m,0}}{a^3}\right)^2, \\
\mathcal{L}(a) &= \frac{2B}{a^4}.
\end{align}

The late-time dynamics is the same as in \eqref{eq:late-time-brane}. 

Then the evolution proceeds as 
\begin{equation}
\phi(a) \approx \phi_0 -4 \sqrt{3} \sqrt{\frac{B}{a \rho_{m,0} }}, \quad 
V(\phi)\sim \frac{\rho_{m,0} ^4 \phi ^8}{B^3}, 
\end{equation}
with an early time limit
\begin{equation}
\mathcal{G}(a) \sim a^{-8} \implies 
a(t) \sim \sqrt[4]{t}, \quad q = 3, \quad 
\phi(t) \sim \sqrt{t}, \quad V(\phi) \sim \phi^{-2}.
\end{equation}

\subsection{Bianchi I brane cosmology}
\label{Asect:4.E}
For the Bianchi I case, we have the field quadrature
\begin{small}
\begin{equation}
\phi(a) = \phi_0 + \sqrt{6Bs(6-s)}\int \frac{a^{\frac{s+6}{2}-1}}{\sqrt{36
   B^2 \sigma_b a^{2 s}+6 B s a^{s+3} \left(a^3+2\rho_{m,0}
  \sigma_b\right)+a^6 s^2 \left(\rho_{m,0} \left(a^3+\rho_{m,0} \sigma_b\right)+\sigma_{BI} \right)}} \, da.
\end{equation}
\end{small}
Let us notice that the inclusion of the anisotropic shear does not change neither the late-time, nor the early-time limit, since the early-time limit will be dominated by a term $\sigma' a^{-p}$ with $p=\max\{6, 2(6-s)\}$, while the late-time dynamics will be dominated by $\rho_{m, 0} a^{-3}$.

\section{Concluding remarks}
\label{conclusions}
Due to the length of the paper, let us first summarize the content of the presented work. In Section~\ref{sect:0}, we formulated the constrained dynamical system for a scalar field coupled to matter and geometry.
The technical core appeared in Section~\ref{sec:decay_center_perturb}. Section~\ref{sec:decay} derived refined decay estimates, including integrability of the dissipative quantity, uniform derivative bounds, pointwise $O(1/t)$ decay, and a bootstrap to exponential convergence near nondegenerate minima. Section~\ref{sec:center} constructed the local invariant manifold $\mathcal{M}_{\mathrm{loc}}$, performed finite--dimensional reduction on the Friedmann constraint, and recorded its regularity and attraction properties. Section~\ref{sec:perturbations} established persistence of equilibria, decay estimates, and local manifolds under small $C^k$ perturbations of the coupling $\chi(\phi)$ and geometric term $G(a)$.

Section~\ref{sec:discussion_averaging} summarized the averaging framework, established a correspondence between its assumptions and the model under study, and explained how averaging errors entered the global analysis. It included, among others, the Barbalat and LaSalle arguments and spectral gap considerations. We also provided a practical recipe for verifying hypotheses and discussed the limitations. Section~\ref{sec:quad} introduced the quadratic potential.

Section~\ref{Asect:2} developed the analytic framework for minimally coupled scalar fields in FLRW and Bianchi~I geometries. Reformulation in $\tau$--time enabled integration, yielding quadrature expressions for $\phi(a)$ and $t(a)$ and closed-form inflationary observables. Subsequent subsections classified FLRW solutions, analyzed asymptotics, and extended the formalism to anisotropic Bianchi~I backgrounds.

In Section~\ref{Asect:3}, we examined scalar--matter non-minimal coupling in both FLRW and Bianchi~I geometries. Section~\ref{Asect:4} introduced scalar fields in brane--world scenarios, incorporating geometric corrections and modified Friedmann equations. It reviewed scalar trapping, derived early-- and late--time behaviors, introduced analytic source terms, and generalized to anisotropic brane cosmologies.

In this work, we developed a unified analytic and numerical framework for scalar-field cosmologies that integrates averaging theory, asymptotic analysis, and exact quadrature techniques. Starting from the full five--dimensional Einstein--scalar--fluid system, we established global decay estimates for matter, scalar kinetic energy, and geometric contributions under broad structural hypotheses. These results, making use of Barbalat’s lemma and uniform derivative bounds, showed that expanding solutions with nonnegative potentials are driven toward low-energy regimes in which dissipation dominates, and the scalar field approaches either a finite equilibrium or a monotonic runaway configuration.

A central contribution of the paper was the implementation of averaging methods for oscillatory scalar fields. By isolating the fast phase and constructing the associated averaged system, we proved that the slow dynamics approximate the full oscillatory flow with an $\mathcal{O}(H)$ error. This reduction allowed us to study the dissipative mechanisms, the influence of geometric and matter sectors, and the structure of asymptotic regimes. Near nondegenerate minima of the potential, the averaging analysis was complemented by a finite-dimensional reduction yielding a smooth local invariant manifold on which the dynamics converge exponentially to equilibrium.

Beyond asymptotic reductions, we introduced a quadrature-based formulation that expresses the evolution of the scalar field, Hubble rate, and scale factor in closed form for a broad class of potentials and geometric terms. This approach unified FLRW, anisotropic Bianchi~I, and brane--world configurations, enabling analytic computation of inflationary observables. 

Taken together, our analysis provides a unified framework for studying scalar-field cosmologies across oscillatory, dissipative, and asymptotic regimes. The framework involved non-minimal couplings, geometric corrections, and mixed matter sectors, and remained robust under small perturbations of the potential, coupling function, and geometric term. The techniques introduced here naturally extended to multifield models, higher-order curvature corrections, and stochastic or quantum--corrected scalar dynamics.

Looking ahead, several natural extensions follow from the present work. A first direction is the derivation of \emph{quantitative thresholds}. Obtaining explicit perturbation bounds that ensure the persistence of equilibria, invariant manifolds, and spectral gaps under small deformations of the potential $V$, the matter coupling $\chi$, and the geometric term $G$, would enhance the applicability of the persistence results and quantify the admissible size of model variations.

A second avenue concerns \emph{resonant and degenerate regimes}. Extending the analysis to resonant normal forms and higher--order averaging would enable a systematic classification of algebraic decay, metastability, and bifurcation phenomena near degenerate minima, which often arise in models with flat potentials or multiple competing scales.

A third direction involves incorporating \emph{weak inhomogeneities}. Generalizing the framework to weakly inhomogeneous FLRW perturbations would allow one to test the robustness of attractors under spatial variations and mode coupling, thereby linking homogeneous dynamical systems with cosmological perturbation theory.

Finally, there is considerable scope for \emph{numerical and phenomenological applications}. We will aim to constrain the parameters of various theories by connecting theoretical results to observational data. 

In summary, the synthesis of dissipative control, averaging, local reduction, and exact quadrature solutions yields a framework for scalar--field cosmologies under perturbations. It clarifies when potential minima act as attractors, quantifies approach rates, and provides analytic tools for computing inflationary observables in both isotropic and anisotropic backgrounds.

\acknowledgments{Genly Leon and Claudio Michea gratefully acknowledge funding from the Agencia Nacional de Investigación y Desarrollo (ANID) through FONDECYT Grant No. 1240514, ETAPA 2025. Aleksander Kozak is supported by Proyecto No. 3250036, Concurso Fondecyt de Postdoctorado 2025. We also extend our gratitude to the Vicerrectoría de Investigación y Desarrollo Tecnológico (VRIDT) of UCN for the scientific support provided through the Núcleo de Investigación en Geometría Diferencial y Aplicaciones, Resolution VRIDT No.~096/2022, and through the Núcleo de Investigación en Simetrías y la Estructura del Universo (NISEU),  Resolution VRIDT No.~200/2025.}

\bibliography{refs}

\appendix

\end{document}